\documentclass[hidelinks, 12pt]{article}

\usepackage[utf8]{inputenc}
\usepackage[english]{babel}
\usepackage{blindtext}
\usepackage[left=2cm, right=2cm, top=2cm, bottom=2cm]{geometry} 
\usepackage{amsmath,amsfonts}
\usepackage{amssymb}
\usepackage{algorithm}
\usepackage{algcompatible}
\usepackage{algpseudocode}
\usepackage{pifont} 
\usepackage{natbib}
\usepackage{setspace}
\usepackage{graphicx}
\usepackage{caption}
\usepackage{subcaption}
\usepackage{hyperref}
\usepackage{mathrsfs}
\usepackage{amsthm}
\usepackage{tikz}
\usepackage{cancel}
\usepackage{lipsum}
\usepackage{booktabs}
\usepackage{multicol}
\usepackage{multirow}
\usepackage{siunitx}
\usepackage{rotating}
\usepackage{titling}

\theoremstyle{definition}
\newtheorem{definition}{Definition} [section]

\newtheorem{proposition}{Proposition}[section]
\newtheorem{lemma}{Lemma}[section]

\newcommand{\commentsymbol}{//}
\algrenewcommand\algorithmiccomment[1]{\hfill \commentsymbol{} #1}

\newtheorem{remark}{Remark}[section]
\newtheorem{assumption}{Assumption}[section]
\newtheorem{claim}{Claim}[section]

\DeclareMathOperator{\Tr}{Tr} 
\DeclareTextAccent{\myacc}{T1}{4}
\DeclareMathOperator{\vect}{vec}

\allowdisplaybreaks

\def\keywordname{{\bfseries \emph Keywords}}%
\def\keywords#1{\par\addvspace\medskipamount{\rightskip=0pt plus1cm
\def\and{\ifhmode\unskip\nobreak\fi\ $\cdot$
}\noindent\keywordname\enspace\ignorespaces#1\par}}

\title{{Regularised Spectral Estimation for High-Dimensional \\ Point Processes}}
\author{Carla Pinkney$^{1,^*}$, Carolina Eu\'{a}n$^1$, Alex Gibberd$^1$ \& Ali Shojaie$^2$ \vspace{0.3cm} \\
    \small{$^1$ STOR-i Centre for Doctoral Training, Department of Mathematics and Statistics, Lancaster} \\ \small{University, LA1 4YR, UK}\\ $^2$ \small{Department of Biostatistics, University of Washington, Seattle, WA 98105, US} \vspace{0.1cm} \\
\footnotesize{$^*$ Correspondence to: c.pinkney@lancaster.ac.uk} \vspace{0.2cm}
\\}

\begin{document}
\maketitle


\newcommand{\Var}{\mathrm{Var}}
\newcommand{\bI}{\boldsymbol{I}}
\newcommand{\bbI}{\bar{\boldsymbol{I}}}
\newcommand{\bbJ}{\bar{\boldsymbol{d}}}
\newcommand{\ssT}{^{(T)}}
\newcommand{\opnorm}[2]{| \! | \! | #1 | \! | \! |_{{#2}}}


\begin{abstract}
Advances in modern technology have enabled the simultaneous recording of neural spiking activity, which statistically can be represented by a multivariate point process. We characterise the second order structure of this process via the spectral density matrix, a frequency domain equivalent of the covariance matrix. In the context of neuronal analysis, statistics based on the spectral density matrix can be used to infer connectivity in the brain network between individual neurons. However, the high-dimensional nature of spike train data mean that it is often difficult, or at times impossible, to compute these statistics. In this work, we discuss the importance of regularisation-based methods for spectral estimation, and propose novel methodology for use in the point process setting. We establish asymptotic properties for our proposed estimators and evaluate their performance on synthetic data simulated from multivariate Hawkes processes. Finally, we apply our methodology to neuroscience spike train data in order to illustrate its ability to infer brain connectivity. 
\end{abstract}

\keywords{Partial coherence; Multivariate point process; Inverse spectral density matrix; Spike train data.}

\section{Introduction}
{\normalsize
We consider the $p$-dimensional point process $\mathbf{N}(t):=\{N_q(t)\}_{q \in \{1, \dots, p\}}$ whose $q^{th}$ component gives the number of events of type $q$ that have occurred in the time interval $(0,t]$, $t\leq T$. 
As in \cite{bartlett1963spectral}, we denote by $dN_q(t)$ the number of events observed in process $q$ in some small time interval $dt$, whereby $dN_q(t)=N_q(t+dt)-N_q(t).$ The first order properties of $\mathbf{N}(t)$ are characterised via the intensity function $\boldsymbol{\Lambda}(t)\in \mathbb{R}^p$, which we define as $\boldsymbol{\Lambda}(t):=\mathbb{E}\{d\mathbf{N}(t)\}/dt.$
The second order structure of the process at times $t$ and $u$ are determined via the covariance density matrix}
\begin{equation*}
    \boldsymbol{\mu}(t,u) = \frac{\mathbb{E}\{d\mathbf{N}(u) d\mathbf{N}^{\prime}(t)\}}{dt  du} - \boldsymbol{\Lambda}(u)\boldsymbol{\Lambda}^{\prime}(t).
\end{equation*}
\normalsize{
In the event that $\boldsymbol{\Lambda}(t)$ is constant for all $t$, and $\boldsymbol{\mu}(t,u)$ depends only on lag $\tau=u-t$, we refer to process $\mathbf{N}(t)$ as being second-order stationary. In this case, we denote the covariance density matrix simply as $\boldsymbol{\mu}(\tau)$. As in \cite{hawkes1971point}, we note that $\boldsymbol{\mu}(-\tau)=\boldsymbol{\mu}^{\prime}(\tau).$ Going forward, we will refer to a second-order stationary process simply as ``stationary''.

In this work, we characterise the second order structure of $\mathbf{N}(t)$ via the spectral density matrix. Analogous to the covariance matrix in the time domain, the spectral density matrix captures both the within and between dynamics of the multivariate point process. More specifically, it describes the variance in each process or the covariance between processes that is attributable to oscillations in the data at a particular frequency \citep{fiecas2019spectral}. 

The spectral density matrix $\mathbf{S}(\omega)$ of a stationary point process is defined as the Fourier transform of its covariance density matrix \citep{bartlett1963spectral}, namely}
\begin{equation*}
    \label{spec_dens_matrix}
    \mathbf{S}(\omega) = \frac{1}{2\pi}\left\{\textrm{diag}(\boldsymbol{\Lambda})+\int_{-\infty}^{\infty}e^{-i\tau\omega}\boldsymbol{\mu}(\tau) d\tau\right\},
\end{equation*} 
\normalsize{where $\mathbf{S}(\omega) \in \mathbb{C}^{p\times p}$ is a $p \times p$ Hermitian positive definite matrix.  

Interactions between the components of $\mathbf{N}(t)$ can be captured by the inverse spectral density matrix $\boldsymbol{\Theta}^*(\omega):=\mathbf{S}(\omega)^{-1}$. In particular, the partial coherence between processes $N_q(t)$ and $N_r(t)$ at a particular frequency $\omega$ is defined as
}
\begin{equation}
\rho^*_{qr}(\omega)=\frac{|\Theta^*_{qr}(\omega)|^2}{\Theta^*_{qq}(\omega)\Theta^*_{rr}(\omega)}.
\label{eq:partial_co}
\end{equation}

\normalsize{Partial coherence provides a normalised measure on $[0,1]$ of the partial correlation structure between pairs of processes in the frequency domain. In the time series literature, it has been used in a variety of application areas such as medicine, oceanography and climatology \citep{kocsis1999interdependence, koutitonsky2002descriptive, song2020potential}. Here, we use partial coherence to investigate interactions between neuronal point processes.

An estimate of the partial coherence matrix can be achieved by substituting estimates of the inverse spectral density matrix into \eqref{eq:partial_co}. However, in high-dimensional settings, the estimation of $\boldsymbol{\Theta}^*(\omega)$ is a challenging task, and one that remains to be addressed in the point process setting. 

The estimation of high-dimensional (inverse) spectral density matrices for time series data has received significant attention in the statistics literature. For example, \cite{bohm2009shrinkage, fiecas2011generalized} and \cite{fiecas2014data} explore shrinkage estimators for the spectral density matrix, leading to numerically stable and therefore invertible estimates. Comparatively, \cite{jung2015graphical, tugnait2021sparse, dallakyan2022time}  obtained estimates of the inverse spectral density matrix directly by optimising a penalised Whittle likelihood. There, a group lasso penalty is used to capture shared zero patterns across neighbouring frequencies in the spectral domain. 
More recently, \cite{deb2024regularized} used a penalised Whittle likelihood to estimate the inverse spectral density matrix at a single frequency. 
The Whittle approximation has also been studied in the Bayesian framework; see for example \cite{tank2015bayesian}. However, to the best of our knowledge there has been little effort made to investigate this problem in the context of multivariate point processes. 

Here, we estimate the inverse spectrum directly using a penalised Whittle likelihood. We propose two differing penalty functions, each possessing their unique advantages for use in the point process setting. The first is a ridge type penalty and the second is an example of group-lasso penalisation. The latter imposes sparsity on the inverse spectral density matrix, easing interpretation of resulting graphical structures.

The methodology proposed here, novel in the point-process setting, is intended for use in the high-dimensional framework, for the regularised estimation of inverse spectral density matrices. Importantly, we note that while we focus our efforts on inferring neuronal connectivity in the brain network, multivariate point-process data appear in a wide variety of disciplines outside of neuroscience. Therefore, the contributions of this work may also be of interest in other fields including, but not limited to, finance, health and social networks.
}

\section{Preliminaries}

\subsection{The Tapered Fourier Transform for Point Processes}

{\normalsize
A classical non-parametric estimator of the spectral density matrix for a stationary point process is based on the periodogram, that is a (Hermitian) outer-product of the Fourier transform of the data \citep{bartlett1963spectral}. 

In this paper, we propose to utilise a variant of the periodogram, stated in \eqref{eq:TA_estimator}, to estimate the spectral density matrix. However, as the periodogram is based on the Fourier transform, it is first useful to understand some properties of this object in relation to the point process (continuous time) setting. 

As the rate of the point process is non-zero and positive, we will attempt to remove this \emph{mean} prior to taking the Fourier transform. To generalise this procedure, let us consider tapering the data prior to calculating the Fourier transform. 

Consider a set of $k=1,\ldots,m$ taper functions $h_k(z):(0,1]\mapsto\mathbb{R}$, of bounded variation such that 
\begin{equation}
\int|h_k(z+u)-h_k(z)|dz < K |u|\;,\forall\;u\in\mathbb{R}\;,
\label{eq:bounded}
\end{equation}
{for some $K<\infty$.}

\begin{definition}
{Let $h_k(z)$ be a set of tapers satisfying \eqref{eq:bounded}. Then, the tapered Fourier Transform of the process $N_q(t)$ for $t \in (0,T]$ is defined as}
\begin{equation*}
    {d_{k,q}(\omega)=\int_{0}^T h_k(t/T)e^{-i\omega t}dN_q(t).}
\end{equation*}
The mean-corrected tapered Fourier Transform is defined as 
\begin{equation}
    {\bar{d}_{k,q}(\omega) = d_{k,q}(\omega) - d_{k,q}(0)\frac{H_k(T\omega)}{H_k(0)}, 
    \label{eq:mean_corrected_FT}}
\end{equation}   
    {where $H_{k}(\omega)$ denotes the Fourier transform of $h_k(z)$, i.e. $H_k(\omega) = \int h_k(z)e^{-i \omega z} dz$.}
\end{definition}

{An alternative view of the Fourier transform is given by the random set of events induced by the point process. That is, instead of the counting process representation, we can consider the events $E_{k,q} = \{t\:|\:dN_q(t)=1\;,\;t/T\in \mathrm{supp}(h_k)\}$. 
In this case, we can write}
\begin{equation*}
{d_{k,q}(\omega) = \sum_{t\in E_{k,q}} h_k(t/T)e^{-i\omega t}.}
\label{eq:tapered_ft_event}
\end{equation*}
{As in multi-taper estimation for regular time-series, different choices of taper functions may be useful for different tasks, each having their own concentration properties in the Fourier domain.
For the purposes of our discussion, we will focus on the canonical choice of non-overlapping tapers that act to segment the point process into $m$ intervals. } 


\begin{assumption}\label{ass:non-overlap} Let the tapers be defined as 
\[
h_k(t/T)=\begin{cases}
(m/2\pi T)^{1/2} & \mathrm{for}\; t\in ((k-1)T/m, kT/m]\\
0 & \mathrm{otherwise}
\end{cases}
\]
i.e., we assume the tapers are non-overlapping and of equal length.  
\end{assumption}

Assumption \ref{ass:non-overlap} allows us to examine the properties of the mean-corrected Fourier transform in further detail. It also aligns with the experimental motivation for our methodology, in that the tapers can be considered to align with different trials in a neuronal-spiking experiment.  We work under this assumption for the remainder of the paper.

\subsection{The Multi-Taper Periodogram}
\label{sec:MT_periodogram}
We use {Theorem 4.2 from \cite{brillinger1972}} to motivate the multi-taper periodogram,
\begin{equation}
    \hat{\mathbf{S}}(\omega):= \frac{1}{m}\sum_{k=1}^m \bar{\mathbf{d}}_k(\omega) \bar{\mathbf{d}}_k^H(\omega),
    \label{eq:TA_estimator}
\end{equation}
where $\bar{\mathbf{d}}_k(\omega)=(\bar{d}_{k,1}(\omega), \dots, \bar{d}_{k,p}(\omega))^{\prime}.$ The estimator in \eqref{eq:TA_estimator} is an asymptotically ($p$ fixed) $T/m\rightarrow\infty$, $m\rightarrow\infty$ consistent estimator for the spectrum. In the sequel we will be interested in the finite $m$ properties of this estimator. However, we first recall some properties of $\hat{\mathbf{S}}(\omega)$ as $T\rightarrow\infty$.

Under Assumption \ref{ass:non-overlap},  we can show that $\bar{\mathbf{d}}_{k}(\omega)$ is asymptotically Gaussian, and thus the multi-taper periodogram  \eqref{eq:TA_estimator} is distributed as a Complex Wishart distribution, $~\hat{\mathbf{S}}(\omega)\sim m^{-1}\mathcal{W}^C_p(m, \mathbf{S}(\omega))$ with $m$ degrees of freedom and centrality matrix $\mathbf{S}(\omega)$. 
As a normalised measure of correlation between the Fourier coefficients, it is useful to study the spectral coherence which we define as
\begin{equation}
    R^2_{qr}(\omega) = \frac{|S_{qr}(\omega)|^2}{S_{qq}(\omega)S_{rr}(\omega)},
    \label{eq:coherence}
\end{equation}
for $q\ne r$ where $q,r = 1,\dots, p$. As a further corollary of the asymptotic normality of (\ref{eq:mean_corrected_FT}), we have \citep{goodman1963statistical} that if $p<m$, then the estimated coherence {$\hat{R}^2_{qr}(\omega)=|\hat{S}_{qr}(\omega)|^2/\{\hat{S}_{qq}(\omega)\hat{S}_{rr}(\omega)\}$} has the density function
\begin{equation}
    f_{\hat{R}^2_{qr}}(x) = (m-1)(1-\hat{R}^2_{qr})(1-x^2)^{(m-2)} \ _2F_1(m,m;1;\hat{R}^2_{qr} x),
    \label{eq:coherence_density}
\end{equation}
where $_2F_1(\alpha_1,\alpha_2;\beta_1;z)$ denote the hypergeometric function with 2 and 1 parameters $\alpha_1,\alpha_2$ and $\beta_1$ and scalar argument $z$.

The above properties are well known; however, it is of interest to examine the statistical properties of our periodogram in high-dimensional, large $p$ settings. In practice, it is common to encounter such settings where the number of trials $m<p$. In these situations, the Goodman distribution (\ref{eq:coherence_density}) will not apply. As an alternative, we propose to study finite-sample deviation bounds, i.e. for finite $m$, that can control the behaviour (error) of the multi-taper periodogram.  

\begin{proposition}{\textit{Deviation Bound for the Multi-Taper Periodogram.}} 
\label{prop:dev_bound}
Let $\mathbf{\hat{S}}(\omega)$ be the multi-taper periodogram estimator defined in \eqref{eq:TA_estimator}.
For any $q,r = 1,\dots, p,$ the multi-taper periodogram satisfies the tail bound
\begin{equation}
     \mathbb{P}\left\{\left|\hat{S}_{qr}(\omega)-S_{qr}(\omega)\right|\geq\delta\right\} 
        \leq 8\exp\left\{-\frac{m \delta^2}{2^9 5^2 \max_{q}\{S_{qq}(\omega)\}^2}\right\},
        \label{eq:dev_bound}
\end{equation}
for all $\delta \in (0, 80 \max_q\{S_{qq}(\omega)\})$, as $T\rightarrow\infty$.

\end{proposition}

\begin{proof}
    The proof, which is based on results in \cite{bickel2008regularized}, can be found in the Supplementary Material.
\end{proof}


To demonstrate how these properties hold in practice, and the challenges when  $p\ge m$, we give an example for the empirical distribution of, and error incurred, by $\hat{\mathbf{S}}(\omega)$. To illustrate the behaviour, we consider a simple homogeneous Poisson point-process with rate $\lambda=1$ and let $T=1000$, and obtain 1000 Monte-Carlo replicates of the multi-taper spectrum with $m=10$ tapers. The results, depicted in Figure \ref{plt:empirical_properties}, confirm the appropriateness of the asymptotic distribution (\ref{eq:coherence_density}), and how the deviations grow as a function of $\log(p)$. They also highlight the issue of instability of the periodogram as one approaches the high-dimensional $p\approx m$, and $p>m$ setting---it is these scenarios where the methods proposed in this paper will be of most utility.


\begin{figure}
\centering
\begin{subfigure}{.33\textwidth}
  \centering
   \caption{}
   
  \includegraphics[width=\textwidth]{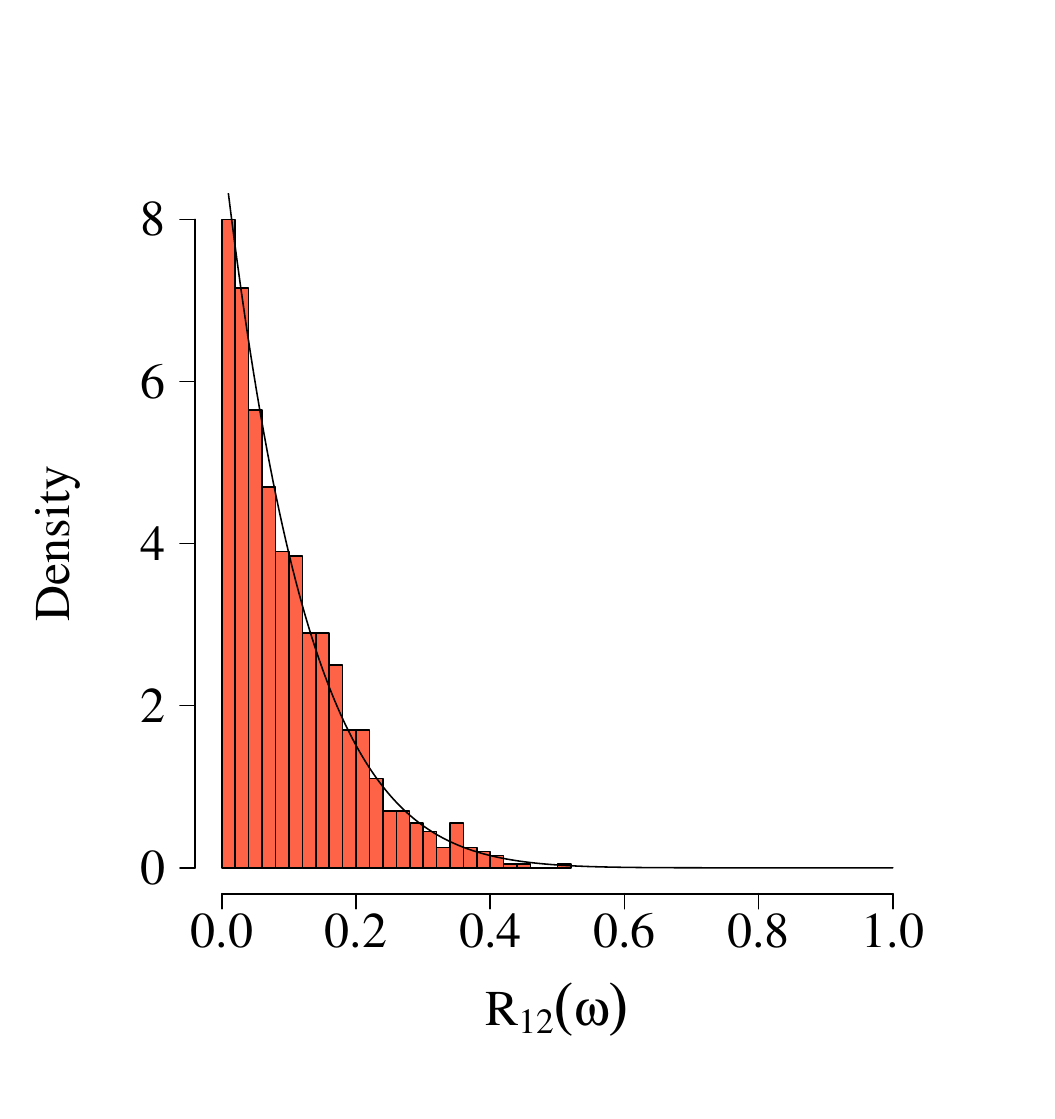}
  \label{fig1a}
\end{subfigure}%
\begin{subfigure}{.33\textwidth}
  \centering
   \caption{}
    
  \includegraphics[width=\textwidth]{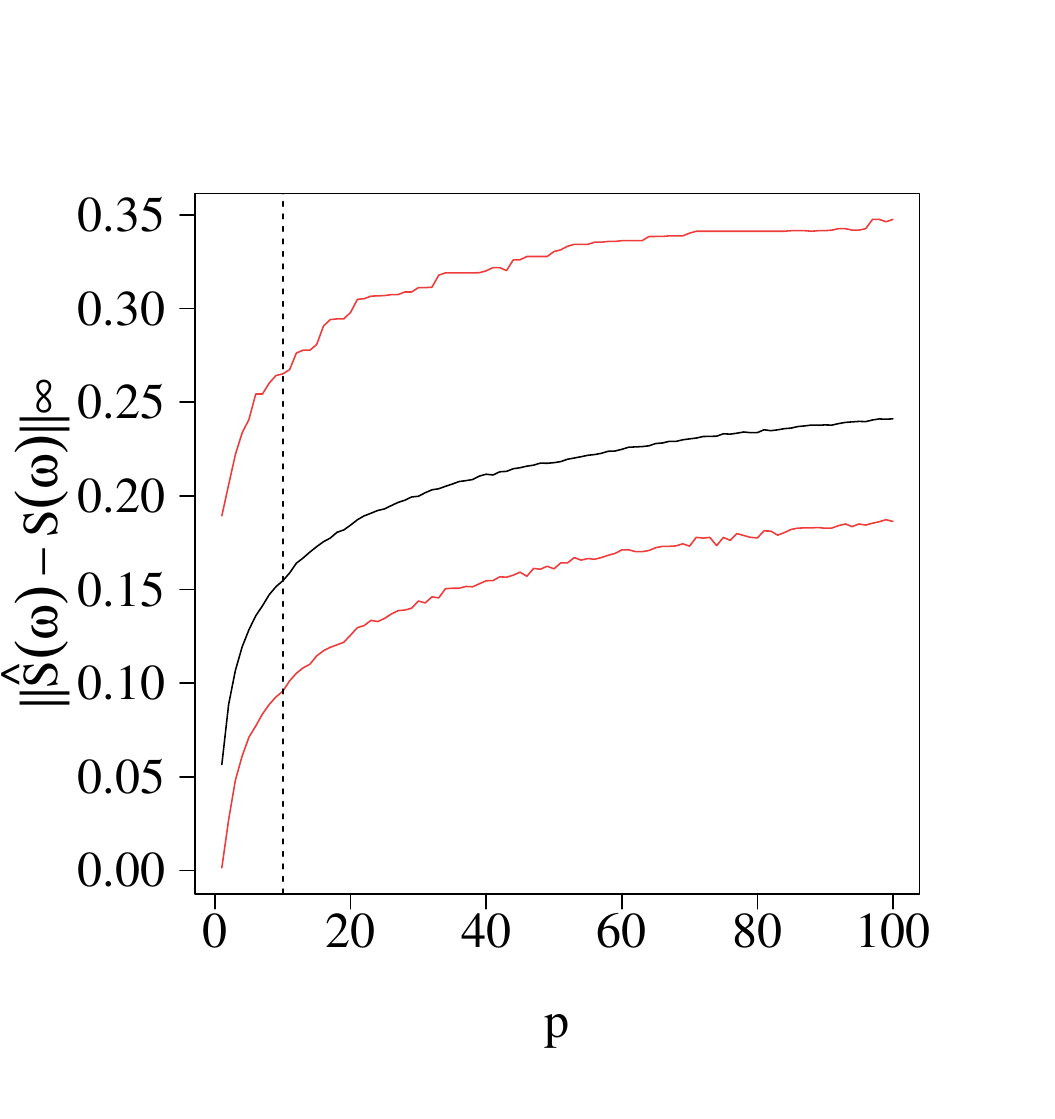}
  \label{fig1b}
\end{subfigure}%
\begin{subfigure}{.33\textwidth}
  \centering
  \caption{}
  
  \includegraphics[width=\textwidth]{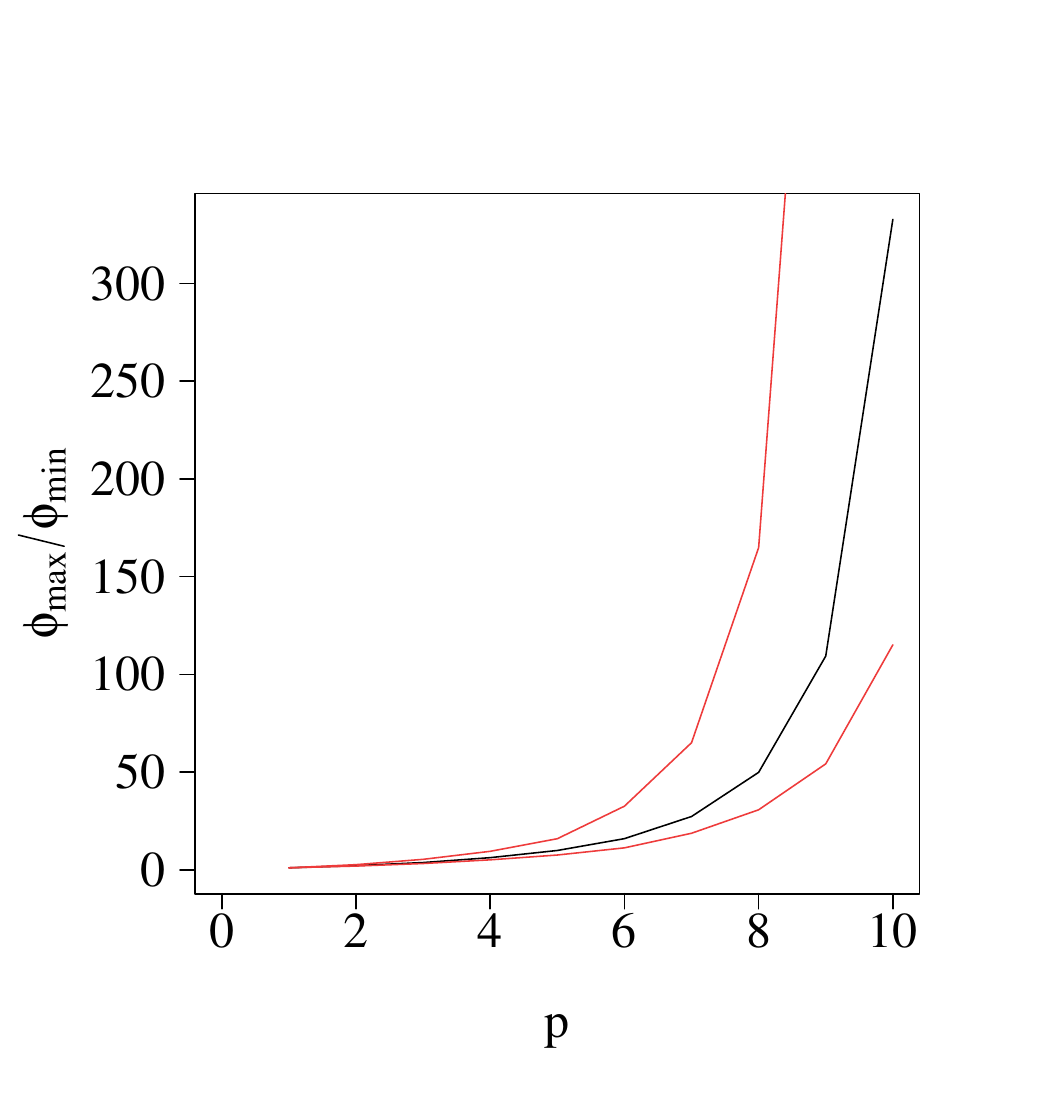}
  \label{fig1c}
\end{subfigure}
\caption{
\small{(a) Empirical (histogram) vs theoretical (solid line) distribution of coherence for {$p=7$.} (b) Plot of the $l_{\infty}$ norm of spectral matrix estimation error $\left\|\hat{\mathbf{S}}(\omega) - \mathbf{S}(\omega)\right\|_{\infty}$ versus dimensionality $p$. (c) Plot of the condition number of the multivariate spectrum as a function of dimensionality $p$. Empirical $95\%$ confidence intervals are shown in red. All plots are generated using $m=10$ tapers, and $\omega=0.0628$ without loss of generality (since the process is Poisson).}} 
\label{plt:empirical_properties}
\end{figure}
}

\section{Regularised Spectral Estimation}
\label{sec:estimation}

\normalsize{
We propose a novel methodology to estimate the inverse spectral density matrix for a high-dimensional multivariate point process at a specified frequency of interest. While estimates of a similar nature have undergone significant development for time series data \citep{jung2015graphical, nadkarnisparse, tugnait2021sparse, dallakyan2022time, deb2024regularized} to the best of our knowledge, this is the first examination of related estimators in the point process framework.

Our estimator is motivated by the asymptotic complex normality of our tapered Fourier coefficients (Theorem 4.2 from \citet{brillinger1972}), where we appeal to the following pseudo (negative) log-likelihood  \citep{whittle1953estimation}
\begin{equation*}
    \mathcal{L}(\boldsymbol{\Theta}(\omega)):=-\log \operatorname{det}(\boldsymbol{\Theta}(\omega))+ \Tr\left\{\mathbf{\hat{S}}(\omega) \boldsymbol{\Theta}(\omega)\right\}, 
    \label{whittle_log_lik}
\end{equation*}

where  $\operatorname{det(\cdot)}$ and $\Tr\{\cdot\}$ denote the determinant and trace of a matrix respectively, and $\mathbf{\hat{S}}(\omega) \in \mathbb{C}^{p\times p}$ is our multi-taper periodogram estimate. 

If $m>p$ the negative log-likelihood achieves its minimum for $\hat{\boldsymbol{\Theta}}(\omega)=\hat{\mathbf{S}}(\omega)^{-1}$. However, as discussed in Section \ref{sec:MT_periodogram}, in the high-dimensional setting where $p>m$, the inverse of the multi-taper periodogram is not uniquely defined. To overcome this, it is common to add a penalty to $\mathcal{L}(\omega)$ in order to constrain the maximum likelihood estimate $\hat{\boldsymbol{\Theta}}(\omega)$ to a feasible region. We consider two such penalties below, each possessing their unique advantages
\begin{align}
\label{pen2}
    P_{1}\{\boldsymbol{\Theta}(\omega)\} & = \left\|\boldsymbol{\Theta}(\omega)\right\|_{1} := \sum_{qr} |\Theta_{qr}(\omega)| = \sum_{qr}\sqrt{\operatorname{Re}^2(\Theta_{qr}(\omega)) + \operatorname{Im}^2(\Theta_{qr}(\omega))}\\
\label{pen1}
    P_{2}\{\boldsymbol{\Theta}(\omega)\} &= \|\boldsymbol{\Theta}(\omega)^{1/2}\|^2_F = \Tr\{\boldsymbol{\Theta}(\omega)\}.
\end{align}
The penalty in \eqref{pen2} encourages sparsity in the inverse spectrum by jointly penalising both the real and imaginary parts of the entries in $\boldsymbol{\Theta}(\omega)$. This is an example of the group-lasso regulariser \citep{huang2010benefit}. In cases where the inverse spectrum is known to be sparse, it is beneficial to use the lasso-type penalty. 
By contrast, the penalty in \eqref{pen1} proportionally shrinks elements of $\boldsymbol{\Theta}(\omega)$ (without inducing sparsity) and can therefore be viewed as a ridge type estimator. Similar estimators for the estimation of high-dimensional covariance matrices have been discussed in \cite{friedman1989regularized} and \cite{warton2008penalized}.

\begin{definition}{{Regularised-Spectral Estimator (RSE).}}
Let $\lambda>0$, the regularised spectral estimator for $\boldsymbol{\Theta}^*(\omega)$ is defined as the M-estimator
\begin{equation}
\label{estimator}
    \hat{\boldsymbol{\Theta}}(\omega):= \arg \min _{\boldsymbol{\Theta}(\omega) \in \mathcal{C}}\left\{\mathcal{L}(\boldsymbol{\Theta}(\omega)) +\lambda P_{i}\left \{\boldsymbol{\Theta}(\omega)\right\}\right\},
\end{equation}
for $i\in\{1,2\}$ and 
with constraint set $\mathcal{C}:=\{\boldsymbol{\Theta}(\omega) \in \mathbb{C}^{p\times p}:\boldsymbol{\Theta}(\omega) \succ 0\}.$
\end{definition}
Going forward, we will refer to the RSE with penalties \eqref{pen2} and \eqref{pen1} as the Lasso and  Ridge estimators, respectively.

An advantage of the Ridge estimator is that solutions to \eqref{estimator} can be found analytically. Specifically,
\begin{equation*}
    {\hat{\boldsymbol{\Theta}}(\omega) = \left(\hat{\mathbf{S}}(\omega)+\lambda \mathbb{I}\right)^{-1},}
\end{equation*}
where $\hat{\mathbf{S}}(\omega)$ is an estimate of the spectrum and $\mathbb{I}$ is the identity matrix. When the Lasso penalty in \eqref{pen2} is used, the convex optimisation problem \eqref{estimator} can be solved numerically using the alternating direction method of multipliers (ADMM) algorithm, or via the complex graphical lasso \citep{deb2024regularized}. In this paper, we adopt the ADMM approach due to it's simplicity. We note this has been used throughout the time series literature to solve optimisation problems similar to those considered in this work \citep{jung2015graphical, nadkarnisparse, dallakyan2022time}. A full description of the algorithm is given in the Supplementary Material.

\subsection{Error Bounds for the Regularised-Spectral Estimator}

In this section, we establish asymptotic properties of the Lasso and  Ridge estimators, providing rates in Frobenius norm. We begin by specifying the following assumptions about the true $\boldsymbol{\Theta}^*(\omega)$ and $\mathbf{S}(\omega)$.

\begin{enumerate}
    \item[A1.] Let the set $E(\boldsymbol{\Theta}^*(\omega)) := \{(q,r)|{\Theta}^*_{qr}(\omega)\neq 0, q\neq r\}$. Then $\operatorname{card}(E(\boldsymbol{\Theta}^*(\omega))) \leq s$.
    \item[A2.] $0<\underline{k}\leq \phi_{\min}(\mathbf{S}(\omega)) \leq \bar{k} < \infty.$
\end{enumerate}
The second assumption provides a lower bound on the eigenvalues of the spectral density matrix, and thus ensures positive definiteness of ${\mathbf{S}(\omega)}$. This guarantees the existence of $\boldsymbol{\Theta}^*(\omega)$.

\begin{proposition} 
\label{prop:ridge}
Let $\hat{\boldsymbol{\Theta}}_R(\omega)$ be the RSE defined in \eqref{estimator} with penalty \eqref{pen1},  and  regularisation parameter $\lambda = \epsilon^{-1}c\sqrt{{p\tau\log p}/{m}}$ for any $\tau>2$ and $0<\epsilon<1$. Let $c=80 \sqrt{2} \max_q\{S_{qq}(\omega)\}$, then, under A2 and using sufficient tapers,
\begin{equation*}
\label{eq:m_bound}
 2c^{-2}\tau \log p < m < {pc^2 \tau\log p}/{\epsilon^2\bar{k}^2},
\end{equation*}

we have,
\begin{align*}
    &\left\|\hat{\boldsymbol{\Theta}}_R(\omega) - \boldsymbol{\Theta}^*(\omega)\right\|_F \leq  \frac{36c\tau^{1/2}}{\epsilon\underline{k}^2}\sqrt{\frac{p^2 \log p}{m}},
\end{align*}
with probability at least $1-8/p^{\tau-2}$.
\end{proposition}
\begin{proof}
   The proof is given in the Supplementary Material. 
\end{proof}

The above result is based on lower-bounding the eigenvalues of the estimated spectrum, alongside bounding the deviations of the periodogram based on Proposition \ref{prop:dev_bound}. In the Ridge case, a somewhat naive bound on the eigenvalue of the inverse spectral density is given by $\|(\hat{\mathbf{S}}+\lambda \mathbb{I})^{-1}\|_2\le \lambda^{-1}$, which leads to the stated bound and holds for $m=\Omega(p\log p)$, i.e. a high-dimensional setting. If we want to study the low-dimensional performance of the estimator, when $m\ge p$, then we can consider an alternative bound based on the Wishart distribution such that $\phi_{\min}(\hat{\mathbf{S}}+\lambda \mathbb{I})\ge c_1\underline{k}$ in high-probability \citep[Theorem 6.1][]{wainwright2019high}. Either way, we still obtain that the estimation error is of order $\|\hat{\boldsymbol{\Theta}}_R(\omega) - \boldsymbol{\Theta}^*(\omega)\|_F=\mathcal{O}_p(p\sqrt{\log p/m})$.

With the Ridge estimator, whilst we are able to obtain estimates of the spectrum (and it's inverse) in high-dimensions, the rate is quite poor, and is unable to adapt to the structure within the set $E(\boldsymbol{\Theta}^*(\omega))$. As an alternative, we now present a result on the Lasso spectral estimator, which can make use of this structure through selecting entries to be set exactly to zero.

\begin{proposition} 
\label{prop:glasso}
\textcolor{black}{Let $\hat{\boldsymbol{\Theta}}_G(\omega)$ be the RSE defined in \eqref{estimator} with penalty \eqref{pen2} and  regularisation parameter $\lambda = \epsilon^{-1}c \sqrt{{\tau \log p}/{m}}$ for any $\tau>2$ and $0<\epsilon<1$. Under A1-2 and with sufficient tapers,}
\begin{equation*}
m > \tau \log p \max\left\{2c^{-2}, (p+s)\{18c\bar{k}/\epsilon\underline{k}^2\}^2 \right\},
\end{equation*}
we have, 
\begin{align}
    &\textcolor{black}{{\left\|\hat{\boldsymbol{\Theta}}_G(\omega)- \boldsymbol{\Theta}^*(\omega)\right\|_F \leq \frac{36c\tau^{1/2}}{\epsilon\underline{k}^2}\sqrt{\frac{(p+s)\log p}{m},}}}\label{bound_glasso}
\end{align}
with probability at least $1-8/p^{\tau-2}$. 
\end{proposition}
\begin{proof}
   The proof is given in the Supplementary Material and follows an argument similar to that in \cite{rothman2008sparse}. 
\end{proof}

The result for the Lasso estimator sharpens the rate by order $p^{1/2}$ compared with the Ridge. The order $\sqrt{p+s}$ term in the bound originates as we need to estimate both the diagonal and off-diagonal structure in the inverse spectral density. In comparison to Proposition \ref{prop:ridge}, we note that this bound will hold as long as $m=\Omega((p+s)\log p)$, whilst this allows a relatively high-dimensional setting we still need sufficient samples to ensure $\|\hat{\boldsymbol{\Theta}}_G(\omega)\|_2$ is bounded. In the following section we further sharpen these bounds by introducing a more restrictive incoherence condition on the spectrum. 

\subsection{Consistent Estimation of Partial Coherence Graphs}

The Lasso RSE in \eqref{estimator} can be used to construct a frequency specific partial coherence graph $\mathcal{G}(\omega)=({V},{E})$ which can be used to visualise conditional relationships between dimensions of the multivariate point process \citep{dahlhaus2000graphical}. In this graph, the nodes ${V}=\{1,\dots,p\}$ represent the $q=1,\ldots,p$ processes $N_q(t)$, and the edge set is characterised by the partial spectral coherence given in \eqref{eq:partial_co}. If the partial spectral coherence $\rho^*_{qr}(\omega)=0 $ for all frequencies $\omega$, then we declare an edge between nodes $q$ and $r$ would be missing in the corresponding conditional independence graph \citep{eichler2003partial}.

Whilst we have shown consistency of the Lasso estimator in \eqref{bound_glasso}, it remains to demonstrate that, with high probability, this estimate can correctly identify the zero pattern of the matrix $\boldsymbol{\Theta}^*(\omega)$. \textcolor{black}{In this section, we state a result on the model selection consistency of the Lasso estimator, which builds on the work presented in \cite{ravikumar2011high}}. Before stating our result, we first set out some quantities required in our analysis.

We denote the augmented edge set by $\mathcal{M}(\omega):= E(\boldsymbol{\Theta}^*(\omega)) \cup \{(1,1), \dots, (p,p)\}$ and its complement by $\mathcal{M}^{\perp}(\omega)$. Going forward, we adopt the shorthand $\mathcal{M}$ and $\mathcal{M}^{\perp}$ respectively. We also define the maximum node degree as $d:=\max_{q=1,\dots, p}\left|\{ r \in \{1, \dots, p\} : \Theta^*_{qr} \neq 0 \}\right|$ which is the maximum number of non-zeros in any row of $\mathbf{\Theta}^*(\omega).$

Let the Hessian of the Whittle likelihood at frequency $\omega$ be given by 
$
    \Gamma(\omega)=\mathbf{S}(\omega) \otimes \mathbf{S}(\omega) \in \mathbb{C}^{p^2 \times p^2}$.
We define the term 
\begin{equation}
    \kappa_{S} := \opnorm{\mathbf{S}(\omega)}{\infty} = \max_{q=1,\dots,p}\sum_{r=1}^p\left|S_{qr}(\omega)\right|,
    \label{kappa_s}
\end{equation}
corresponding to the $\ell_{\infty}$ norm of the true spectral density matrix. Similarly, we consider the subset of the Hessian relevant to the true model subset $\Gamma_{\mathcal{M}\mathcal{M}}(\omega)$ and let 
\begin{equation}
    \kappa_{\Gamma} = \opnorm{\left(\Gamma_{\mathcal{M}\mathcal{M}}(\omega)\right)^{-1}}{\infty}.
    \label{kappa_gam}
\end{equation}

The two quantities \eqref{kappa_s} and \eqref{kappa_gam} play an important role in describing the behaviour of the Lasso estimator, and are used in the final result.
 
We now introduce a frequency specific incoherence condition, that holds across all frequencies, which is required for the model selection consistency of the Lasso estimator. 

\begin{enumerate}
    \item[A3.] Let $\Gamma := \Gamma(\omega)$. Then for frequency $\omega \in \mathbb{R}$ there exists some $\alpha \in (0,1]$ such that
    \begin{equation}
        \max_{e \in \mathcal{M}^{\perp}} \left\| \Gamma_{e\mathcal{M}}(\Gamma_{\mathcal{M}\mathcal{M}})^{-1} \right\| \leq (1-\alpha).
        \label{ass:incoherence}
    \end{equation}
\end{enumerate}
The intuition behind this assumption is that it limits the influence of the non-edge terms, indexed by $\mathcal{M}^{\perp}$, on the edge terms, indexed by $\mathcal{M}$. 

Our next result establishes an elementwise deviation bound for the Lasso estimator which can be used to establish consistent estimation of the resulting partial coherence graphs.

\begin{proposition}{\textit{Consistent Estimation of Partial Coherence Graphs.}}
\label{prop:l_inf}
Consider a distribution satisfying the incoherence assumption \eqref{ass:incoherence} with parameter $\alpha \in(0,1].$ Let $\boldsymbol{\hat{\Theta}}_G(\omega)$ be the RSE defined by \eqref{estimator} with penalty \eqref{pen2} and regularisation parameter $\lambda=(8c/\alpha) \sqrt{\log p^{\tau}/{m}}$ for any $\tau>2$. Then, under assumptions A1-3. with $p>3$ and sufficient tapers,

\begin{equation*}
    m> \tau \log p \max \left\{ 2c^{-2}, \left\{6c(1+8/\alpha)^2 d \kappa^3_S \kappa^2_\Gamma\right\}^2 , \left\{6c(1+8/\alpha) d \kappa^2_S \kappa_\Gamma \max(1, \kappa_S \kappa_\Gamma)\right\}^2 \right\},
\end{equation*}
where $c=80 \sqrt{2} \max_q\{S_{qq}(\omega)\}$, with probability greater than $1-8/p^{\tau-2}$, we have:
\begin{itemize}
    \item[(a)] The estimate $\hat{\boldsymbol{{\Theta}}}_G(\omega)$ satisfies the element-wise $\ell_{\infty}$-bound: 
    \begin{equation}
    \label{bound:l_inf}
        \left\|\hat{\boldsymbol{{\Theta}}}_G(\omega) - \boldsymbol{\Theta}^*(\omega)\right\|_{\infty} \leq 2c \kappa_{\Gamma}\left(1+\frac{8}{\alpha}\right)\sqrt{\frac{\tau \log p}{m}}.
    \end{equation}
    \item[(b)] The estimate $\hat{\boldsymbol{{\Theta}}}_G(\omega)$ specifies an edge set $E(\hat{\boldsymbol{{\Theta}}}_G(\omega))$ that is a subset of the true edge set $E(\boldsymbol{\Theta}^*(\omega))$, and includes all edges $(q,r)$ with $\left|\Theta^*_{qr}(\omega)\right|> 2c \kappa_{\Gamma}\left(1+\alpha^{-1}8\right)\sqrt{\tau \log p/m}$.
\end{itemize}


\end{proposition}

\begin{proof}
    The proof, which is based on a primal dual witness approach, is given in the Supplementary Material.
\end{proof}

Part $(b)$ of Proposition \ref{prop:l_inf} states that the estimated edge set is a subset of the true edge set. More specifically, it asserts that $E(\boldsymbol{\hat{\Theta}}_G(\omega))$ correctly excludes all non-edges and only includes edges that are `large' relative to decay of the error $c\sqrt{\tau \log p /m}$. 

Sign consistency can also be achieved by placing a lower bound on the entries in the inverse spectral density matrix. In this case, not only will the estimator have the same edge set as $\boldsymbol{\Theta}^*(\omega)$, but it will also be able to recover the correct signs on these edges. Mathematically, this translates to the event
\begin{equation*}
    \mathcal{K}(\boldsymbol{\hat{\Theta}}(\omega); \boldsymbol{\Theta}^*(\omega)) := \left\{\operatorname{sign}(\hat{\Theta}_{qr}(\omega)) = \operatorname{sign}(\Theta^*_{qr}(\omega)) \ \forall (q,r) \in \mathcal{M}(\omega)\right\}.
\end{equation*}
Using this notation, and with a sufficiently modified bound on the number of tapers, the estimator $\hat{\boldsymbol{\Theta}}_G(\omega)$ is model selection consistent with high probability as $p \rightarrow \infty$, 
\begin{equation*}
    \mathbb{P}( \mathcal{K}(\boldsymbol{\hat{\Theta}}(\omega); \boldsymbol{\Theta}^*(\omega))) \geq 1-8/p^{\tau-2} \rightarrow 1.
\end{equation*}
The sign consistency can also be achieved without requiring Assumption~A3 and instead using a thresholded or adaptive lasso penalty \cite{buhlmann2011statistics}. 

    }

\section{Synthetic Experiments}
\normalsize{
We use synthetic data to compare the performance of the Lasso and  Ridge estimators to existing statistical approaches where applicable.
The goal of the study is two-fold: (1) to assess the accuracy of the recovered inverse spectrum and (2) to investigate the conditional relationships of high-dimensional multivariate point processes. 
Throughout, we use simulated data from different parameterisations of the multivariate Hawkes process. This particular point process has been used extensively to model neuronal interactions in the brain, and as such will serve as an illustrative example in our discussion \citep[see for example][]{pernice2011structure, chen2017nearly, lambert2018reconstructing, wang2021joint}.

\subsection{{Simulation Setting}}
{In the simulation study, we consider three different parameterisations of stationary multivariate Hawkes processes. Motivated by applications is neuroscience, we consider processes that are governed by exponential excitation functions. More specifically, the conditional intensity function of the process can be written as}
\begin{equation*}
    {\boldsymbol{\Lambda}=(\mathbb{I}-\mathbf{G}(0))^{-1}\boldsymbol{\nu}, }
\end{equation*}
{where $\boldsymbol{\nu}$ is a vector of background intensities, $\mathbb{I}$ is the identity matrix and $ \mathbf{G}$ is the Fourier transform of the excitation function defined in the exponential case as}
\begin{equation*}
   {\mathbf{G}(\omega)=\int_{-\infty}^{\infty}e^{-i\omega \tau} \boldsymbol{\alpha}e^{-\boldsymbol{\beta}(\tau)}d\tau = \frac{\boldsymbol{\alpha}}{\boldsymbol{\beta}+i \omega}.}
\end{equation*}
{The spectral density matrix of a stationary multivariate Hawkes process is subsequently defined as}
\begin{equation*}
    {\mathbf{S}(\omega)=\frac{1}{2\pi}\{\mathbb{I}-\mathbf{G}(\omega)\}^{-1} \mathbf{D}\{\mathbb{I}-\mathbf{G}^{\prime}(-\omega)\}^{-1},}
\end{equation*}
{where $\mathbf{D}=\textrm{diag}(\boldsymbol{\Lambda}).$ Further details regarding the multivariate Hawkes process used in this section can be found in the Supplementary Material.}

We consider multivariate Hawkes processes of dimension $p=12, 48$ and $96$ and simulate $m=10$ and $m=50$ independent trials of synthetic data. 
These data are simulated via Ogata's method \citep{ogata1981lewis} using the \verb|hawkes|\footnote{https://cran.r-project.org/web/packages/hawkes/index.html} package in R, and each dimension of the process is simulated for $T=200$ seconds. 

Multiple scenarios are considered in this study. In the first instance, we simulate from a $p$-dimensional Hawkes process whose dependence structure is captured by a block diagonal excitation matrix $\boldsymbol{\alpha}$. Each block consists of a $3 \times 3$ matrix $\boldsymbol{\alpha}^*$ designed to ensure both self- and cross-excitation between the dimensions in a particular block. Thus, dependence between dimensions of the process is null between the blocks and  non-null within each block. 
We consider two models of this nature: model (a) in which there are interactions between two of the processes within a block, and an additional independent process and model (b) where all three processes in a particular block interact with each other. In the final setting (c) we simulate from a $p$-dimensional Hawkes process whose dependence structure is driven by an arbitrarily sparse excitation matrix $\boldsymbol{\alpha}$. All parameters were chosen in such a way to ensure both stability and stationarity of the multivariate process. For further details on the specification of the Hawkes process used in this study, see the Supplementary Material.

\subsection{Parameter Tuning and Performance Metrics}
The Ridge and Lasso estimators require careful specification of a tuning parameter $\lambda$. 
Using a training group of simulated data, we obtain estimates of $\boldsymbol{\Theta}^*(\omega)$ over a fixed search grid of $\lambda$ values. 
\textcolor{black}{To assess the effect of the tuning parameter, }
the optimal value of $\lambda$ \textcolor{black}{in this simulation is} selected according to an appropriate performance metric. Often, one may consider a measure of accuracy (i.e. precision of the recovered inverse spectrum), or edge selection (i.e. estimation of the correct sparsity pattern). For the former, we use the notion of mean squared error (MSE) defined in \eqref{eq:MSE} and for the latter, we consider the $\mathrm{F}_1$ score defined as
\begin{equation*}
    \mathrm{F}_1 =\frac{\textrm{TP}}{\textrm{TP}+\frac{1}{2}(\textrm{FP}+ \textrm{FN})}, 
\end{equation*}
where $\mathrm{TP}$ is the number of correctly identified edges in the estimate graph, while FP and FN are the number of false positives and false negatives respectively. 
For each multivariate point process in the training group, an optimal choice of $\lambda$ is chosen either by maximising the $\mathrm{F}_1$ score or minimising the MSE. These values are then averaged across the number of samples in the training set to arrive at a final optimal parameter $\lambda^*.$

We estimate the inverse spectral density matrix across $N=100$ Monte Carlo samples. Then, we quantify the accuracy of our estimation using the mean squared error (MSE) defined as
\begin{equation}
    \textrm{MSE}\left\{\hat{\boldsymbol{\Theta}}^{(n)}(\omega), \boldsymbol{\Theta}^*(\omega)\right\} =  \frac{2}{N p(p-1)}\sum_{n=1}^{N} \sum_{q<r} \left(\hat{\Theta}^{(n)}_{qr} - \Theta^*_{qr}\right)^2
    \label{eq:MSE}
\end{equation}
where $N$ denotes the number of Monte Carlo samples in our simulation study.

We also report results for the recovery of the sparsity pattern for the Lasso estimate. Specifically, we define the true positive rate (TPR) and false positive rate (FPR), respectively, as the proportion of true and false edges in the graph.
Since the Ridge and periodogram estimators do not produce sparse results, we only report the TPR and FPR for the Lasso estimate. Numerical results averaged across 100 repetitions are presented in Table \ref{tab:Sim1}.

\subsection{Simulation Results}

Table \ref{tab:Sim1} reports simulation results for Models (a)-(c), respectively. Numerical results are averaged across 100 Monte Carlo samples, and each entry of the table is of the form mean (standard error). We omit standard errors of $\leq 10^{-2}$ for brevity. Lasso$_1$ and Lasso$_2$ refer to the Lasso estimator tuned using the MSE and $\mathrm{F}_1$ score respectively.   

{Examining the results across all simulation settings, one can see that the multi-taper periodogram matrix cannot be inverted in settings where $p>m$, thus motivating the need for regularisation or shrinkage.}
\begin{table}[t]
\small
 \caption{\textit{Simulation results over 100 replications for estimating the inverse spectral density matrix, with mean squared error (MSE) and F$_1$ score. All results are recorded at a particular frequency $\omega=0.0628$ and are in the form of mean (standard error). Standard errors of $<10^{-2}$ are omitted for brevity. Hyphenated entries (-) denote that the multi-taper periodogram matrix could not be inverted.}}
   \centering  
  \begin{tabular}{ccccccccc}
         \multirow[b]{2}{*}{Model} & \multirow[b]{2}{*}{$p$} & \multirow[b]{2}{*}{$m$}  & \multicolumn{4}{c}{\text { Mean Squared Error - ISDM }}  & \multicolumn{2}{c}{\text{F$_1$ Score}} \\

 & & & \text {Inverted Periodogram} & \text { Ridge } & \text {Lasso}$_1$ & \text {Lasso}$_2$   & \text {Lasso}$_1$ & \text {Lasso}$_2$ \\ 

\\

(a) & 12 & 10 & -  &22.00 (0.03) & 19.27 (0.02) & 28.62 & 0.11 & 0.75 \\
 & & 50 & 9.80 (0.01) & 9.86 (0.01)  & {20.21} (0.01) & 28.08 & 0.11& 0.99 \\ 
 & 48 & 10 & - &  6.51  & 5.31  & 6.74 & 0.06 & 0.70 \\
& & 50  & 33944.94 (450.89) &  4.77  & 3.49 & 6.58 & 0.04 & 0.98 \\ 
& 96 & 10 & - & 3.29  & 2.81  & 3.34 & 0.04 & 0.69 \\
& & 50 & - & 2.89  & 2.04 & 3.27 & 0.03 & 0.99 \\ 
\\
(b) & 12 & 10&- & 2.36 (0.01)  & 1.86 (0.01)  & 4.37 & 0.32 & 0.81  \\
 & & 50 & 1.59  & 1.58 (0.01) & 1.53  & 4.20 & 0.31 & 0.97 \\
 & 48 & 10 &-& 0.70  & 0.37 & 1.03  & 0.13 & 0.74\\
 & & 50 & 6218.81 (151.89) & 0.59  & 0.31  & 1.01  & 0.09 & 0.98  \\
& 96 & 10 &-& 0.43  & 0.19 & 0.51& 0.10 & 0.71 \\ 
  & & 50 &-& 0.28  & 0.14  & 0.50 & 0.05& 0.96\\ 
\\
(c) & 12 & 10 &-&  14.01 (0.01) & 8.56 (0.01) & 12.71 & 0.12 & 0.83 \\
 & & 50& 26.85 (0.08) & 10.06 (0.01) &  9.17 &12.24 & 0.12 & 0.82 \\
 & 48 & 10 & -& 3.19  &  2.12  & 3.10  & 0.06 & 0.77\\
 & & 50 & 34088.54 (451.49) & 3.30  & 1.51  &3.07 & 0.04 & 0.87\\
& 96 & 10 &-& 1.34  & 1.14  & 1.54& 0.05 & 0.73  \\
  & & 50 &-& 1.83  &  0.74  & 1.53 & 0.03 & 0.87 \\
\\

    \end{tabular}
    \label{tab:Sim1}
\end{table}
In terms of MSE, the Ridge and Lasso$_1$ estimators are comparable in setting (a), while Lasso$_1$ yields more favourable results in settings (b) and (c). Comparatively, Lasso$_2$ attains larger values of MSE for models (a) and (b), however it too is competitive with Ridge in setting (c). This highlights a potential drawback of the Ridge estimator in that it is biased towards a scaled identity matrix, of which the sparsity patterns of models (a) and (b) somewhat resemble \citep{fiecas2019spectral}. Therefore, one might suggest the use of the Lasso estimator when the sparsity pattern of the inverse spectrum is known to to be arbitrarily sparse rather than containing some sort of banded structure. Notably, the periodogram estimate is competitive in terms of MSE in scenarios where $p=12$ and $m=50$. This highlights the competence of the multi-taper estimate in scenarios where $p<m$, while also emphasising the importance for appropriate regularisation in high-dimensional settings. 
The $\mathrm{F}_1$ score results highlight that the Lasso$_2$ estimator yields favourable results in terms of sparse recovery of the spectrum, however this comes at the expense of an increase in MSE compared to Lasso$_1$. \textcolor{black}{This confirms the previously observed challenges of selecting the tuning parameter of lasso penalty to provide satisfactory performances in model selection and parameter estimation \citep{meinshausen2006high}.}

Overall, the Lasso estimator is preferable when the goal is to learn the sparsity pattern of the inverse spectrum. However, the Ridge estimator yields competitive results in terms of the accuracy of entries in the recovered estimate. Furthermore, the results presented here highlight that the proposed methodology improves upon existing estimators, such as the multi-taper periodogram, which break down in high-dimensional settings. 
}

\section{Neuronal Synchronicity and Partial Coherence Networks}
\normalsize{
We illustrate how our proposed methodology can be used to infer a network of interactions between a population of neurons, using the spike train data from \cite{bolding2018recurrent}. 
In this experiment, an optogenetic approach was used to directly stimulate the olfactory bulb (OB) region of the mouse brain. In particular, $1$ second light pulses were presented above the OB region at increasing intensity levels from $0$ to $50mW/mm^2$ over a total of $m=10$ experimental trials. 

The dataset consists of spike times for multiple mice, collected at $8$ intensities of laser stimuli. As an illustrative example, we will consider data collected at three stimuli ($0, 10$ and $50mW/mm^{2}$) from a single mouse with $p=26$ neurons measured in the OB region.

\begin{figure}[t]
\centering
\begin{subfigure}[]{0.3\textwidth}
\centering
   \includegraphics[width=0.95\textwidth]{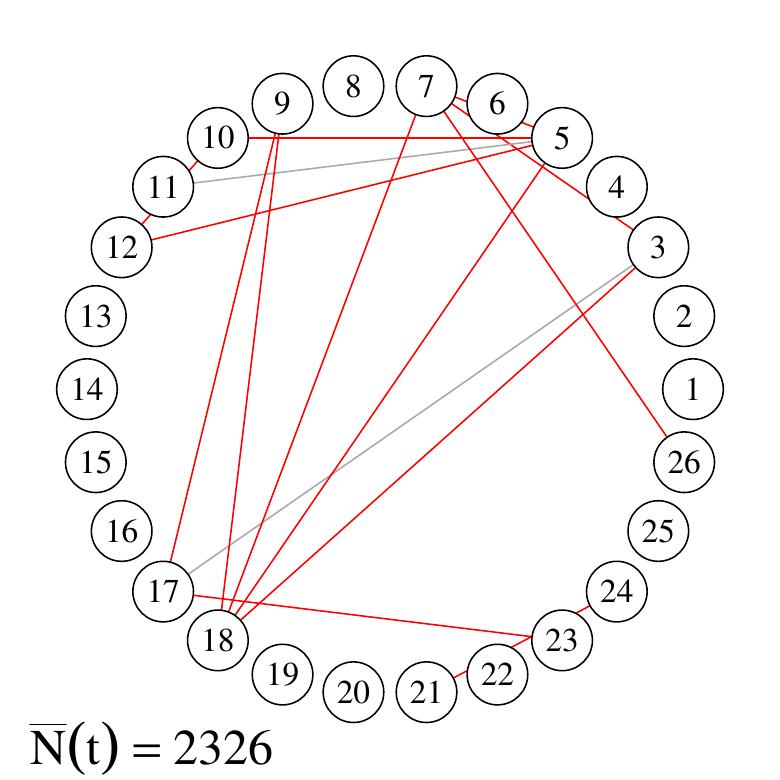}
\end{subfigure}
\begin{subfigure}[]{0.3\textwidth}
   \centering
    \includegraphics[width=0.95\textwidth]{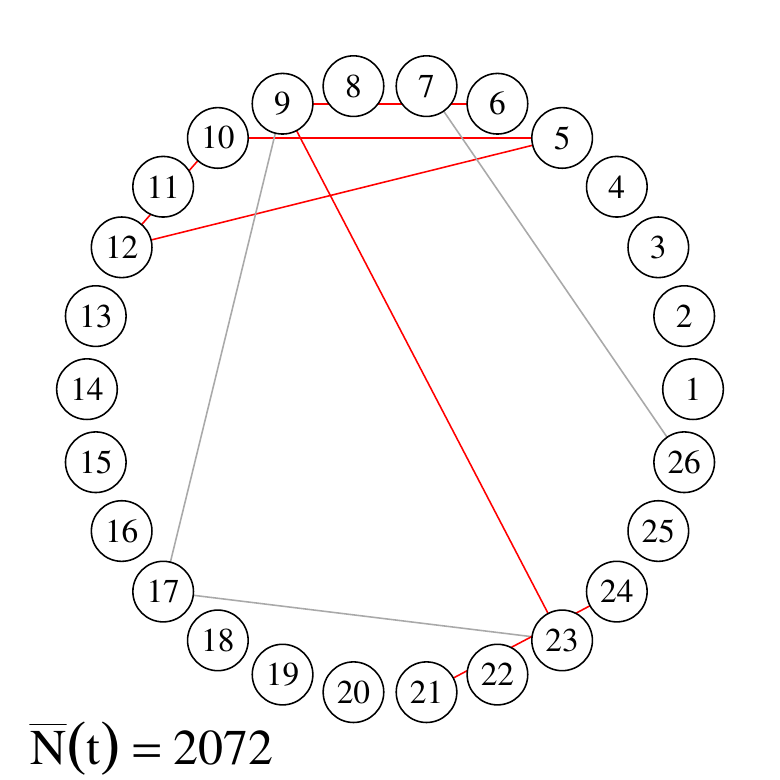}
\end{subfigure}
\begin{subfigure}[]{0.3\textwidth}
\centering
   \includegraphics[width=0.95\textwidth]{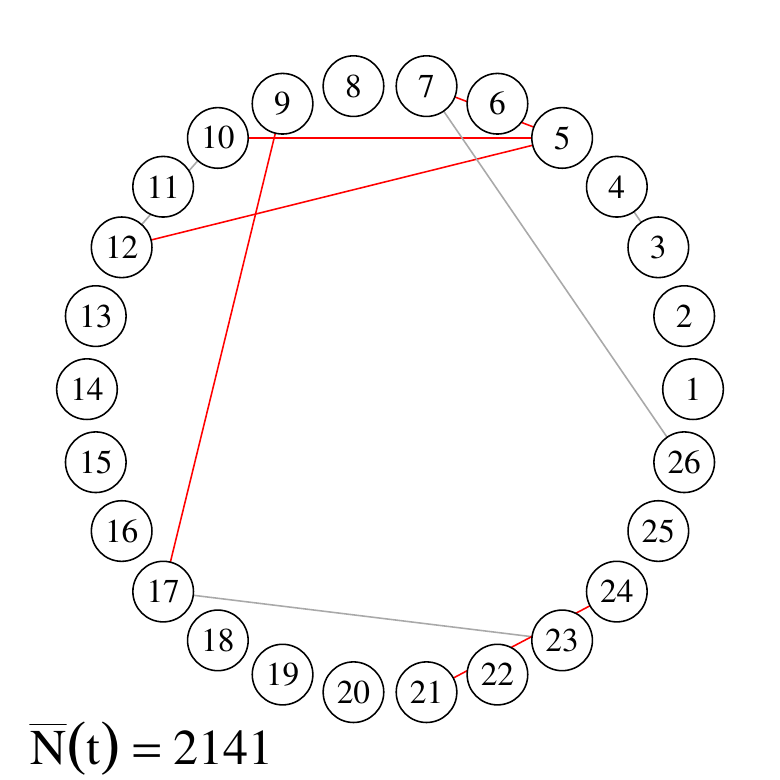}
\end{subfigure}

\begin{subfigure}[]{0.3\textwidth}
\centering
   \includegraphics[width=0.95\textwidth]{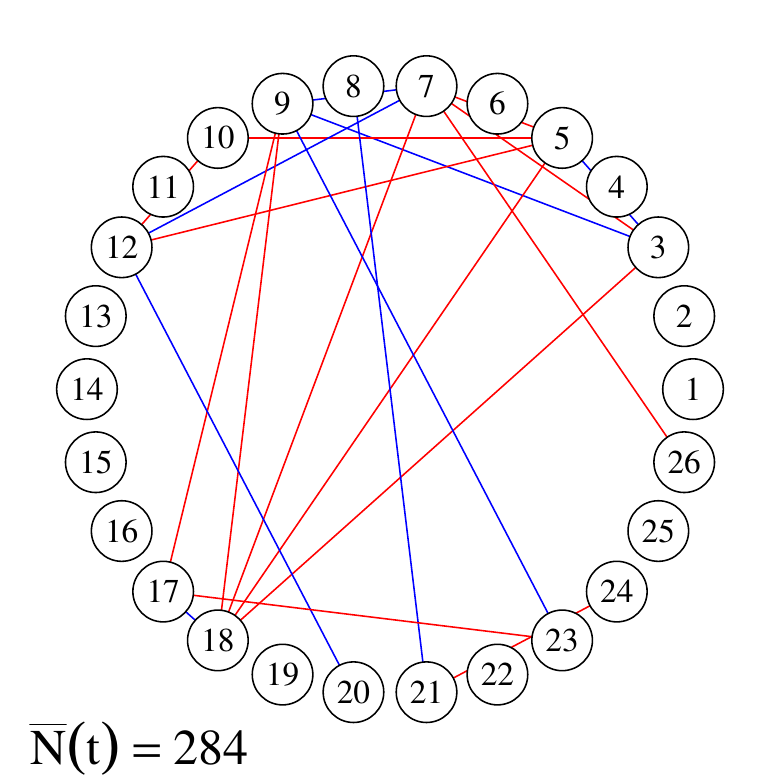}
   \caption{}
\end{subfigure}
\begin{subfigure}[]{0.3\textwidth}
   \centering
    \includegraphics[width=0.95\textwidth]{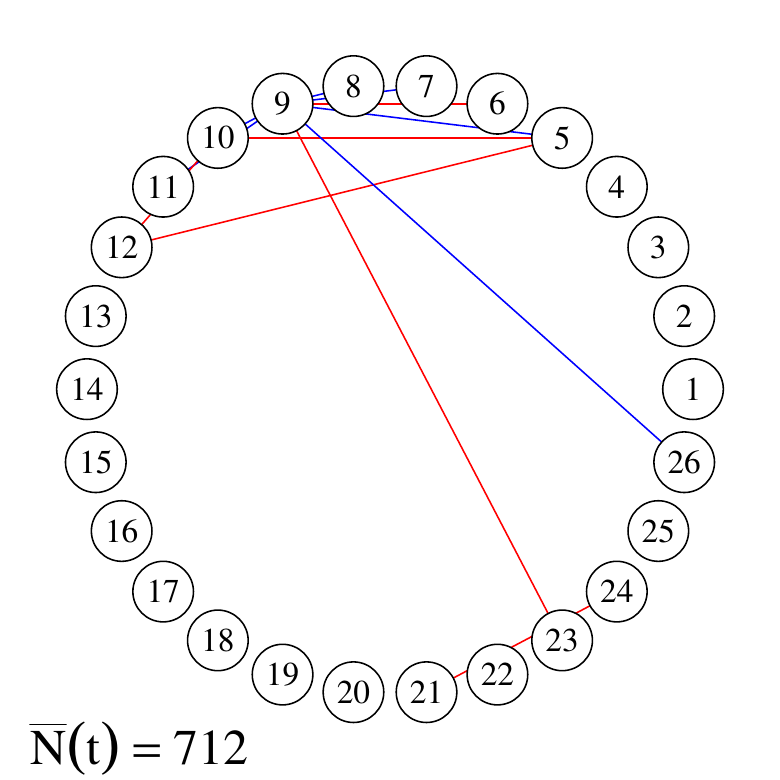}
     \caption{}
\end{subfigure}
\begin{subfigure}[]{0.3\textwidth}
\centering
   \includegraphics[width=0.95\textwidth]{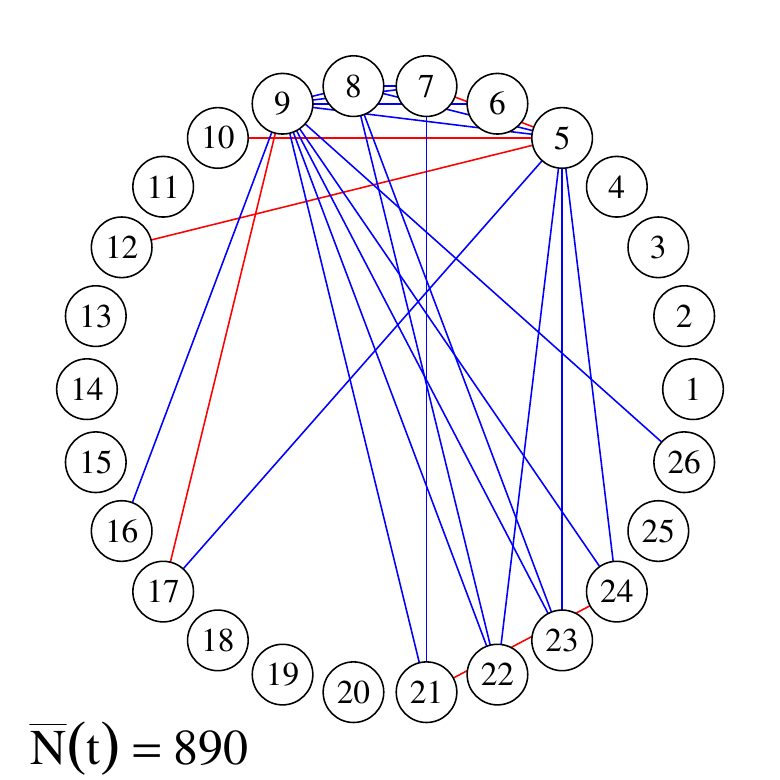}
    \caption{}
\end{subfigure}
\caption{\small{Estimated networks of neural interactions on the delta band using the spike train data from \cite{bolding2018recurrent}. Each column shows the estimated networks for the laser off (top) and laser on (bottom) condition for laser stimuli (a) $0mw/mm^2$, (b) $10mw/mm^2$ and (c) $50mw/mm^2$. Common edges between the on and off conditions for each level of intensity are shown in red, and edges unique to the laser on condition are shown in blue. All other edges are shown in grey.}}
\label{plt:networks}
\end{figure}

For each of the considered levels of stimuli, data is collected during a `laser on' and `laser off' condition. For the laser on condition, we have $1$ second worth of data, compared to a total of $8$ seconds for the laser off condition. We use our Lasso estimator to visualise functional connectivity for each condition, under the differing levels of laser stimuli.  Since a laser stimulus of intensity $0 mw/mm^2$ is equivalent to the laser off condition, we use this setting as a plausibility check for the proposed methodology.

Initial exploratory analysis showed that each of the $26$ point processes are driven by lower frequencies in the spectrum, which in neuroscience are named as frequencies in the delta, $(0,4)$ Hz, and theta, $(4,8)$ Hz, bands \citep{buzsaki2004neuronal}. Consequently, we will apply our method to these lower frequency bands.
The overall aim of our analysis is to use the regularised estimator of the partial coherence matrix to determine if and how estimates of neuronal connectivity differ in response to the changing stimuli.  

To estimate connectivity in a particular frequency band, we consider an adaptation of the Lasso estimator whereby we use both a trial and frequency smoothed periodogram estimate as input to \eqref{estimator}. Specifically, we define the trial-frequency smoothed estimator as
\begin{align*}
    \hat{\mathbf{S}}(\omega_*) = \frac{1}{\operatorname{card}(F)}\sum_{f\in F}\frac{1}{m}\sum_{k=1}^m \bar{\mathbf{d}}_k(\omega_f)\bar{\mathbf{d}}_k^H(\omega_f),
\end{align*}
where the frequencies $f$ lie in a specified range. 
This modified version of the Lasso estimator is tuned using the extended Bayesian Information Criteria (eBIC) with the hyper-parameter set to $0.5$ as suggested in \cite{foygel2010extended}. The degrees of freedom are determined by the number of non-zero elements in our estimate $\hat{\boldsymbol{\Theta}}(\cdot).$

{Figure \ref{plt:networks} shows the estimated partial coherence graphs obtained for the delta band under the laser on and laser off conditions. That is, we present estimates of neuronal interactions for frequencies in the range $(0,4)$ Hz that are specific to each experimental condition. In these graphical representations, each node represents a neuron and an edge indicates a non-zero value of partial coherence estimated using the Lasso estimator. Each plot is annotated with the average number of spikes that are observed under each experimental condition. We highlight that more spikes are observed in the laser off condition compared to the laser on condition, which is a consequence of the differing lengths of the observation periods. Moreover, for the laser on condition, the number of spikes observed increases with the increasing level of laser stimuli.

In general, more edges are detected when the laser stimuli is applied ($12$  under $10 mW/mm^2$ and $24$ under $50 mW/mm^2$ versus $9$ under both $10 mW/mm^2$ and $50 mW/mm^2$ for the laser off setting). For the $0 mW/mm^2$ setting, $15$ edges are detected in the `laser off' condition versus $21$ for the `laser on' condition. Whilst not exactly the same, it is easy to see that the estimated networks are visually comparable under both conditions, as we should expect.
Notably, there are far fewer spikes observed in the laser on condition compared to laser off condition. However, the structural similarities in the estimated networks highlight our methodologies' robustness to the number of points, indicating that our asymptotic assumption is relevant. 
{Overall, the results presented here align with the findings in \cite{bolding2018recurrent} where it is highlighted by neuroscientists that the OB response is sensitive to the application of external stimuli.}
}

\section{Discussion}
\normalsize{
We have proposed novel statistical methodology for the estimation of high-dimensional inverse spectral density matrices in the point process setting. 
The main advantage of our proposal is the ability to estimate the inverse spectrum when the dimension of the multivariate process is large, i.e. when $p>m$, 
unlike existing methods which break down in high dimensional settings. 
Importantly, our Lasso estimator yields sparse estimates of the inverse spectrum, easing interpretation of the resulting partial coherence networks. Moreover under the stated incoherence condition, this estimator can correctly identify the sparse structure of the inverse spectrum. In general, we would recommend the Lasso based estimator due to this interpretability, however, in situations where the sparsity assumption may not hold, we demonstrate that a Ridge based regulariser can also provide estimates of the spectral density in high-dimensional scenarios.

There are many possible and interesting avenues for future work. Whilst we have demonstrated consistent estimation of the inverse spectrum, it remains to provide a way in which to quantify uncertainty around the resulting partial coherence graphs. 
Existing work on statistical inference has been discussed for 
high-dimensional inverse covariance estimation in \cite{jankova2015confidence}. Furthermore, extending the methodology to non-stationary settings is also of interest, where one could easily extend the localised estimation methods provided in \cite{Roueff2019},which may find utility both in neuroscience, but also in high-frequency econometric applications. Finally, one may consider further extensions beyond the current assumption of non-overlapping tapers, for example, it would be relatively straightforward to apply our regularisation methods (and analysis) in the more general multi-tapering framework of \cite{Walden2000,cohen2019wavelet}.


}

\section*{Acknowledgement}
This paper is based on work completed while Carla Pinkney was part of the EPSRC funded STOR-i centre for doctoral training (EP/S022252/1).

\section*{Supplementary Material}
\label{SM}
The Supplementary Material includes additional results and proofs. The code associated with this paper is available on request.



\bibliographystyle{apalike}
\bibliography{references}

\begin{thebibliography}{}

\bibitem[Bartlett, 1963]{bartlett1963spectral}
Bartlett, M.~S. (1963).
\newblock The spectral analysis of point processes.
\newblock {\em Journal of the Royal Statistical Society: Series B (Methodological)}, 25(2):264--281.

\bibitem[Bickel and Levina, 2008]{bickel2008regularized}
Bickel, P.~J. and Levina, E. (2008).
\newblock Regularized estimation of large covariance matrices.
\newblock {\em The Annals of Statistics}, 36(1):199--227.

\bibitem[B{\"o}hm and von Sachs, 2009]{bohm2009shrinkage}
B{\"o}hm, H. and von Sachs, R. (2009).
\newblock Shrinkage estimation in the frequency domain of multivariate time series.
\newblock {\em Journal of Multivariate Analysis}, 100(5):913--935.

\bibitem[Bolding and Franks, 2018]{bolding2018recurrent}
Bolding, K.~A. and Franks, K.~M. (2018).
\newblock Recurrent cortical circuits implement concentration-invariant odor coding.
\newblock {\em Science}, 361(6407).

\bibitem[Boyd et~al., 2011]{boyd2011distributed}
Boyd, S., Parikh, N., Chu, E., Peleato, B., Eckstein, J., et~al. (2011).
\newblock Distributed optimization and statistical learning via the alternating direction method of multipliers.
\newblock {\em Foundations and Trends{\textregistered} in Machine learning}, 3(1):1--122.

\bibitem[Brillinger, 1972]{brillinger1972}
Brillinger, D. (1972).
\newblock The spectral analysis of stationary interval functions.
\newblock {\em Proceedings of the Sixth Berkeley Symposium on Mathematical Statistics and Probability, Volume 1: Theory of Statistics}, pages 483--513.

\bibitem[B{\"u}hlmann and Van De~Geer, 2011]{buhlmann2011statistics}
B{\"u}hlmann, P. and Van De~Geer, S. (2011).
\newblock {\em Statistics for high-dimensional data: methods, theory and applications}.
\newblock Springer Science \& Business Media.

\bibitem[Buzsaki and Draguhn, 2004]{buzsaki2004neuronal}
Buzsaki, G. and Draguhn, A. (2004).
\newblock Neuronal oscillations in cortical networks.
\newblock {\em {S}cience}, 304(5679):1926--1929.

\bibitem[Chen et~al., 2017]{chen2017nearly}
Chen, S., Witten, D., and Shojaie, A. (2017).
\newblock Nearly assumptionless screening for the mutually-exciting multivariate hawkes process.
\newblock {\em Electronic journal of statistics}, 11(1):1207.

\bibitem[Cohen and Gibberd, 2021]{cohen2019wavelet}
Cohen, E. A.~K. and Gibberd, A.~J. (2021).
\newblock Wavelet spectra for multivariate point processes.
\newblock {\em Biometrika}.

\bibitem[Dahlhaus, 2000]{dahlhaus2000graphical}
Dahlhaus, R. (2000).
\newblock Graphical interaction models for multivariate time series.
\newblock {\em Metrika}, 51(2):157--172.

\bibitem[Dallakyan et~al., 2022]{dallakyan2022time}
Dallakyan, A., Kim, R., and Pourahmadi, M. (2022).
\newblock Time series graphical lasso and sparse var estimation.
\newblock {\em Computational Statistics \& Data Analysis}, 176:107557.

\bibitem[Deb et~al., 2024]{deb2024regularized}
Deb, N., Kuceyeski, A., and Basu, S. (2024).
\newblock Regularized estimation of sparse spectral precision matrices.
\newblock {\em arXiv preprint arXiv:2401.11128}.

\bibitem[Eichler et~al., 2003]{eichler2003partial}
Eichler, M., Dahlhaus, R., and Sandk{\"u}hler, J. (2003).
\newblock Partial correlation analysis for the identification of synaptic connections.
\newblock {\em Biological Cybernetics}, 89(4):289--302.

\bibitem[Fiecas et~al., 2019]{fiecas2019spectral}
Fiecas, M., Leng, C., Liu, W., and Yu, Y. (2019).
\newblock Spectral analysis of high-dimensional time series.
\newblock {\em Electronic Journal of Statistics}, 13(2):4079--4101.

\bibitem[Fiecas and Ombao, 2011]{fiecas2011generalized}
Fiecas, M. and Ombao, H. (2011).
\newblock The generalized shrinkage estimator for the analysis of functional connectivity of brain signals.
\newblock {\em Annals of Applied Statistics}, 5(2A):1102--1125.

\bibitem[Fiecas and von Sachs, 2014]{fiecas2014data}
Fiecas, M. and von Sachs, R. (2014).
\newblock Data-driven shrinkage of the spectral density matrix of a high-dimensional time series.
\newblock {\em Electronic Journal of Statistics}, 8(2):2975--3003.

\bibitem[Foygel and Drton, 2010]{foygel2010extended}
Foygel, R. and Drton, M. (2010).
\newblock Extended {B}ayesian information criteria for {G}aussian graphical models.
\newblock {\em Advances in neural information processing systems}, 23.

\bibitem[Friedman, 1989]{friedman1989regularized}
Friedman, J.~H. (1989).
\newblock Regularized discriminant analysis.
\newblock {\em Journal of the American Statistical Association}, 84(405):165--175.

\bibitem[Goodman, 1963]{goodman1963statistical}
Goodman, N.~R. (1963).
\newblock Statistical analysis based on a certain multivariate complex {G}aussian distribution (an introduction).
\newblock {\em The Annals of Mathematical Statistics}, 34(1):152--177.

\bibitem[Hawkes, 1971a]{hawkes1971point}
Hawkes, A.~G. (1971a).
\newblock Point spectra of some mutually exciting point processes.
\newblock {\em Journal of the Royal Statistical Society: Series B (Methodological)}, 33(3):438--443.

\bibitem[Hawkes, 1971b]{hawkes1971spectra}
Hawkes, A.~G. (1971b).
\newblock Spectra of some self-exciting and mutually exciting point processes.
\newblock {\em Biometrika}, 58(1):83--90.

\bibitem[Huang and Zhang, 2010]{huang2010benefit}
Huang, J. and Zhang, T. (2010).
\newblock The benefit of group sparsity.
\newblock {\em The Annals of Statistics}, 38(4):1978--2004.

\bibitem[Jankov{\'a} and van~de Geer, 2015]{jankova2015confidence}
Jankov{\'a}, J. and van~de Geer, S. (2015).
\newblock Confidence intervals for high-dimensional inverse covariance estimation.
\newblock {\em Electronic Journal of Statistics}, 9(1):1205--1229.

\bibitem[Jung et~al., 2015]{jung2015graphical}
Jung, A., Hannak, G., and Goertz, N. (2015).
\newblock Graphical lasso based model selection for time series.
\newblock {\em IEEE Signal Processing Letters}, 22(10):1781--1785.

\bibitem[Kocsis et~al., 1999]{kocsis1999interdependence}
Kocsis, B., Bragin, A., and Buzs{\'a}ki, G. (1999).
\newblock Interdependence of multiple theta generators in the hippocampus: a partial coherence analysis.
\newblock {\em Journal of Neuroscience}, 19(14):6200--6212.

\bibitem[Koutitonsky et~al., 2002]{koutitonsky2002descriptive}
Koutitonsky, V., Navarro, N., and Booth, D. (2002).
\newblock Descriptive physical oceanography of great-entry lagoon, gulf of st. lawrence.
\newblock {\em Estuarine, Coastal and Shelf Science}, 54(5):833--847.

\bibitem[Lambert et~al., 2018]{lambert2018reconstructing}
Lambert, R.~C., Tuleau-Malot, C., Bessaih, T., Rivoirard, V., Bouret, Y., Leresche, N., and Reynaud-Bouret, P. (2018).
\newblock Reconstructing the functional connectivity of multiple spike trains using hawkes models.
\newblock {\em Journal of Neuroscience Methods}, 297:9--21.

\bibitem[Meinshausen and B{\"u}hlmann, 2006]{meinshausen2006high}
Meinshausen, N. and B{\"u}hlmann, P. (2006).
\newblock High-dimensional graphs and variable selection with the lasso.
\newblock {\em Ann. Statist.}, 34(1):1436--1462.

\bibitem[Nadkarni et~al., 2016]{nadkarnisparse}
Nadkarni, R., Foti, N.~J., Lee, A.~K., and Fox, E.~B. (2016).
\newblock Sparse plus low-rank graphical models of time series to infer functional connectivity from {MEG} recordings.
\newblock {\em 2nd SIGKDD Workshop on Mining and Learning from Time Series}.

\bibitem[Negahban et~al., 2012]{negahban2009unified}
Negahban, S., Yu, B., Wainwright, M.~J., and Ravikumar, P. (2012).
\newblock A unified framework for high-dimensional analysis of ${M} $-estimators with decomposable regularizers.
\newblock {\em Statistical Sciecne}, 27(4):538--557.

\bibitem[Ogata, 1981]{ogata1981lewis}
Ogata, Y. (1981).
\newblock On {L}ewis' simulation method for point processes.
\newblock {\em IEEE Transactions on Information Theory}, 27(1):23--31.

\bibitem[Pernice et~al., 2011]{pernice2011structure}
Pernice, V., Staude, B., Cardanobile, S., and Rotter, S. (2011).
\newblock How structure determines correlations in neuronal networks.
\newblock {\em PLOS Computational Biology}, 7(5):e1002059.

\bibitem[Ravikumar et~al., 2011]{ravikumar2011high}
Ravikumar, P., Wainwright, M.~J., Raskutti, G., and Yu, B. (2011).
\newblock High-dimensional covariance estimation by minimizing $\ell_1$-penalized log-determinant divergence.
\newblock {\em Electronic Journal of Statistics}, 5:935--980.

\bibitem[Rothman et~al., 2008]{rothman2008sparse}
Rothman, A.~J., Bickel, P.~J., Levina, E., and Zhu, J. (2008).
\newblock Sparse permutation invariant covariance estimation.
\newblock {\em Electronic Journal of Statistics}, 2:494--515.

\bibitem[Roueff and Sachs, 2019]{Roueff2019}
Roueff, F. and Sachs, R.~V. (2019).
\newblock Time-frequency analysis of locally stationary hawkes processes.
\newblock {\em Bernoulli}, 25:1355--1385.

\bibitem[Song et~al., 2020]{song2020potential}
Song, X., Zhang, C., Zhang, J., Zou, X., Mo, Y., and Tian, Y. (2020).
\newblock Potential linkages of precipitation extremes in {B}eijing-{T}ianjin-{H}ebei region, {C}hina, with large-scale climate patterns using wavelet-based approaches.
\newblock {\em Theoretical and Applied Climatology}, 141:1251--1269.

\bibitem[Tank et~al., 2015]{tank2015bayesian}
Tank, A., Foti, N., and Fox, E. (2015).
\newblock Bayesian structure learning for stationary time series.
\newblock {\em arXiv preprint arXiv:1505.03131}.

\bibitem[Tugnait, 2022]{tugnait2021sparse}
Tugnait, J.~K. (2022).
\newblock On sparse high-dimensional graphical model learning for dependent time series.
\newblock {\em Signal Processing}, 197.

\bibitem[Wainwright, 2019]{wainwright2019high}
Wainwright, M.~J. (2019).
\newblock {\em High-{D}imensional {S}tatistics: {A} {N}on-{A}symptotic {V}iewpoint}.
\newblock Cambridge University Press.

\bibitem[Walden, 2000]{Walden2000}
Walden, A.~T. (2000).
\newblock A unified view of multitaper multivariate spectral estimation.
\newblock {\em Biometrika}, 87:767--788.

\bibitem[Wang and Shojaie, 2021]{wang2021joint}
Wang, X. and Shojaie, A. (2021).
\newblock Joint estimation and inference for multi-experiment networks of high-dimensional point processes.
\newblock {\em arXiv preprint arXiv:2109.11634}.

\bibitem[Warton, 2008]{warton2008penalized}
Warton, D.~I. (2008).
\newblock Penalized normal likelihood and ridge regularization of correlation and covariance matrices.
\newblock {\em Journal of the American Statistical Association}, 103(481):340--349.

\bibitem[Whittle, 1953]{whittle1953estimation}
Whittle, P. (1953).
\newblock Estimation and information in stationary time series.
\newblock {\em Arkiv f{\"o}r matematik}, 2(5):423--434.

\end{thebibliography}

\newpage
\renewcommand{\thefigure}{S\arabic{figure}}
\renewcommand{\thetable}{S\arabic{table}}
\renewcommand{\theequation}{S.\arabic{equation}}

\setcounter{page}{1}
\setcounter{figure}{0}
\setcounter{table}{0}
\setcounter{equation}{0}
\setcounter{section}{0}

\title{\huge{Supplementary Material for `Regularised Spectral Estimation for High Dimensional Point Processes'}}
\author{Carla Pinkney$^{1,^*}$, Carolina Eu\'{a}n$^1$, Alex Gibberd$^1$ \& Ali Shojaie$^2$ \vspace{0.3cm} \\
    \small{$^1$ STOR-i Centre for Doctoral Training, Department of Mathematics and Statistics, Lancaster} \\ \small{University, LA1 4YR, UK}\\ $^2$ \small{Department of Biostatistics, University of Washington, Seattle, WA 98105, US} \vspace{0.1cm} \\
\footnotesize{$^*$ Correspondence to: c.pinkney@lancaster.ac.uk} \vspace{0.2cm}
\\}
\maketitle
\setcounter{section}{0}
\section{Introduction}
This document outlines the supplementary material for ``Regularised Spectral Estimation for High Dimensional Point Processes'', subsequently referred to as the main paper. Firstly, we give proofs of all the results given in the main paper. Secondly we present additional results for the experimental procedure outlined in Section 2 of the main paper. In Section 4, we give a full description of the ADMM algorithm used for the regularised spectral estimator (RSE) with the Lasso penalty.  Then in Section 5, we give additional details of the multivariate Hawkes process used in the synthetic experiments. Finally, in section 6 we provide further analysis of the \cite{bolding2018recurrent} dataset and present partial coherence networks obtained for the the theta band.

\section{Proofs}
\subsection{Proof of Proposition \ref{prop:dev_bound}}
\begin{proof}
\label{proof:prop_A2}
  
In the proof of Proposition \ref{prop:dev_bound} we require the following lemma based on the work of \cite{brillinger1972}. We state the result and prove it here for completeness. 
\begin{lemma}{\cite{brillinger1972}} 
\label{lemma:Brillinger}
Under Assumption \ref{ass:non-overlap}, we consider the set of $k=1, \dots, m$ tapers and let $T^{\prime}=T/m$. Now, let $f_{T^{\prime}}$ be an integer sequence with $\omega_{T^{\prime}} = 2\pi f_{T^{\prime}}/T^{\prime} \rightarrow \omega$ as $T^{\prime} \rightarrow \infty.$ The tapered and mean corrected Fourier coefficients are distributed as 
\begin{equation*}
    \mathbf{\bar{d}}_{k}(\omega_{T^{\prime}}) \sim \mathcal{N}_C (0, \mathbf{S}(\omega)) \ \textrm{as} \ T^{\prime}\rightarrow\infty.
\end{equation*}
\end{lemma}

\begin{proof}

Consider first the form of the mean corrected tapered Fourier transform defined in \cite{brillinger1972} for a particular frequency $\omega$ as
\begin{equation*}
    \bar{\mathbf{d}}_{k}(\omega) = \mathbf{d}_{k}(\omega) - \underbrace{\frac{\mathbf{d}_{k}(0)H_k(T\omega)}{H_k(0)}}_{\mathbf{b}(\omega)},
\end{equation*}
where we refer to $\mathbf{b}(\omega)$ as the so called ``mean correction" term. 

We proceed by considering the mean and variance of $\mathbf{b}(\omega)$ before an application of Theorem 4.2 \cite{brillinger1972} to complete the proof. However, we first state the following lemma.

\begin{lemma} The scaled Fourier transform of the taper $h_k(t/T)=(2\pi T^{\prime})^{-1/2}$ is 
\begin{equation*}
    H_k(T\omega ) = \frac{(2\pi T^{\prime})^{-1/2}}{m} \textrm{sinc}\left\{\frac{\omega T^{\prime}}{2}\right\} \exp\{-i\omega T^{\prime}(k-1/2)\}
\end{equation*}
and $H_k(0) = \frac{(2\pi T^{\prime})^{-1/2}}{m}.$
\label{lem:FT_taper}
\end{lemma}
\begin{proof} We begin by defining the taper
\begin{equation*}
    h_k(t/T) = \begin{cases}
        (2\pi T^{\prime})^{-1/2} \ & (k-1)T^{\prime}<t<kT^{\prime} \\
        0 \ &\textrm{otherwise},
    \end{cases}
\end{equation*}
and its Fourier transform as 
\begin{align*} &\int h_k(t/T)\exp\{-i\omega t\} \ dt \\ & \int h_k(t')\exp\{-i\omega t' T\}) \ dt' T \\ &=H_k(T \omega ) T. \end{align*}

We might need to consider phase shift for $H_1(\omega)$ compared to $H_2(\omega)$.

Explicitly, we have that
\begin{align*}
    H_k(T\omega)T&=\int_{(k-1)T^{\prime}}^{kT^{\prime}}  (2\pi T^{\prime})^{-1/2} e^{-i \omega t} d t \\
   &=  (2\pi T^{\prime})^{-1/2} \left[\frac{\exp\{-i\omega t\}}{-i\omega}\right]_{(k-1)T^{\prime}}^{kT^{\prime}} \\
   &= (2\pi T^{\prime})^{-1/2} \exp\{-i\omega k T^{\prime}\}\left[\frac{\exp\{i\omega T^{\prime}\}-1}{iw}\right] \\
   &= (2\pi T^{\prime})^{-1/2} \exp\{-i\omega k T^{\prime}\} \frac{2\exp\{i\omega T^{\prime}/2\}}{\omega} \left(\frac{\exp\{i\omega T^{\prime}/2\} - \exp\{-i\omega T^{\prime}/2\}}{2i}\right) \\
   &= (2\pi T^{\prime})^{-1/2} T^{\prime} \textrm{sinc}\left\{\frac{\omega T^{\prime}}{2}\right\} \exp\{-i\omega T^{\prime}(k-1/2)\}
\end{align*}
using the identity that $\sin (x)=\frac{e^{i x}-e^{-i x}}{2 i}.$ 

It is therefore immediate that $$H_k(T\omega ) =  \frac{(2\pi T^{\prime})^{-1/2}}{m} \textrm{sinc}\left\{\frac{\omega T^{\prime}}{2}\right\} \exp\{-i\omega T^{\prime}(k-1/2)\}.$$
Further, we have that $H_k(0) = \frac{1}{T}\int_{(k-1)T^{\prime}}^{kT^{\prime}}  h(t/T) \ dt = \frac{(2\pi T^{\prime})^{-1/2}}{m}.$
\end{proof}

Using Lemma \ref{lem:FT_taper} we compute the mean of $\mathbf{b}(\omega)$ as
\begin{align*}
    \mathbb{E}\{\mathbf{b}(\omega)\} &= \frac{H_k(T\omega)}{H_k(0)}\mathbb{E}\left\{ \mathbf{d}_k(0)\right\}\\
    &= \textrm{sinc}\left\{\frac{\omega T^{\prime}}{2}\right\}\exp\{-i\omega T^{\prime}(k-1/2)\} \mathbb\{\mathbf{d}_k(0)\}.
\end{align*}
At the sampled frequencies  $\omega_{T^{\prime}} = 2\pi f_{T^{\prime}}/T^{\prime}$ we have that $\textrm{sinc}\left\{ \frac{\omega T^{\prime}}{2}\right\} = \frac{\sin\{\pi f_{T^{\prime}}\}}{\pi f_{T^{\prime}}} = 0$ and thus $\mathbb{E}\{\mathbf{b}(\omega)\}=0.$

The variance of $\mathbf{b}(\omega)$ is 
\begin{align*}
    \mathbb{V}\{\mathbf{b}(\omega)\} &= \left\{\frac{H_k(T\omega)}{H_k(0)}\right\}^2 \mathbb{V}\{\mathbf{d}_{k}(\omega)\} \\
    &= \left\{\textrm{sinc}\left\{\frac{\omega T^{\prime}}{2}\right\}\exp\{-i\omega T^{\prime}(k-1/2)\}\right\}^2 \mathbb{V}\{\mathbf{d}_k(0)\}. 
\end{align*}
As above, we have that $\textrm{sinc}\left\{\frac{\omega T^{\prime}}{2}\right\}=0$ at the sample frequencies $\omega_{T^{\prime}}$. Therefore, $\mathbb{V}\{\mathbf{b}(\omega)\}=0.$

Thus at the sampled frequencies, Theorem 4.2 \cite{brillinger1972} also applies to the mean corrected tapered Fourier transform. Asymptotically, we have that
\begin{equation*}
    \bar{\mathbf{d}}_k(\omega)\sim \mathcal{N}_C({0}, \mathbf{S}(\omega)),
\end{equation*}
and $\bar{\mathbf{d}}_k(0)={0}$.
\end{proof}

We begin the proof by recalling the form of the trial averaged periodogram
\begin{equation*}
    \hat{\mathbf{S}}(\omega) = \frac{1}{m}\sum_{k=1}^m\bar{\mathbf{d}}_k(\omega) \bar{\mathbf{d}}_k^H(\omega).
\end{equation*}
We know from Lemma \ref{lemma:Brillinger} that $\bar{\mathbf{d}}_k(\omega)\sim \mathcal{N}_C(\mathbf{0}, \mathbf{S}(\omega))$ and therefore
\begin{equation*}
    \left(\begin{array}{c}
\operatorname{Re}(\mathbf{d}_{k}(\omega) \\
\operatorname{Im}(\mathbf{d}_k(\omega)  
\end{array}\right) \sim \mathcal{N}_{2 p}\left[\left(\begin{array}{l}
0 \\
0
\end{array}\right), \frac{1}{2}\left(\begin{array}{cc}
\operatorname{Re}(\mathbf{S}(\omega)) & -\operatorname{Im}(\mathbf{S}(\omega)) \\
\operatorname{Im}(\mathbf{S}(\omega)) & \operatorname{Re}(\mathbf{S}(\omega))
\end{array}\right)\right].
\end{equation*}
As such, we can apply Lemma 1 from \cite{ravikumar2011high} to the real and imaginary parts of the spectrum separately in order to bound the probability 
\begin{equation*}
    \mathbb{P}\{|\hat{S}_{qr}(\omega)-S_{qr}(\omega)|\geq\delta\}.
\end{equation*}
Specifically, we note that $\mathbb{P}\{|\hat{S}_{qr}(\omega)-S_{qr}(\omega)|\geq\delta\}$ is at most 
\begin{equation}
    \mathbb{P}\left(|\operatorname{Re}(\hat{S}_{qr}(\omega)-S_{qr}(\omega))|\geq\delta/2\right) +  \mathbb{P}\left(|\operatorname{Im}(\hat{S}_{qr}(\omega)-S_{qr}(\omega))|\geq\delta/2\right).
\end{equation}

Applying Lemma 1 \cite{ravikumar2011high} to each part yields
\begin{align*}
     \mathbb{P}\{|\hat{S}_{qr}(\omega)-S_{qr}(\omega)|\geq\delta\} &\leq 4\exp\left\{\frac{-m (\delta/2)^2}{128(1+4\sigma^2)^2 \max_{q}(\operatorname{Re}(S_{qq}(\omega))^2)}\right\} \\ &+4\exp\left\{\frac{-m (\delta/2)^2}{128(1+4\sigma^2)^2 \max_{q}(\operatorname{Re}(S_{qq}(\omega))^2)}\right\} \\
        &= 8\exp\left\{\frac{-m \delta^2}{2^9 (1+4\sigma^2)^2 \max_{q}\{S_{qq}(\omega)\}^2}\right\}
\end{align*}

for all $\delta \in (0, 16(\max_{q}\{S_{qq}(\omega)\})(1+4\sigma^2) ).$

Noting that any multivariate Gaussian random vector is sub-Gaussian, and in particular the rescaled variates are sub-Gaussian with parameter $\sigma=1$, we arrive at the final result

\begin{align*}
     \mathbb{P}\{|\hat{S}_{qr}(\omega)-S_{qr}(\omega)|\geq\delta\}   
     &\leq 8\exp\left\{\frac{-m \delta^2}{2^9 5^2 \max_{q}\{S_{qq}(\omega)\}^2}\right\}
\end{align*}

for all $\delta \in (0, 80(\max_{q}\{S_{qq}(\omega)\})).$

\end{proof}


\subsection{Proof of Proposition \ref{prop:ridge}}
\label{proof:prop_A4}
\begin{proof}
In the proof, we will need the following lemma which is a direct consequence of the deviation bound given in Proposition \ref{prop:dev_bound}. 
\begin{lemma}
\label{lemma:union}
    For any $q,r=1,\dots, p$ and $\tau>2$ there exists a frequency specific constant $c$ such that 
    \begin{equation}
        \Pr\left(\max_{qr}|\hat{{S}}_{qr}(\omega)-S_{qr}(\omega)|\geq c\sqrt{\frac{\log p ^\tau}{m}}\right) \leq \frac{8}{p^{\tau-2}} \rightarrow 0 \ \textrm{as} \ p \rightarrow \infty.
    \end{equation}
\end{lemma}


\label{proof:prop_A5}
\begin{proof}
   Recall the tail bound given in \eqref{eq:dev_bound}. 
   Applying a union bound argument yields
  
    \begin{align*}
       \Pr\left(\max_{qr}|\hat{S}_{qr}(\omega)-S_{qr}(\omega)|\geq c\delta\right) \leq 8 p^2 \exp\left\{\frac{-mc^2\delta^2}{512(1+4\sigma^2)^2\max_q\{S_{qq}(\omega)\}^2}\right\}, 
   \end{align*}
   where we define
   \begin{equation*}
   \label{eq:c}
      c = \sqrt{512} (1+4\sigma^2)\max_q\{S_{qq}(\omega)\}.
   \end{equation*}

Hence, 
\begin{align*}
    \Pr\left(\max_{qr}|\hat{S}_{qr}(\omega) - S_{qr}(\omega)| \geq c \delta \right) &\leq 8p^2\exp\{-\delta^2 m\} \\
    &= 8\exp\{\log p^2-\delta^2m\}
\end{align*}
Letting $\delta = \sqrt{\frac{\log  p^\tau}{m}}$ yields
\begin{align}
    \Pr\left(\max_{qr}|\hat{S}_{qr}(\omega) - S_{qr}(\omega)|\geq c \delta \right) &\leq 8 \exp\{\log p^2-\tau \log p \} \nonumber\\
    &= 8\exp\{\log p^2-\log  p^{\tau-2}-\log p^2 \} \nonumber \\
    &= 8\exp\{-\log p^{\tau-2}\} \nonumber\\
    &= \frac{8}{p^{\tau-2}} \rightarrow 0 \ \textrm{as} \ p\rightarrow\infty \nonumber.
\end{align}

\end{proof}

We begin the proof of Proposition \ref{prop:glasso} by defining the event 
\begin{equation}
\label{eq:event}
    \mathcal{E}:=\left\{\|\hat{S}_{qr}(\omega) - S_{qr}(\omega)\|_{\infty} {\leq} c \sqrt{\frac{\log p^{\tau}}{m}}\right\}.
\end{equation}
We have shown in Lemma \ref{lemma:union} that $\Pr(\mathcal{E}^c)\leq 8/p^{\tau-2}$, and therefore $\Pr(\mathcal{E})\geq 1-8/p^{\tau-2}.$ We condition on the event $\mathcal{E}$ in the following analysis. 

Going forward, we acknowledge a change in notation for the true inverse spectral density matrix, which we will now denote by  $\boldsymbol{\Theta}_0$ rather than $\boldsymbol{\Theta}^*$. We will also drop the dependence on $\omega$ for ease of notation. 

Consider the function
    \begin{align}
    \label{eq:F_function}
        F(\boldsymbol{\hat{\Theta}}):&=\mathcal{L}(\boldsymbol{\hat{\Theta}})-\mathcal{L}(\boldsymbol{\Theta}_0)+\lambda(P_{2}\left\{\hat{\boldsymbol{\Theta}}\right\} - P_{2}\left\{\boldsymbol{\Theta}_0\right\}) \nonumber\\
        &=\underbrace{\Tr\{(\boldsymbol{\hat{\Theta}}-\boldsymbol{\Theta}_0)(\hat{\mathbf{S}}-\mathbf{S})\}}_{(i)}
        - \underbrace{(\log\operatorname{det}(\boldsymbol{\hat{\Theta}})-\log\operatorname{det}(\boldsymbol{\Theta}_0))}_{(ii)} \nonumber\\
        &+ \underbrace{\Tr\{(\boldsymbol{\hat{\Theta}}-\boldsymbol{\Theta}_0)\mathbf{S}\}}_{(iii)} 
        + \underbrace{\lambda(\Tr\left\{\hat{\boldsymbol{\Theta}}\right\} - \Tr\left\{\boldsymbol{\Theta}_0\right\})}_{(iv)}
    \end{align}
  where $\boldsymbol{\Theta}_0=\mathbf{S}(\omega)^{-1}.$

Let $\boldsymbol{\hat{\Theta}}=\boldsymbol{\Theta}_0+\Delta$. Our aim is to bound the size of $\Delta$ such that $F(\boldsymbol{\hat{\Theta}})\leq F(\boldsymbol{\Theta}_0)=0.$ 
Since the ridge estimator $\boldsymbol{\hat{\Theta}}$ is the solution to a minimisation problem, we know that it will always attain a lower loss than the true value $\boldsymbol{\Theta}_0.$

The main idea of the proof is based on an adaptation of the argument presented in \cite{rothman2008sparse} where one can consider the shell of errors
\begin{equation}
\label{eq:shell}
    \mathcal{B}_m(G)=\{\Delta \ | \ \Delta=\Delta^H, \ \|\Delta\|_F=Gr_m\},
\end{equation}
i.e. all possible symmetric errors of a given radius e.g. $Gr_m, G>0.$ We can then adjust the radius $r_m$ such that $$F(\boldsymbol{\Theta}_0+\Delta)>0,$$ and when this occurs we know that our shell is too big, and therefore the error must lie within $\mathcal{B}_m(G)$ and thus $\|\Delta\|_F\leq Gr_m$.

We now consider each term in \eqref{eq:F_function} separately letting $\Delta=\boldsymbol{\hat{\Theta}}-\boldsymbol{\Theta}_0$. We start by simplifying term $(ii)$ in \eqref{eq:F_function}. 

\begin{lemma}
\label{lemma:Taylor}
    Taking the first order Taylor expansion of $f(t)=\log\operatorname{det}(\boldsymbol{\Theta}_0+t\Delta)$ with the integral form of the remainder yields 
    \begin{align*}
        \log\operatorname{det}(\boldsymbol{\Theta}_0+\Delta)&-\log\operatorname{det}(\boldsymbol{\Theta}_0) = \\  &\Tr\{\mathbf{S}\Delta\} - \vec{\Delta}^H\left[ \int_0^1(1-v)((\boldsymbol{\Theta}_0+t\Delta)^{-1})^{{H}} \otimes (\boldsymbol{\Theta}_0+t\Delta)^{-1})dv \right]\vec{\Delta}
    \end{align*}
    where $\otimes$ denotes the Kronecker product.
\end{lemma}

\begin{proof}
    The proof of Lemma \ref{lemma:Taylor} begins by computing the derivatives
    \begin{align*}
        \frac{\partial f(t)}{\partial t} &= \frac{1}{\operatorname{det}(\Theta_0+t\Delta)}\left\{\frac{\partial}{\partial t} \operatorname{det}(\Theta_0+t\Delta)\right\} \\
        &= \frac{1}{\operatorname{det}(\Theta_0+t\Delta)} \operatorname{det}(\Theta_0+t\Delta) \Tr\left\{(\Theta_0+t\Delta)^{-1} \frac{\partial}{\partial t} (\Theta_0+t\Delta)\right\} \ \textrm{(Jacobi's formula)}\\
        &= \Tr\{(\Theta_0+t\Delta)^{-1} \Delta\}.
    \end{align*}

Computing the second derivative 
\begin{align*}
    \frac{\partial ^2 f(t)}{\partial t^2} = \frac{\partial}{\partial t} \left(\Tr\{(\Theta_0+t\Delta)^{-1} \Delta\}\right) =   \frac{\partial}{\partial t} \underbrace{\vect (\Delta^H)}_{f_1(t)}\underbrace{\vect((\boldsymbol{\Theta}_0+t\Delta)^{-1})}_{f_2(t)}
\end{align*}
where by the product rule 
\begin{align*}
     \frac{\partial ^2 f(t)}{\partial t^2} = f_2(t)\frac{\partial f_1(t)}{\partial t} + f_1(t) \frac{\partial f_2(t)}{\partial t}.
\end{align*}
Since $\frac{\partial f_1(t)}{\partial t} = 0$ we proceed by considering 
\begin{align*}
    \frac{\partial f_2(t)}{\partial t} &= \frac{\partial}{\partial t} \vect((\boldsymbol{\Theta}_0+t\Delta)^{-1}) \\
    &= \vect\left(\frac{\partial}{\partial t} (\boldsymbol{\Theta}_0+t\Delta)^{-1}\right)  \\
    &= \vect(-(\boldsymbol{\Theta}_0+t\Delta)^{-1} \Delta (\boldsymbol{\Theta}_0+t\Delta)^{-1}) \\ 
    &= -((\boldsymbol{\Theta}_0+t\Delta)^{-1})^H\otimes (\boldsymbol{\Theta}_0+t\Delta)^{-1} \vect(\Delta) \ 
\end{align*}

Hence,
\begin{align*}
              \frac{\partial ^2 f(t)}{\partial t^2} &= -\vec{\Delta}^H[((\boldsymbol{\Theta}_0+t\Delta)^{-1})^{H} \otimes(\boldsymbol{\Theta}_0+t\Delta)^{-1}]\vec{\Delta}.
\end{align*}

The first order Taylor expansion with integral remainder term is defined as
\begin{align*}
    f(t) = f(a) +\frac{\partial f(a)}{\partial t}(t-a) + \int^t_a \frac{\partial ^2 f(a)}{\partial t} (t-v) dv.
\end{align*}

Applying the above Taylor expansion with $t=1,a=0$ and re-arranging yields, 
\begin{align*}
        \log\operatorname{det}(\boldsymbol{\Theta}_0+\Delta)&-\log\operatorname{det}(\boldsymbol{\Theta}_0) = \\  &\Tr\{\mathbf{S}\Delta\} - \vec{\Delta}^H\left[ \int_0^1(1-v)((\boldsymbol{\Theta}_0+v\Delta)^{-1})^{H} \otimes (\boldsymbol{\Theta}_0+v\Delta)^{-1})dv \right]\vec{\Delta}
\end{align*}

\end{proof}
Therefore, we can write \eqref{eq:F_function} as 
\begin{align}
\label{eq:f_simpl}
     F(\boldsymbol{\Theta}_0+\Delta) &= \underbrace{\Tr\{\Delta (\hat{\mathbf{S}}-\mathbf{S})\}}_{(i)} + \underbrace{\vec{\Delta}^H\left[ \int_0^1(1-v)((\boldsymbol{\Theta}_0+v\Delta)^{-1})^{{H}} \otimes (\boldsymbol{\Theta}_0+v\Delta)^{-1})dv \right]\vec{\Delta} \nonumber}_{(ii)} \\ 
    & + \underbrace{\lambda(\Tr\left\{\hat{\boldsymbol{\Theta}}\right\} - \Tr\left\{\boldsymbol{\Theta}_0\right\})}_{(iii)}.
\end{align}

Noting that $$(iii) =  \lambda(\Tr\{\boldsymbol{\Theta}_0+\Delta\} - \Tr\{\boldsymbol{\Theta}_0\}) = \lambda (\Tr\{\Delta\}),$$ we can rewrite \eqref{eq:f_simpl} as
\begin{align}
    \label{eq:F_function_ridge}
   F(\boldsymbol{\Theta}_0+\Delta) &= \underbrace{\Tr\{\Delta (\hat{\mathbf{S}}-\mathbf{S} + \lambda \mathbb{I})\}}_{(i)} + \underbrace{\vec{\Delta}^H\left[ \int_0^1(1-v)((\boldsymbol{\Theta}_0+v\Delta)^{-1})^{{}{H}} \otimes (\boldsymbol{\Theta}_0+v\Delta)^{-1})dv \right]\vec{\Delta},}_{(ii)} 
    \end{align}

and proceed by bounding each term separately.

We begin by writing
\begin{align*}
     (i) \leq |\Tr\{\Delta (\hat{\mathbf{S}}-\mathbf{S} + \lambda \mathbb{I})\}| &\leq \left|\Tr\{\Delta(\mathbf{\hat{S}} - \mathbf{S})\}\right| + \left|\lambda \Tr \{\Delta \mathbb{I}\}\right|.\\
&\leq  \|\hat{\mathbf{S}}-\mathbf{S}\|_F \|\Delta\|_F + \lambda \|\Delta\|_F \|\mathbb{I}\|_F \\
&\leq p\|\mathbf{\hat{S}} - \mathbf{S}\|_{\infty}\|\Delta\|_F + \lambda \sqrt{p}\|\Delta\|_F\\
&\leq \left(c \sqrt{\frac{p^2\log p^\tau}{m}}+\lambda \sqrt{p}\right)\|\Delta\|_F.
\end{align*}

To bound the integral term $(ii)$ in \eqref{eq:F_function_ridge}, we recall that for a symmetric matrix $M$ we have that the minimum eigenvalue is given by $\phi_{\min}(M) = \min_{\|x\|=1}x^HMx$. Hence, after factoring out the norm of $\vec{\Delta}$, we have for $\Delta \in \mathcal{B}_m(G)$
\begin{align*}
    \phi_{\min}&\left(\int_0^1(1-v)((\boldsymbol{\Theta}_0+v\Delta)^{-1})^{H} \otimes (\boldsymbol{\Theta}_0+v\Delta)^{-1})dv \right)\\
    &\geq \int_0^1 (1-v)\phi^2_{min}(\boldsymbol{\Theta}_0+v\Delta)^{-1} dv \\  
    &\geq \frac{1}{2} \min_{0\leq v \leq 1} \phi^2_{\min}(\boldsymbol{\Theta}_0+v\Delta)^{-1} \\
    &\geq \frac{1}{2}\min\{\phi^2_{\min}(\boldsymbol{\Theta}_0+\Delta)^{-1}: \|\Delta\|_F \leq Gr_m\}.
\end{align*}
The first inequality uses the property that eigenvalues of Kronecker products of Hermitian matrices are equal to the products of the eigenvalues of their factors, i.e., $\textrm{eig}(A \otimes B) = eig(A)eig(B).$

\begin{claim}
\label{claim}
   $\phi^2_{\min}(\boldsymbol{\Theta}_0+\Delta)^{-1} \geq \frac{1}{9} \underline{k}^2 \ \forall \ \|\Delta\|_F\leq Gr_{m}$
\end{claim}
\begin{proof}
\begin{align*}
    \phi^2_{\min}(\boldsymbol{\Theta}_0+\Delta)^{-1} &= \phi^{-2}_{\max}(\boldsymbol{\Theta}_0+\Delta) \geq(\|\boldsymbol{\Theta}_0\|_2+\|\Delta\|_2)^{-2} 
\end{align*}
Now,
\begin{align*}
    \|\Delta\|_2 = \|\boldsymbol{\hat{\Theta}}-\boldsymbol{\Theta}_0\|_2 &\leq \|\hat{\boldsymbol{\Theta}}\|_2 + \|\boldsymbol{\Theta}_0\|_2 \\
    &= \|(\hat{\boldsymbol{S}}+\lambda \mathbb{I})^{-1}\|_2 + \|\boldsymbol{\Theta}_0\|_2 \\
    &\leq \frac{1}{\lambda} + \|\boldsymbol{\Theta}_0\|_2 \\
    &\leq 2\|\boldsymbol{\Theta}_0\|_2
\end{align*}
provided that $ 1/\lambda \leq \|\boldsymbol{\Theta}_0\| \iff \lambda \geq \bar{k},$ which is satisfied through our choice of $\lambda = \frac{c}{\epsilon}\sqrt{\frac{p\log p^\tau}{m}}$ and constraints on $m$.

Hence, 
\begin{align*}
    (\|\boldsymbol{\Theta}_0\|_2+ 2\|\boldsymbol{\Theta}_0\|_2)^{-2} \geq (3\|\boldsymbol{\Theta}_0\|_2)^{-2}=\frac{1}{9}\underline{k}^2
\end{align*}
where $\underline{k}$ is the minimum eigenvalue of $\mathbf{S}.$
\end{proof}

Thus, 
\begin{equation*}
       \phi_{\min}\left(\int_0^1(1-v)((\boldsymbol{\Theta}_0+v\Delta)^{-1})^{H} \otimes (\boldsymbol{\Theta}_0+v\Delta)^{-1})dv \right) \geq \frac{1}{18}\underline{k}^2,
\end{equation*}

and we have the lower bound
\begin{align*}
   F(\boldsymbol{\Theta}_0+ \Delta) &\geq \frac{1}{18}\underline{k}^2 \|\Delta\|^2_F - c\left(1+\frac{1}{\epsilon}\right)\sqrt{\frac{p^2\log p}{m}}\|\Delta\|_F .
\end{align*}

Solving this quadratic in $\|\Delta\|_F$ yields one positive root 
\begin{align*}
    \|\Delta\|_F &= \frac{18c\left(1+\frac{1}{\epsilon}\right)}{\underline{k}^2}\sqrt{\frac{p^2 \log p^\tau}{m}},
\end{align*} 
which equating with $Gr_m$ gives the stated error bound. Since the above analysis was conditioned on the event $\mathcal{E}$, Proposition \ref{prop:ridge} holds with probability $1-\frac{8}{p^{\tau-2}}\rightarrow 1,$ for any $\tau>2$ and sufficient tapers as stated.
\end{proof}

\subsection{Proof of Proposition \ref{prop:glasso}}

\begin{proof}
In the proof, we adopt the notation $M^+ = \operatorname{diag}(M)$ for a diagonal matrix with the same diagonal as matrix $M$, and $M^- = M - M^+.$

In the same manner as the proof of Proposition \ref{prop:ridge}, we condition on the event $\mathcal{E}$ and seek to find the value $Gr_m$ which defines the shell of errors \eqref{eq:shell}. The main difference is that of the regularisation term, therefore \eqref{eq:f_simpl} becomes
\begin{align}
    \label{eq:F_function_glasso}
   F(\boldsymbol{\Theta}_0+\Delta) &= \underbrace{\Tr\{\Delta (\hat{\mathbf{S}}-\mathbf{S})\}}_{(i)} + \underbrace{\vec{\Delta}^H\left[ \int_0^1(1-v)((\boldsymbol{\Theta}_0+v\Delta)^{-1})^{{}{H}} \otimes (\boldsymbol{\Theta}_0+v\Delta)^{-1})dv \right]\vec{\Delta} \nonumber}_{(ii)} \\ 
        &+ \underbrace{\lambda (\|(\boldsymbol{\Theta}_0+\Delta)\|_1 - \|\boldsymbol{\Theta}_0\|_1)}_{(iv)}
    \end{align}
    where $\boldsymbol{\Theta}_0=\mathbf{S}^{-1}(\omega).$

We consider bounding each term in \eqref{eq:F_function_glasso} separately, and begin by writing

\begin{align*}
    (i) \leq |\Tr\{\Delta (\hat{\mathbf{S}}-\mathbf{S})\}| \leq \underbrace{\left|\sum_{q\neq r}(\hat{S}_{qr}-S_{qr})\Delta_{qr}\right|}_{\textrm{I}} + \underbrace{\left|\sum_{q}(\hat{S}_{qq}-S_{qq})\Delta_{qq}\right|}_{\textrm{II}} = \textrm{I} + \textrm{II}.
\end{align*}

To bound I, we appeal to our earlier result in Lemma \ref{lemma:union} and argue that with probability tending to 1, 
\begin{align*}
    \max_{q\neq r}|\hat{S}_{qr}-S_{qr}|\leq c \sqrt{\frac{\log p^\tau}{m}}
\end{align*}
therefore term I is bounded by
\begin{align*}
    \textrm{I} \leq  c \sqrt{\frac{\log p^\tau}{m}} \|\Delta^{-}\|_1.
\end{align*}

To bound the second term, we consider an application of the Cauchy-Schwartz inequality combined with our result from Lemma \ref{lemma:union}, namely, 
\begin{align*}
    \textrm{II}\leq \|\hat{S}_{qq}-S_{qq}\|_F \ \|\Delta^+\|_F &= \sqrt{\sum_{q=1}^p (\hat{S}_{qq}-S_{qq})^2} \ \|\Delta^+\|_F \\
    &\leq \sqrt{p}\|\hat{S}_{qq}-S_{qq}\|_{\infty} \ \|\Delta^+\|_F \\
    &\leq c \sqrt{\frac{p \log p^\tau}{m}}\|\Delta^+\|_F \\
    &\leq c\sqrt{\frac{(p + s)\log p^\tau}{m}}\|\Delta^+\|_F.
\end{align*}

To bound the penalty term $(iii)$ we apply the decomposability argument of \cite{negahban2009unified}, arguing that the norm is linearly decomposable over the on and off support of the true model. 

Importantly, we note the applicability of this argument to the group regulariser $\|\boldsymbol{\Theta}\|_1$ defined for a complex valued matrix $\boldsymbol{\Theta} \in \mathbb{C}^{p \times p}$ since the group norm is specified at the level of elements in $\mathcal{M}.$ 

Hence, by this decomposability argument we have that
\begin{align*}
    &\lambda(\|\boldsymbol{\Theta}_0 + \Delta\|_1 - \|\boldsymbol{\Theta}_0\|_1) \\
    =  &\lambda(\|\boldsymbol{\Theta}^{-}_{0\mathcal{M}} + \boldsymbol{\Theta}^{-}_{0\mathcal{M}^{\perp}} +  \boldsymbol{\Theta}^{+}_{0\mathcal{M}} + \Delta^{-}_{\mathcal{M}} + \Delta^{-}_{\mathcal{M}^{\perp}} + \Delta^{+}_{\mathcal{M}}\|_1 - \|\boldsymbol{\Theta}^{-}_{0\mathcal{M}} + \boldsymbol{\Theta}^{-}_{0\mathcal{M}^{\perp}}+\boldsymbol{\Theta}^{+}_{0\mathcal{M}} \|_1) \\
    = &\lambda(\underbrace{\|\boldsymbol{\Theta}_{0\mathcal{M}}^{-} + \Delta_{\mathcal{M}}^{-}\|_1 - \|\boldsymbol{\Theta}_{0\mathcal{M}}^{-}\|_1}_{(a)} +
   \underbrace{\|\boldsymbol{\Theta}_{0\mathcal{M}^{\perp}}^{-} + \Delta_{\mathcal{M}^{\perp}}^{-}\|_1 - \|\boldsymbol{\Theta}_{0\mathcal{M^\perp}}^{-}\|_1}_{(b)} + \underbrace{\|\boldsymbol{\Theta}_{0\mathcal{M}}^{+} + \Delta_{\mathcal{M}}^{+}\|_1 - \|\boldsymbol{\Theta}_{0\mathcal{M}}^{+}\|_1)}_{(c)}.  \\  
\end{align*}
Bounding term separately, we use the reverse triangle inequality to write
\begin{align*}
    (a) = \|\boldsymbol{\Theta}_{0\mathcal{M}}^{-} - (- \Delta_{\mathcal{M}}^{-})\|_1 - \|\boldsymbol{\Theta}_{0\mathcal{M}}^{-}\|_1 &\geq \left| \|\boldsymbol{\Theta}^{-}_{0\mathcal{M}}\|_1 - \|-\Delta^{-}_{\mathcal{M}}\|_1 \right| - \|\boldsymbol{\Theta}_{0\mathcal{M}}^{-}\|_1 \\ 
    &\geq \|\boldsymbol{\Theta}^{-}_{0\mathcal{M}}\|_1 - \|\Delta^{-}_{\mathcal{M}}\|_1 -\|\boldsymbol{\Theta}_{0\mathcal{M}}^{-}\|_1 = -\|\Delta^{-}_{\mathcal{M}}\|_1.
\end{align*}
A similar argument yields $(c) \geq - \|\Delta^+_{0\mathcal{M}}\|_1$. Finally, we note that under the true model  $\|\boldsymbol{\Theta}^{-}_{0\mathcal{M^\perp}}\|_1=0.$ Hence, $(b) = \|\Delta^{-}_{\mathcal{M}^{\perp}}\|_1$.  

The above implies that
\begin{equation*}
   \lambda(\|\boldsymbol{\Theta}_0 + \Delta\|_1 - \|\boldsymbol{\Theta}_0\|_1) \geq \lambda(\|\Delta^{-}_{\mathcal{M}^{\perp}}\|_1 - \|\Delta_{\mathcal{M}}^{-}\|_1 - \|\Delta^+_{\mathcal{M}}\|_1).
\end{equation*}
and the lower bound
\begin{align*}
     F(\boldsymbol{\Theta}_0+\Delta) &\geq \frac{1}{18}\underline{k}^2\|\Delta\|^2_F - c\sqrt{\frac{\tau \log p}{m}} \|\Delta^{-}\|_1 - c\sqrt{\frac{(p+s) \tau\log p}{m}}\|\Delta^+\|_F  \\ &+  \lambda (\|\Delta^{-}_{\mathcal{M}^{\perp}} \|_1-\|\Delta^{-}_{\mathcal{M}}\|_1-\|\Delta^{+}_{\mathcal{M}}\|_1).
\end{align*}

Now, substituting $\textcolor{black}{\lambda = \frac{c}{\epsilon}\sqrt{\frac{\tau \log p}{m}}}$ yields
\begin{align*}
      F(\boldsymbol{\Theta}_0+\Delta) &\geq \frac{1}{18}\underline{k}^2\|\Delta\|^2_F - c\sqrt{\frac{\tau\log p}{m}} \|\Delta^{-}_{\mathcal{M}}\|_1 \left(1+\frac{1}{\epsilon}\right) - c \sqrt{\frac{\tau\log p}{m}} \|\Delta^{-}_{\mathcal{M}_{\perp}}\|_1 \left(1-\frac{1}{\epsilon}\right) \\
      &- c\sqrt{\frac{\tau \log p}{m}} \|\Delta^{+}_{\mathcal{M}}\|_1 \left(\frac{1}{\epsilon}\right) - c\sqrt{\frac{(p+s)\log p}{m}}\|\Delta^+\|_F.
\end{align*}

For sufficiently small $0<\epsilon<1$ the third term is always positive, and hence may be omitted from the lower bound. Now, we note that
\begin{align*}
    \|\Delta_{\mathcal{M}}^{-}\|_1 \leq \sqrt{s}\|\Delta^{-}_{\mathcal{M}}\|_F \leq \sqrt{s}\|\Delta^{-}\|_F \leq \sqrt{p+s}\|\Delta^{-}\|_F,
\end{align*}
and similarly, $\|\Delta^{+}_{\mathcal{M}}\|_1 \leq \sqrt{p+s}\|\Delta^{+}\|_F$. 
Thus, we have 
\begin{align}
      F(\boldsymbol{\Theta}_0+\Delta) &\geq \|\Delta^{-}\|^2_F \left[\frac{1}{18}\underline{k}^2- c\sqrt{\frac{(p+s)\log p^\tau}{m}}\left(1+\frac{1}{\epsilon}\right)\|\Delta^{-}\|^{-1}_F\right] \nonumber \\
      &+  \|\Delta^{+}\|^2_F \left[\frac{1}{18}\underline{k}^2- c\sqrt{\frac{(p+s)\log p^\tau}{m}}\left(1+\frac{1}{\epsilon}\right)\|\Delta^{+}\|^{-1}_F\right] \nonumber \\
\end{align}
where using $\|\Delta\|_F=Gr_m$ and letting $r_m=\sqrt{\frac{(p+s)\log p^\tau}{m}}$ we have
\begin{align}
     F(\boldsymbol{\Theta}_0+\Delta) &\geq \|\Delta\|^2_F \left[\frac{1}{18}\underline{k}^2 - \frac{c(1+\epsilon)}{\epsilon G}\right] \nonumber\\
     & {\geq \|\Delta\|^2_F \left[\frac{1}{18}\underline{k}^2 - \frac{2 c}{\epsilon G}\right]>0}\;,
\end{align}
for $G$ sufficiently large i.e.,
\begin{equation*}
   {G \ge \frac{36c}{\epsilon\underline{k}^2}}.
\end{equation*}
Hence, 
\begin{equation}
\label{eq:final_delta}
    \|\Delta\|_F \leq Gr_m = \frac{36c}{\epsilon \underline{k}^2} \sqrt{\frac{(p+s) \log p^\tau}{m}}
\end{equation}

In the above, we assumed that $(ii)$ in \eqref{eq:F_function_glasso} was bound according to  claim \ref{claim}. The above bound confirms this claim holds under the assumption that we have sufficient tapers 
\begin{equation*}
    m > \frac{18^2 c^2}{\epsilon^2 \underline{k}^4} (p+s) \log p^{\tau}\bar{k}^2,
\end{equation*}
which ensures that $\|\Delta\|_2\leq 2\|\boldsymbol{\Theta}_0\|_2$. Finally, as before, the above analysis was conditioned on the event $\mathcal{E}$, statement \eqref{eq:final_delta} holds with probability $1 - \frac{8}{p^{\tau-2}}\rightarrow 1,$ for any $\tau>2.$

\end{proof}

\subsection{Proof of Proposition \ref{prop:l_inf}}
\begin{proof}
    The proof follows a primal dual witness approach, and extends the arguments presented in \cite{ravikumar2011high} to allow for complex positive definite matrices. Similar to the proofs of Propositions \ref{prop:ridge} and \ref{prop:glasso}, we condition on the event $\mathcal{E}$ defined in \eqref{eq:event} in the following analysis.
    The results are developed for a generic frequency $\omega$, which we omit below for ease of notation.

    As in \cite{ravikumar2011high}, the proof is broken down into a series of lemmas, thus easing the comparison between the real and complex analysis. Importantly, we note that a similar approach has been discussed in \cite{deb2024regularized}.

    \begin{lemma}
    \label{lemma:rav_3}
        The RSE in \eqref{estimator} with the lasso penalty \eqref{pen2} has a unique solution $\boldsymbol{\hat{\Theta}}$ characterised by
        \begin{equation}
        \label{eq:optimality_cond}
            \hat{\mathbf{S}}-\hat{\boldsymbol{\Theta}}^{-1}+\lambda\hat{\mathbf{Z}}=0,
        \end{equation}
        where $\mathbf{\hat{Z}}$ is an element of the sub-differential $\partial \|\hat{\boldsymbol{\Theta}}\|_1$ which has the following entries
    \end{lemma}

    \begin{align*}
        Z_{qr} = \operatorname{Re}(Z_{qr}) + i \operatorname{Im}(Z_{qr})=
        \begin{cases}
            1 \ &\textrm{if} \ q=r\\
            \left[{\operatorname{Re}(\Theta_{qr}) + i \operatorname{Im}(\Theta_{qr})}\right]/{|\Theta_{qr}|} \ &\textrm{if} \ q\neq r \ \textrm{and} \ \Theta_{qr}\neq 0\\
            \in \{\psi \in \mathbb{C} : |\psi|<1\} \  &\textrm{if} \ q\neq r \ \textrm{and} \ \Theta_{qr}=0.
        \end{cases}
    \end{align*}
    
Lemma \ref{lemma:rav_3} is a modified version of Lemma 3 in \cite{ravikumar2011high} where we have substituted our multi-taper periodogram estimate $\hat{\mathbf{S}} \in \mathbb{C}^{p \times p}$ in place of the estimated covariance matrix $\hat{\Sigma} \in \mathbb{R}^{p\times p}$. 


\newenvironment{sproof}{%
  \renewcommand{\proofname}{Proof Sketch}\proof}{\endproof}

\begin{proof}
    The RSE in \eqref{estimator} with the lasso penalty is a strict convex program. Therefore, if the minimum is obtained, then it is unique. 
    
    Suppose $\lambda>0$. Then, by Lagrangian duality, we can rewrite our optimisation problem as 
    \begin{equation*}
        \min_{\boldsymbol{\Theta} \in \mathcal{C}^\prime} \left\{-\log \det (\boldsymbol{\Theta}) + \Tr\left\{\mathbf{\hat{S}}\boldsymbol{\Theta}\right\} \right\},
    \end{equation*}
    where $\mathcal{C}^{\prime}:=\{\boldsymbol{\Theta} \in \mathbb{C}^{p \times p}:\boldsymbol{\Theta}\succ 0, \  \boldsymbol{\Theta}^H = \boldsymbol{\Theta}, \ \|\boldsymbol{\Theta}\|_1 \leq c(\lambda) \}$ for some $c(\lambda)<\infty.$ Since both the off- and on-diagonal elements are bounded within the $\ell_1$-ball, one can see that the minimum is attained, and as argued above, is also unique.

    Recall the form of the RSE in \eqref{estimator} with the lasso penalty. By standard optimality conditions for convex programs, a matrix $\boldsymbol{\hat{\Theta}} \in \mathcal{C}$ is an optimal solution if and only if the zero matrix belongs to the sub-differential of the objective. Equivalently, there must exist a matrix $\hat{\mathbf{Z}}$ in the sub-differential of the norm $\|\cdot\|_1$ evaluated at $\hat{\boldsymbol{\Theta}}$ such that
    \begin{equation*}
        \hat{\mathbf{S}} - \hat{\boldsymbol{\Theta}}^{-1}+\lambda\hat{\mathbf{Z}}=0.
    \end{equation*}
    as claimed.
 
\end{proof}

Given that Lemma \ref{lemma:rav_3} holds, we now outline the primal-dual witness approach which is central to our proof. 

\subsection*{Primal-Dual Witness Approach}

We define the witness, with oracle knowledge about the edge set $\mathcal{M}$, as 
\begin{equation}
\label{eq:restricted_estimator}
    \dot{\boldsymbol{\Theta}} = \arg \min _{\boldsymbol{\Theta}(\omega) \in \dot{\mathcal{C}}}\left\{-\log \operatorname{det}(\boldsymbol{\Theta})+ \Tr\left\{\mathbf{\hat{S}}\boldsymbol{\Theta}\right\} +\lambda \|\boldsymbol{\Theta}\|_1\right\},
\end{equation}
where $\dot{\mathcal{C}}:= \{\boldsymbol{\Theta}\succ 0 : \boldsymbol{\Theta}^H = \boldsymbol{\Theta}, \  \boldsymbol{\Theta}_{\mathcal{M}^{\perp}}=0\}.$ Therefore, $\dot{\boldsymbol{\Theta}}$ is a positive definite Hermitian matrix which correctly specifies the true edge set.

Let $\dot{\mathbf{Z}}$ be a member of the sub-differential of the regulariser $\|\cdot\|_1$ evaluated at $\dot{\boldsymbol{\Theta}}$. Now, considering each $(q,r) \in \mathcal{M}^{\perp}$, we write
\begin{equation*}
\label{Z_qr}
    \dot{{Z}}_{qr} := \frac{1}{\lambda}\left\{-\hat{S}_{qr} + \dot{\Theta}^{-1}_{qr} \right\},
\end{equation*}
which ensures that the primal-dual pair $(\dot{\boldsymbol{\Theta}}, \dot{\mathbf{Z}})$ satisfy the optimality condition in \eqref{eq:optimality_cond} replacing $(\hat{\boldsymbol{\Theta}}, \hat{\mathbf{Z}})$.

Finally, it remains to check that 
\begin{equation}
\label{dual_feas}
    |\dot{Z}_{qr}| < 1 \ \forall (q,r) \ \in \mathcal{M}^{\perp}.
\end{equation}
This step ensures that $\dot{Z}_{qr}$ satisfies the necessary conditions to belong to the sub-differential.
We will subsequently refer to \eqref{dual_feas} as the \textit{strict dual feasibility} condition. 

In the following discussion, we require some additional notation. Firstly, let $\mathbf{W} \in \mathbb{C}^{p \times p}$ denote the discrepancy between the multi-taper periodogram and the true spectral density matrix, i.e.,
\begin{equation}
\label{W}
    \mathbf{W} = \hat{\mathbf{S}} - \mathbf{S}.
\end{equation}
Secondly, let $\dot{\Delta} = \dot{\boldsymbol{\Theta}} - \boldsymbol{\Theta}_0$ measure the difference between the primal matrix and the ground truth.  Finally, let $R(\dot{\Delta})$ denote the remainder of the first-order Taylor expansion of ${\partial}/{\partial \boldsymbol{\Theta}}\left(-\log \det (\dot{\boldsymbol{\Theta}})\right) $ around $\boldsymbol{\Theta}_0.$ That is,
\begin{equation}
\label{R}
    R(\dot{\Delta}) = \dot{\boldsymbol{\Theta}}^{-1} - \boldsymbol{\Theta}_0^{-1} + \boldsymbol{\Theta}_0^{-1} \dot{\Delta}\boldsymbol{\Theta}_0^{-1}.
\end{equation}

We now provide sufficient conditions on $\mathbf{W}$ and $R(\dot{\Delta})$ for the strict dual feasibility condition in \eqref{dual_feas} to hold. The following lemma follows from Lemma 4 in \cite{ravikumar2011high}.

\begin{lemma}\label{lemma:rav4}(Strict Dual Feasibility). Suppose that 
\begin{equation}
\label{ass:W,R}
    \max \left\{\|\mathbf{W}\|_{\infty}, \|R(\dot{\Delta})\|_{\infty}\right\} \leq \frac{\alpha \lambda}{8}.
\end{equation}
Then, the vector $\dot{\mathbf{Z}}_{\mathcal{M}^\perp}$ constructed from \eqref{Z_qr} satisfies $\|\dot{\mathbf{Z}}_{\mathcal{M}^\perp}\|_{\infty}<1,$ and therefore $\dot{\boldsymbol{\Theta}} = \hat{\boldsymbol{\Theta}}.$    
\end{lemma}
\begin{proof}
    Using \eqref{W} and \eqref{R}, we can rewrite the optimality condition \eqref{eq:optimality_cond} as
    \begin{equation}
        \label{W,R}
        \mathbf{W} - R(\dot{\Delta}) + \boldsymbol{\Theta}_0^{-1} \dot{\Delta}\boldsymbol{\Theta}_0^{-1} + \lambda \dot{Z}=0.
    \end{equation}

Equation \eqref{W,R} can be re-written as an ordinary linear equation by `vectorising' the matrices. We adopt the notation $\operatorname{vec}(M)$ or equivalently $\vec{M}$ for the vector version of  matrix $M$ obtained by stacking the rows of $M$ into a single column vector. 

Noting that 
\begin{equation*}
    \operatorname{vec}\left(\boldsymbol{\Theta}_0^{-1} \dot{\Delta} \boldsymbol{\Theta}_0^{-1} \right) = \left(\boldsymbol{\Theta}_0^{-1} \otimes \Theta_0^{-1} \right) \vec{\dot{\Delta}} = \Gamma\vec{\dot{\Delta}},
\end{equation*}
we rewrite equation \eqref{W,R} as two blocks of linear equations in terms of the disjoint decomposition $\mathcal{M}$ and $\mathcal{M}^{\perp}$ as follows:
\begin{align}
    \Gamma_{\mathcal{M}\mathcal{M}} \vec{\dot{\Delta}}_{\mathcal{M}} + \vec{\mathbf{W}}_{\mathcal{M}} - \vec{R}_{\mathcal{M}} + \lambda \vec{\dot{\mathbf{Z}}}_{\mathcal{M}} &= 0 \label{a}\tag{36a} \\
    \Gamma_{\mathcal{M}^{\perp}\mathcal{M}} \vec{\dot{\Delta}}_{\mathcal{M}} + \vec{\mathbf{W}}_{\mathcal{M}^\perp} - \vec{R}_{\mathcal{M}^{\perp}} + \lambda \vec{\dot{\mathbf{Z}}}_{\mathcal{M}^\perp} &= 0 \label{b}\tag{36b}.
\end{align}
Here, we used the fact that $\dot{\Delta}_{\mathcal{M}^{\perp}}=0$ by construction. Further, using the fact that $\Gamma_{\mathcal{M}\mathcal{M}}$ is invertible, we can solve for $\vec{\dot{\Delta}}$ from \eqref{a} and obtain
\begin{equation*}
    \vec{\dot{\Delta}}_{\mathcal{M}} = \Gamma_{\mathcal{M}\mathcal{M}}^{-1} \left[\vec{R}_{\mathcal{M}} 
- \vec{\mathbf{W}}_{\mathcal{M}} - \lambda \vec{\dot{\mathbf{Z}}}_{\mathcal{M}}\right].
\end{equation*}
Substituting this into \eqref{b}, we find that 
\begin{align*}
    \vec{\dot{\mathbf{Z}}}_{\mathcal{M}^{\perp}} &= \frac{1}{\lambda} \left(\vec{R}_{\mathcal{M}^{\perp}} - \vec{\mathbf{W}}_{\mathcal{M}^\perp}\right) - \frac{1}{\lambda} \Gamma_{\mathcal{M}^{\perp}\mathcal{M}} \Gamma_{\mathcal{M}\mathcal{M}}^{-1} \left(\vec{R}_{\mathcal{M}} 
- \vec{\mathbf{W}}_{\mathcal{M}} \right) + \Gamma_{\mathcal{M}^{\perp}\mathcal{M}} \Gamma_{\mathcal{M}\mathcal{M}}^{-1}\vec{\dot{\mathbf{Z}}}_{\mathcal{M}}.
\end{align*}
Taking the $\ell_{\infty}$ norm of both sides yields

\begin{align*}
    \| \vec{\dot{\mathbf{Z}}}_{\mathcal{M}}\|_{\infty} \leq \frac{1}{\lambda}\left(\|\vec{R}_{\mathcal{M}}\|_{\infty} + \|\vec{\mathbf{W}}_{\mathcal{M}}\|_{\infty}\right) &+ \frac{1}{\lambda}  \opnorm{\Gamma_{\mathcal{M}^{\perp}\mathcal{M}} \Gamma_{\mathcal{M}\mathcal{M}}^{-1}}{\infty}\left(\|\vec{R}_{\mathcal{M}}\|_{\infty} + \|\vec{\mathbf{W}}_{\mathcal{M}}\|_{\infty}\right) \\
    &+ \|\Gamma_{\mathcal{M}^\perp\mathcal{M}} \Gamma_{\mathcal{M}\mathcal{M}}^{-1}\vec{\dot{\mathbf{Z}}}_{\mathcal{M}}\|_{\infty}.
\end{align*}
Under assumption A3 we know that $\|\Gamma_{\mathcal{M}^{\perp}\mathcal{M}} \Gamma_{\mathcal{M}\mathcal{M}}^{-1}\|_{\infty}\leq \|\Gamma_{\mathcal{M}^{\perp}\mathcal{M}} \Gamma_{\mathcal{M}\mathcal{M}}^{-1}\|_1 \leq (1-\alpha)$. Further, we know that $\|\vec{\dot{\mathbf{Z}}}_{\mathcal{M}}\|_{\infty}\leq 1$, since $\dot{\mathbf{Z}}$ belongs to the sub-differential of the norm $\|\cdot\|_1$ by construction. Hence,

\begin{align*}
     \| \vec{\dot{\mathbf{Z}}}_{\mathcal{M}^{\perp}}\|_{\infty} \leq \frac{2-\alpha}{\lambda} \left(\|\vec{R}_{\mathcal{M}}\|_{\infty} + \|\vec{\mathbf{W}}_{\mathcal{M}}\|_{\infty}\right) + (1-\alpha).
\end{align*}

Applying assumption \eqref{ass:W,R} we arrive at
\begin{equation*}
    \| \vec{\dot{\mathbf{Z}}}_{\mathcal{M}^{\perp}}\|_{\infty} \leq \frac{2-\alpha}{\lambda}\left(\frac{\alpha \lambda}{4}\right) +(1-\alpha) \leq \frac{\alpha}{2} + (1-\alpha) <1,
\end{equation*}
as claimed.
\end{proof}

\begin{lemma}(Control of remainder). \label{lemma:rav5} Suppose that $\|\dot{\Delta}\|_\infty\leq \frac{1}{3\kappa_S d}$.
Then 
\begin{equation}
\label{eq:lemma5}
    \|R(\dot{\Delta})\|_{\infty} \leq \frac{3}{2} d \|\dot{\Delta}\|^2_{\infty} \kappa^3_{S}.
\end{equation}
\end{lemma}

\begin{proof}
   We begin by rewriting the remainder term in \eqref{R} as 
   \begin{equation*}
           R(\dot{\Delta}) = \underbrace{({\boldsymbol{\Theta}_0 +\dot{\Delta}})^{-1}}_A - \boldsymbol{\Theta}_0^{-1} + \boldsymbol{\Theta}_0^{-1} \dot{\Delta}\boldsymbol{\Theta}_0^{-1}.
   \end{equation*}
   Expanding the power series of the complex matrix $A$ yields
   \begin{align*}
       ({\boldsymbol{\Theta}_0 +\dot{\Delta}})^{-1} &= (\boldsymbol{\Theta}_0(\mathbb{I} + \boldsymbol{\Theta}_0^{-1}\dot{\Delta}))^{-1}\\
       &= (\mathbb{I} + \boldsymbol{\Theta}_0^{-1}\dot{\Delta}))^{-1} (\boldsymbol{\Theta}_0)^{-1} \\
       &= \sum_{k=0}^\infty (-1)^k (\boldsymbol{\Theta}_0^{-1}\dot{\Delta})^k (\boldsymbol{\Theta}_0)^{-1} \\
       &= \boldsymbol{\Theta}_0^{-1} - \boldsymbol{\Theta}_0^{-1}\dot{\Delta} \boldsymbol{\Theta}_0^{-1} + \boldsymbol{\Theta}_0^{-1}\dot{\Delta} \boldsymbol{\Theta}_0^{-1} \dot{\Delta} J \boldsymbol{\Theta}_0^{-1}, 
   \end{align*}
where $J = \sum_{k=0}^\infty(-1)^k(\boldsymbol{\Theta}_0^{-1}\dot{\Delta})^k.$ The above expansion can be verified under the assumption that $\|\dot{\Delta}\|_{\infty}<1/(3 \kappa_S d)$. More specifically, by the sub-multiplicativity of the matrix norm $\opnorm{\cdot}{\infty}$, we have that
\begin{align}
\label{eq:lemma5_3}
    \opnorm{\boldsymbol{\Theta}_0^{-1}\Delta}{\infty} &\leq \opnorm{\boldsymbol{\Theta}_0^{-1}}{\infty} \opnorm{\dot{\Delta}}{\infty} \\
    &\leq \kappa_S d \|\dot{\Delta}\|_\infty < 1/3.
\end{align}
The second inequality follows from the definition of $\kappa_S$ in \eqref{kappa_s} and the fact that $\dot{\Delta}$ has at most $d$ non-zeroes per row/column. We then use our assumption that $\|\dot{\Delta}\|_{\infty}<1/(3 \kappa_S d)$ to arrive at the final result.

Using the above convergent matrix expansion, we can rewrite the remainder term as
\begin{equation}
\label{eq:remainder2}
    R(\dot{\Delta}) = \boldsymbol{\Theta}_0^{-1}\dot{\Delta} \boldsymbol{\Theta}_0^{-1} \dot{\Delta} J \boldsymbol{\Theta}_0^{-1}.
\end{equation}
Let $e_q$ denote the unit vector with 1 in position $q$ and zeroes elsewhere. Then, using \eqref{eq:remainder2} and standard vector inequalities we can write
\begin{align*}
    \|R(\dot{\Delta})\|_\infty &= \max_{qr}\left|e_q^\prime \boldsymbol{\Theta}_0^{-1}\dot{\Delta} \boldsymbol{\Theta}_0^{-1}\dot{\Delta}J \boldsymbol{\Theta}_0^{-1}e_r\right| \\
    &\leq \max_q\|e_q^\prime \boldsymbol{\Theta}_0^{-1}\dot{\Delta}\|_{\infty} \max_r\|\boldsymbol{\Theta}_0^{-1}\dot{\Delta}J\boldsymbol{\Theta}_0^{-1}e_r\|_1 \\
    &\leq \max_q \|e_q^\prime \boldsymbol{\Theta}_0^{-1}\|_1 \|\dot{\Delta}\|_{\infty} \max_r\|\boldsymbol{\Theta}_0^{-1}\dot{\Delta}J\boldsymbol{\Theta}_0^{-1}e_r\|_1.
\end{align*}
If we define the $\ell_1$-operator norm for a generic matrix $M$ as $\opnorm{M}{1}:= \max_{\|x\|_{1=1}}\|Mx\|_1$ then we can write
\begin{align}
\label{eq:lemma5_2}
      \|R(\dot{\Delta})\|_\infty &\leq \opnorm{\boldsymbol{\Theta}_0^{-1}}{\infty} \|\dot{\Delta}\|_\infty \opnorm{\boldsymbol{\Theta}_0^{-1}\dot{\Delta} J \boldsymbol{\Theta}_0^{-1} }{1} \nonumber \\
      &\leq \|\dot{\Delta}\|_\infty \opnorm{\boldsymbol{\Theta}_0^{-1}}{\infty} \opnorm{\boldsymbol{\Theta}_0^{-1}\dot{\Delta} J^H \boldsymbol{\Theta}_0^{-1} }{\infty} \nonumber \\
      &\leq  \|\dot{\Delta}\|_\infty  \kappa_S \opnorm{\boldsymbol{\Theta}_0^{-1}}{\infty}^2 \opnorm{J^H}{\infty} \opnorm{\Delta}{\infty},
\end{align}
where the second inequality follows from the fact that $\opnorm{M}{1}=\opnorm{M^H}{\infty}$ for any complex matrix $M$, and the third inequality follows from the definition of $\kappa_S$ in \eqref{kappa_s} and the sub-multiplicativity of the matrix norm $\opnorm{\cdot}{\infty}.$

 To complete the proof, it remains to consider the term $\opnorm{J^H}{\infty}$ in \eqref{eq:lemma5_2}. By the sub-multiplicative property of matrix norms we find that
 \begin{equation}
 \label{eq:lemma5_4}
     \opnorm{J^H}{\infty} \leq \sum_{k=0}^\infty \opnorm{\boldsymbol{\Theta}_0^{-1}\dot{\Delta}}{\infty}^k \leq \frac{1}{1-\opnorm{\boldsymbol{\Theta}_0^{-1}}{\infty}\opnorm{\dot{\Delta}}{\infty}} \leq \frac{3}{2},
 \end{equation}
 since $ \opnorm{\boldsymbol{\Theta}_0^{-1}\Delta}{\infty} \leq 1/3$ from \eqref{eq:lemma5_3}. Substituting this into \eqref{eq:lemma5_4} yields the final result.


\end{proof}

\begin{lemma}(Control of $\dot{\Delta}$). \label{lemma:rav6} Suppose that 
\begin{equation}
\label{ass:lemma_rav6}
    \tilde{r} := 2 \kappa_{\Gamma}\left(\|W\|_\infty+\lambda\right) \leq \frac{1}{3\kappa_S d}\min\left\{1, \frac{1}{\kappa_S^2 \kappa_\Gamma}\right\}.
\end{equation}
Then we have the elementwise $\ell_\infty$ bound
\begin{equation}
\label{eq:delta_inf}
    \|\dot{\Delta}\|_\infty = \|\boldsymbol{\dot{\Theta}} - \boldsymbol{\Theta}_0\|_{\infty} \leq \tilde{r}.
\end{equation}
\end{lemma}

\begin{proof}
  Following the proof of Lemma \ref{lemma:rav_3}, we know that the restricted problem \eqref{eq:restricted_estimator} has a unique solution $\dot{\boldsymbol{\Theta}}$. If we take partial derivatives of the Lagrangian of the restricted problem \eqref{eq:restricted_estimator} with respect to unrestricted elements $\boldsymbol{\Theta}_{\mathcal{M}}$, then we have that
  \begin{equation}
  \label{eq:zero_grad}
      \mathcal{U}_1(\boldsymbol{\Theta}_{\mathcal{M}}) := \hat{\mathbf{S}}_{\mathcal{M}} - \boldsymbol{\Theta}^{-1}_{\mathcal{M}} + \lambda \dot{\mathbf{Z}}_{\mathcal{M}} = 0.
  \end{equation}
   
The above zero gradient condition \eqref{eq:zero_grad} is necessary and sufficient for an optimum of the Lagrangian problem, and thus has a unique solution $\dot{\boldsymbol{\Theta}}_{\mathcal{M}}$. 

Our aim is to bound the deviation of $\dot{\boldsymbol{\Theta}}_{\mathcal{M}}$ from $\boldsymbol{\Theta}_{0\mathcal{M}}.$ This is equivalent to bounding $\dot{\Delta}=\dot{\boldsymbol{\Theta}} - \boldsymbol{\Theta}_0$ since $\boldsymbol{\Theta}_{\mathcal{M}^{\perp}}=\boldsymbol{\Theta}_{0\mathcal{M}}=0$. As in \cite{ravikumar2011high}, we consider the $\ell_\infty$-ball defined as
\begin{equation}
    \label{eq:l_inf_ball}
    \mathcal{B}(\tilde{r}):=\left\{\boldsymbol{\Theta}_{\mathcal{M}} \ \ | \ \ \|\boldsymbol{\Theta}_{\mathcal{M}}\|_{\infty}\leq \tilde{r}\right\},
\end{equation}
where $\tilde{r}$ is defined in \eqref{ass:lemma_rav6}. The main idea is to show that a solution to the zero gradient condition \eqref{eq:zero_grad} is contained within the ball \eqref{eq:l_inf_ball}. Consequently, we can conclude that $\dot{\boldsymbol{\Theta}}-\boldsymbol{\Theta}_0$ belongs to this ball.

We define the map $\mathcal{U}_2:\mathbb{C}^{|\mathcal{M}|}\rightarrow \mathbb{C}^{|\mathcal{M}|}$ as
\begin{equation}
    \label{eq:map}
    \mathcal{U}_2(\vec{\dot{\Delta}}_{\mathcal{M}}):= -(\Gamma_{\mathcal{M}\mathcal{M}})^{-1} (\vec{\mathcal{U}}_1(\boldsymbol{\Theta}_{0\mathcal{M}}+\dot{\Delta}_{\mathcal{M}})) + \vec{\dot{\Delta}}_{\mathcal{M}},
\end{equation}
where $\vec{\mathcal{U}_1}$ and $\vec{\dot{\Delta}}_{\mathcal{M}}$ denote vectorised versions of $\mathcal{U}_1$ and $\dot{\Delta}_{\mathcal{M}}$ respectively. By construction, we therefore have that $ \mathcal{U}_2(\vec{\dot{\Delta}}_{\mathcal{M}})=0$ if and only if $\mathcal{U}_1(\boldsymbol{\Theta}_{0\mathcal{M}}+\dot{\Delta}_{\mathcal{M}})=\mathcal{U}_1(\dot{\boldsymbol{\Theta}}_{\mathcal{M}})=0.$

\begin{claim} Suppose that $$\mathcal{U}_2(\mathcal{B}(\tilde{r})) \subseteq \mathcal{B}(\tilde{r}).$$ Then, there exists a fixed point $\vec{\dot{\Delta}}_{\mathcal{M}} \in \mathcal{B}(\tilde{r})$. By the uniqueness of the zero condition we can therefore conclude that $\|\dot{\boldsymbol{\Theta}}_\mathcal{M} - \boldsymbol{\Theta}_{0\mathcal{M}}\|_{\infty}\leq \tilde{r}$.  
\end{claim}

Claim 6.2 therefore establishes the bound \eqref{eq:delta_inf}. A proof of the claim is given in Appendix C of \cite{ravikumar2011high}.

\end{proof}


We now have all the necessary parts to prove Proposition \ref{prop:l_inf}. To proceed, we first need to show that assumption \eqref{ass:W,R} in Lemma \ref{lemma:rav4} holds. This will allow us to conclude that the restricted estimate is equal to Lasso estimate, i.e., $\dot{\boldsymbol{\Theta}}=\hat{\boldsymbol{\Theta}}_G$.

It is easy to see that by choice of our regularisation parameter $\lambda = (8c/\alpha) \sqrt{\log p^\tau/m}$, we have $\|\mathbf{W}\|_{\infty}\leq (\alpha \lambda)/8.$ Thus, in order to show that assumption \eqref{ass:W,R} holds, it remains to demonstrate that $\|R(\dot{\Delta})\|_\infty\leq (\alpha \lambda)/8.$ To do so, we use 
Lemmas \ref{lemma:rav5} and \ref{lemma:rav6} above.

We start by showing that assumption \eqref{ass:lemma_rav6} holds under for our choice of $\lambda$ and provided that we have sufficient tapers $m$. 

From Lemma \ref{lemma:union} and our choice of  regularisation parameter $\lambda = (8c/\alpha) \sqrt{\log p^\tau/m}$, 
\begin{align*}
         2 \kappa_{\Gamma}\left(\|W\|_\infty+\lambda\right) \leq  2 c \kappa_{\Gamma} \left(1+\frac{8}{\alpha}\right) \sqrt{\frac{\log p^\tau}{m}},
\end{align*}
provided that $c \sqrt{\log p^\tau/m}\leq 80 \max_q\{S_{qq}(\omega)\}$.

Moreover, if we have sufficient tapers
\begin{align}
    m > \left\{6 c (1+8/\alpha) d \kappa_S^2 \kappa_\Gamma \max\{1, \kappa_S\kappa_\Gamma\}\right\}^2 \tau \log p,
\end{align}
we can write
\begin{align}
\label{eq:lemma6}
     2 c \kappa_{\Gamma} \left(1+\frac{8}{\alpha}\right) \sqrt{\frac{\log p^\tau}{m}} \leq \frac{1}{3\kappa_S d}\min \left\{1,\frac{1}{\kappa^2_S \kappa_\Gamma}\right\}.
\end{align}
showing that assumption \eqref{ass:lemma_rav6} is therefore satisfied. We can now apply Lemma \ref{lemma:rav6} and conclude that
\begin{equation}
\label{eq:lemma6_2}
    \|\dot{\Delta}\|_{\infty} \leq    2 \kappa_{\Gamma}\left(\|W\|_\infty+\lambda\right) \leq 2 c \kappa_{\Gamma} \left(1+\frac{8}{\alpha}\right) \sqrt{\frac{\log p^\tau}{m}}.
\end{equation}

Now, we can apply Lemma \ref{lemma:rav5} and equations \eqref{eq:lemma6} and \eqref{eq:lemma6_2} arriving at

\begin{align*}
    \|R(\dot{\Delta})\|_{\infty} &\leq \frac{3}{2} d \|\dot{\Delta}\|_{\infty}^2 \kappa^3_S \\
    &\leq \left\{6 \kappa^3_S d \kappa^2_{\Gamma} \left(1+\frac{8}{\alpha}\right)^2 c \sqrt{\frac{\log p^\tau}{m}}\right\} \frac{\lambda \alpha}{8} \\
    &\leq \frac{\alpha \lambda}{8}
\end{align*}

for sufficient tapers
\begin{equation}
    m > \left\{6 c (1+8/\alpha)^2 d \kappa_S^3 \kappa_{\Gamma}^2\right\}^2 \tau \log p.
\end{equation}

Now that we have shown that assumption \eqref{ass:W,R} holds, we can conclude that $\dot{\boldsymbol{\Theta}}= \hat{\boldsymbol{\Theta}}_G$. Thus, the estimator $\hat{\boldsymbol{\Theta}}$ satisfies the $\ell_\infty$-bound \eqref{eq:lemma6_2} of $\dot{\boldsymbol{\Theta}}$, as claimed in Proposition \ref{prop:l_inf} (a). Moreover, we have that $\hat{\boldsymbol{\Theta}}_{\mathcal{M}^\perp} = \dot{\boldsymbol{\Theta}}_{\mathcal{M}^\perp}=0$ and thus $E(\hat{\boldsymbol{\Theta}}) = E(\dot{\boldsymbol{\Theta}}).$

Finally, noting that for all $(q,r) \in \mathcal{M}$ with $|\Theta_{0qr}| > B_{\alpha}\sqrt{\log p^\tau /m}$ we have
\begin{equation*}
   |\hat{\Theta}_{qr}| \geq |\hat{\Theta}_{qr} + {\Theta}_{0qr} - {\Theta}_{0qr}| \geq | {\Theta}_{0qr} | - |\hat{\Theta}_{qr} - {\Theta}_{0qr}|>0.
\end{equation*}

Thus, we conclude that $E(\hat{\boldsymbol{\Theta}}_G)$ includes all edges $(q,r)$ with $|\Theta_{0qr}|>B_{\alpha}\sqrt{\log p^\tau /m}$ as claimed in  \ref{prop:l_inf} (b).
Since the above analysis was conditioned on the event $\mathcal{E}$, Proposition \ref{prop:l_inf} holds with probability $1-\frac{8}{p^{\tau-2}}\rightarrow 1,$ for any $\tau>2$ and sufficient tapers as stated.

\end{proof}


\section{Additional Figures for Section 2}
In this section, we present additional results for the experimental procedure of Section 2 in the main paper. Recall that we we consider a simple homogeneous Poisson point process with rate $\lambda=1$, and obtain 1000 Monte-Carlo replicates of the multi-trial spectrum with $m=10$ trials.

The results of the additional experiments are shown in Figure \ref{supp:experiments}, where we investigate the appropriateness of the asymptotic distributional results, for varying values of $T$. Firstly, we plot the histogram of $\hat{R}_{12}(\omega)$ alongside the theoretical density as defined in \eqref{eq:coherence_density}, with the substitution $R^2=0$. A visual comparison between Fig. \ref{T=100} and Fig. \ref{T=1000} shows the improvement in fit to the asymptotic distribution when $T=100$ compared to when $T=1000.$

\begin{figure}[H]
\centering

\begin{subfigure}[]{\textwidth}
   \centering
   \includegraphics[width=0.95\textwidth]{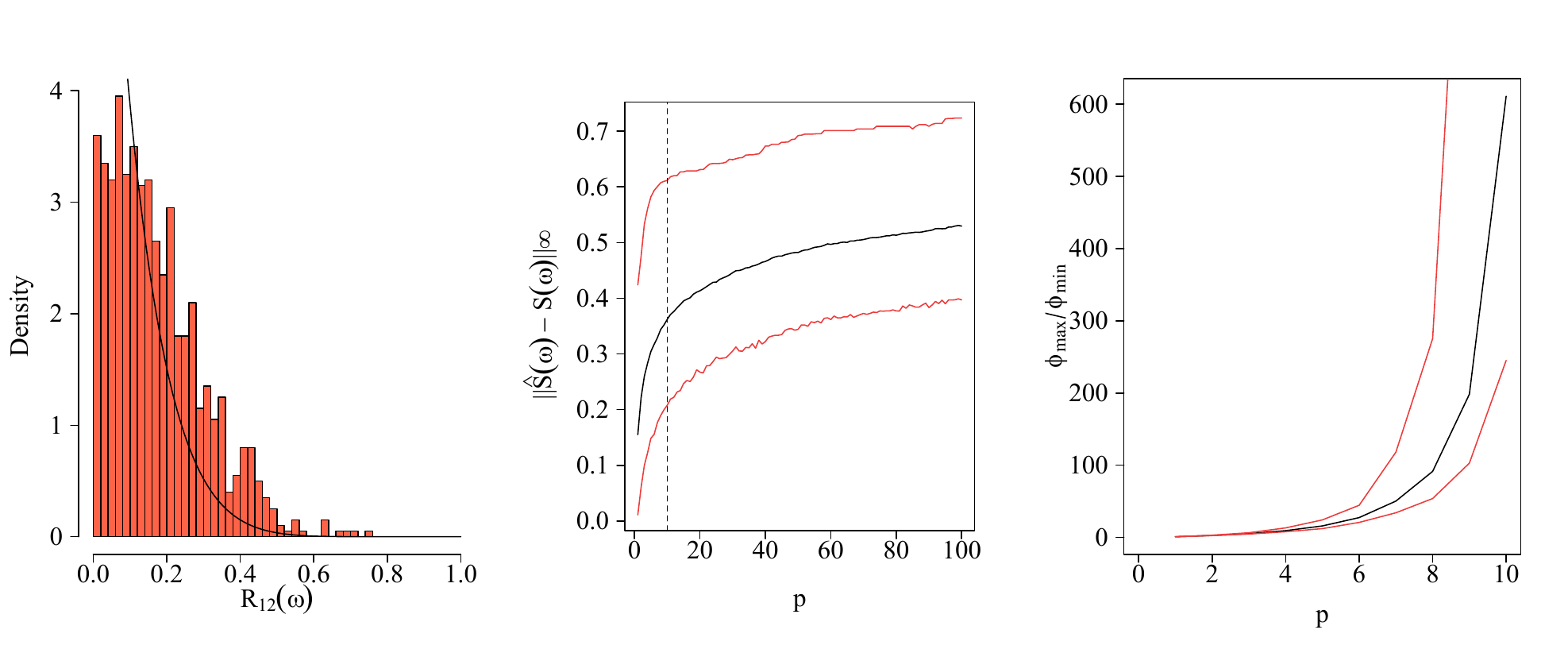}
   \caption{$T=100$}
~
\label{T=100}
\end{subfigure}
\begin{subfigure}[]{\textwidth}
\centering
   \includegraphics[width=0.95\textwidth]{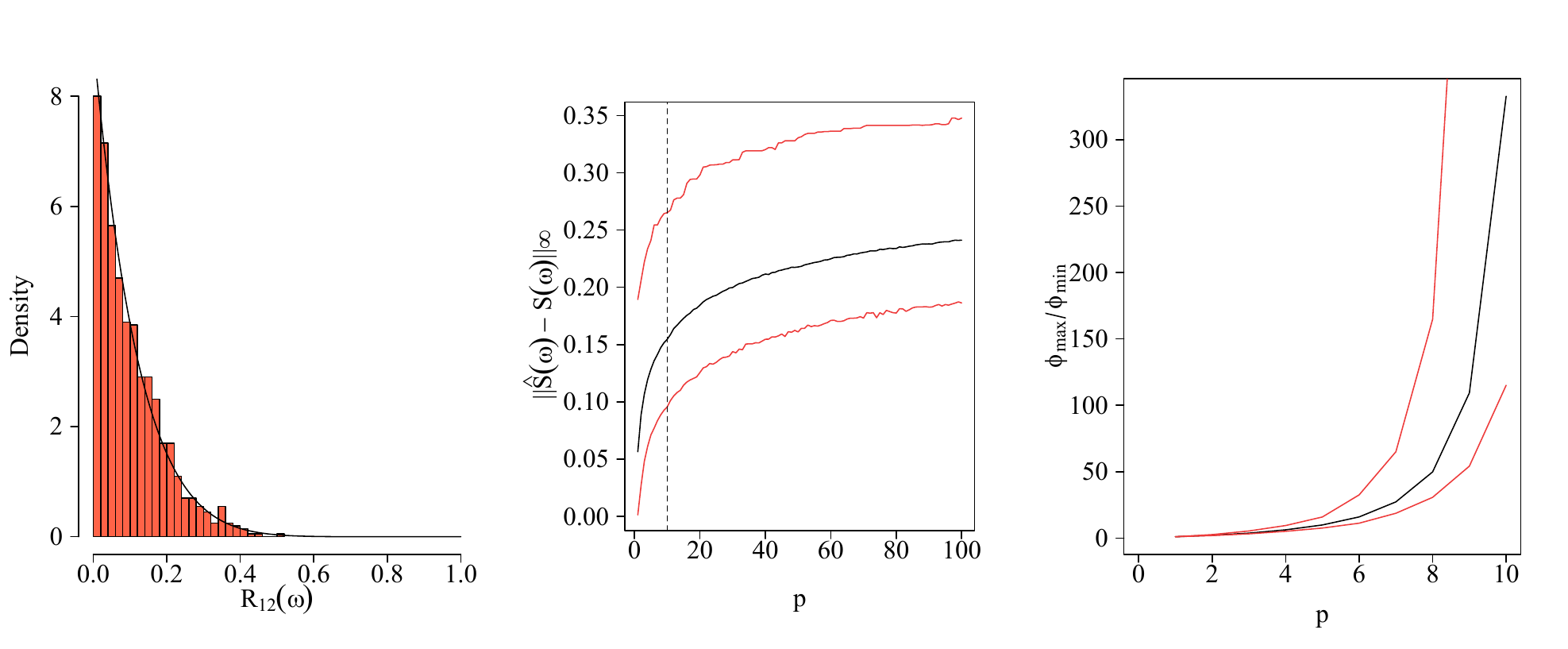}
   \caption{$T=1000$}
   \label{T=1000}
\end{subfigure}
\caption{\small{Results from the experimental procedure in Section 2 of the main paper repeated for different values of $T=(100, 1000)$.}}
\label{supp:experiments}
\end{figure}

We also plot the deviation of the estimator from the true spectrum in terms of the element-wise maximum norm. We note the increased average estimation error for $T=100$ compared to $T=1000.$ Finally, we highlight the difference in behaviour of the condition number for the estimated spectral matrix when $T=100$ compared to $T=1000$.

\section{ADMM Implementation}
In this section we give details of the ADMM algorithm used to solve the optimisation problem \eqref{estimator} in the main paper with the lasso penalty.
ADMM consists of a decomposition-coordination procedure, whereby solutions to small local sub-problems are coordinated to find a solution to a large global problem \citep{boyd2011distributed}. Often viewed as a an attempt to merge the advantages of dual decomposition and augmented Lagrangian methods, ADMM has been used throughout the time series literature to solve convex optimisation problems similar to those considered in this work \citep{jung2015graphical, nadkarnisparse, dallakyan2022time}.

We consider first rewriting the the minimisation as
\begin{align*}
     &\min _{\boldsymbol{\Theta} \in \mathcal{C}}\left\{-\log \operatorname{det}(\boldsymbol{\Theta})+ \Tr\{\mathbf{\hat{S}}, \boldsymbol{\Theta}\} +\lambda\left\|\mathbf{Z}\right\|_{1}\right\}\\
     & \ \textrm{s.t.} \ \boldsymbol{\Theta}-\mathbf{Z}=0,
\end{align*}
where for ease of notation we have dropped the dependence on $\omega$.

Then, following the work of \cite{boyd2011distributed}, we define the augmented Lagrangian of the above problem as 
\begin{equation*}
    L_{\tau}(\boldsymbol{\Theta}, \mathbf{Z}, \mathbf{U}) := \Tr\{\mathbf{\hat{S}}, \boldsymbol{\Theta}\} - \log \operatorname{det}(\boldsymbol{\Theta})+\lambda||\mathbf{Z}||_{1} + \tau/2 ||\mathbf{\Theta}-\mathbf{Z} + \boldsymbol{U}||^2_F, 
\end{equation*}
where $||\cdot||_F$ denotes the Frobenius norm.

If we begin with arbitrary initialisations for $\boldsymbol{\Theta}^{(0)}, \mathbf{Z}^0$ and $\mathbf{U}^{(0)}$, then the (scaled) ADMM method has the following update rules at the $(j+1)$ iteration 
\begin{align}
    \boldsymbol{\Theta}^{(j+1)} &= \arg \min _{\boldsymbol{\Theta}}  \left\{ \Tr\{\mathbf{\hat{S}}, \boldsymbol{\Theta}\} - \log \operatorname{det}(\boldsymbol{\Theta})+ \tau/2 ||\mathbf{\Theta}-\mathbf{Z}^{(j)} + \boldsymbol{U}^{(j)}||^2_F \right\} \label{theta_min}\\
    \mathbf{Z}^{(j+1)} &= \arg \min _{\mathbf{Z}} \left\{\lambda||\mathbf{Z}||_{1} + \tau/2 ||\mathbf{\Theta}^{(j+1)}-\mathbf{Z} + \boldsymbol{U}^{(j)}||^2_F \right\} \label{zmin} \\
    \mathbf{U}^{(j+1)} &= \mathbf{U}^{(j)} + \boldsymbol{\Theta}^{(j+1)} - \mathbf{Z}^{(j+1)}.\nonumber
\end{align}

One can obtain closed form updates for \eqref{theta_min} and \eqref{zmin}.
Consider first equation \eqref{theta_min}. 
If we denote the eigenvalue decomposition of the matrix $\tau(\mathbf{Z}^j-\mathbf{U}^j)-\hat{\mathbf{S}}$ by $\mathbf{Q}\mathbf{C}\mathbf{Q}^H$, where $\mathbf{C}={\textrm{{\textbf{diag}}}}(c_1, \dots, c_p)$. Then the ADMM update rule \eqref{theta_min} is reduced to the following form 
\begin{equation*}
    \boldsymbol{\Theta}^{(j+1)} = \mathbf{Q} \boldsymbol{\tilde{C}}\mathbf{Q}^H 
\end{equation*}
with diagonal matrix $\boldsymbol{\tilde{C}}$ whose $r^{th}$ diagonal element is $\tilde{c}_r=\frac{c_r+\sqrt{c_r^2+4\tau}}{2\tau}.$

The $Z$-minimisation step can be simplified using element-wise block soft thresholding. The ADMM update rule \eqref{zmin} becomes
\begin{equation*}
    Z^{(j+1)}_{qr} := \mathcal{S}_{\lambda/\tau}(\Theta_{qr}^{(j+1)}+U_{qr}^j),
\end{equation*} where $\mathcal{S}$ is the block-thresholding operator defined as 
\begin{equation*}
    S_\kappa(W_{qr}):=(1-\kappa/|W_{qr}|)_+W_{qr}, 
\end{equation*}
where $|W_{qr}| = \sqrt{\textrm{Re}(W_{qr})^2 + \textrm{Im}(W_{qr})^2}$, $a_+:=\max\{a,0\}$ and $\mathcal{S}_\kappa(0)=0.$

We use standard ADMM stopping criteria based on the primal and dual residuals 
which converge to zero as the ADMM algorithm proceeds. 

\section{ Supplementary Material for Section 4}
\label{app:Hawkes}
In this section, we provide additional information regarding the multivariate Hawkes process used in the synthetic experiments described in Section 4 of the main paper.

\subsection{The Multivariate Hawkes Process}
\begin{definition}
\cite{hawkes1971spectra} introduced the multivariate or ``mutually exciting" Hawkes process in which the conditional intensity function at time $t$ is linearly dependent on the entire history of the process up to time $t$. That is, if we consider the increment $dN_q(t)$ over the time interval $(t, t+dt)$ then
\begin{equation*}
P\{dN_q(t) = 1 | \mathbf{N}(s) \;,\; s\leq t\} = \Lambda_q(t) dt + o(dt),
\end{equation*}
and the probability of more than one event is $o(dt).$

The stochastic intensity function of the marginal point-process $N_q(t)$ is defined as
\begin{equation*}
    \Lambda_q(t) = \nu_q + \sum_{s=1}^p \int_{-\infty}^t \gamma_{qr}(t-u)dN_r(u),
\end{equation*}
where $\nu_q>0$ and $\gamma_{qr}:(0,\infty)\rightarrow[0,\infty)$ are the background intensity and excitation functions respectively. Specifically, $\gamma_{qr}(t-u)$ adjusts the intensity function of the $q^{th}$ dimension at time $t$ to account for an event in the $r^{th}$ dimension at a previous time point $u$.
\end{definition}
Motivated by applications in neuroscience, we consider a multivariate Hawkes process with exponentially decaying intensity functions. The form of the conditional intensity function for each marginal process is therefore
\begin{equation*}
\label{mv_intensity}
    \Lambda_q(t) = \nu_q + \sum_{r=1}^p\sum_{i:t_i^r<t}\alpha_{qr}e^{-\beta_{qr}(t-t_i^r)},
\end{equation*}
where $\nu_q>0, \alpha_{qr}\geq0$, $\beta_{qr}\geq0 \ \textrm{for} \ q,r=1,\dots,p$ and $t_i^r$ denotes the time of the $i^{th}$ event in process $r$. 
In this formulation, $\alpha_{qr}$ determines the size of the jump in $\Lambda_q(t)$ caused by an event in $N_r(\cdot)$ and $\beta_{qr}$ determines the decay exhibited in $\Lambda_q(t)$ caused by an event in $N_r(\cdot)$. 

\begin{remark}
    An important requirement here is the assumption of stationarity of the multivariate Hawkes process, i.e., $\mathbb{E}\{\Lambda_q(t)\}=\lambda \ \textrm{for} \ r=1,\dots, p.$ As such, the results presented in Sections 2 and 3 for general multivariate point-processes also hold for the multivariate Hawkes process.
\end{remark}

\begin{remark}
The intensity function for a multivariate Hawkes process can be written in vector notation as in \cite{hawkes1971point} as
\begin{equation*}
\label{vector_rep_int}
    \mathbf{\Lambda}(t) = \boldsymbol{\nu} + \int_{-\infty}^t \boldsymbol{\gamma}(t-u)d\mathbf{N}(u),
\end{equation*}
where $\boldsymbol{\gamma}$ is a $p\times p$ matrix. Under the assumption of stationarity, we have that
\begin{equation*}
    \boldsymbol{\Lambda}=(\boldsymbol{I}-\mathbf{G}(0))^{-1}\boldsymbol{\nu}, 
\end{equation*}
provided this has positive elements. Furthermore, note that $ \mathbf{G}$ is the Fourier Transform of the excitation function, which in the exponential case is given by

\begin{equation*}
    \mathbf{G}(\omega)=\int_{-\infty}^{\infty}e^{-i\omega \tau} \boldsymbol{\alpha}e^{-\boldsymbol{\beta}(\tau)}d\tau = \frac{\boldsymbol{\alpha}}{\boldsymbol{\beta}+i \omega}.
\end{equation*}

In order for the process to be stationary it is required that $\xi(\boldsymbol{G}(0))<1$, where $\xi(\boldsymbol{G}(0))$ is the spectral radius of $\boldsymbol{G}(0)$ defined as
\begin{equation*}
    \xi(\boldsymbol{G}(0))=\max_{x \in \mathcal{H}(\boldsymbol{G}(0))}|x|
\end{equation*}
where $\mathcal{H}(\boldsymbol{G}(0))$ is the set of eigenvalues of $\boldsymbol{G}(0).$ Following the work of \cite{hawkes1971point}, we arrive at the following definition.
\end{remark}

\begin{definition}
The spectral density matrix of a stationary multivariate Hawkes process is given as
\begin{equation*}
    \mathbf{S}(\omega)=\frac{1}{2\pi}\{\mathbf{I}-\mathbf{G}(\omega)\}^{-1} \mathbf{D}\{\mathbf{I}-\mathbf{G}^T(-\omega)\}^{-1},
\end{equation*}
where $\mathbf{D}=\textrm{diag}(\boldsymbol{\Lambda}).$
    
\end{definition}

\subsection{Parameterisation of the Hawkes Process for Synthetic Experiments}

In all experimental settings, the background intensity $\nu$ of the multivariate Hawkes process is set to $0.2$ for all dimensions. Furthermore, to allow for an appropriate comparison among the scenarios, each considered setting is parameterised in such a way to ensure that the spectral radius of the Fourier Transform of the excitation function is approximately $0.83$. In addition, all entries in the decay matrix $\boldsymbol{\beta}$ $\in \mathbb{R}^{p \times p}$ were set to $0.86$ for both settings (a) and (b). 
The $3 \times 3$ excitation matrices for settings (a) and (b) are specified as
\begin{equation*}
    \boldsymbol{\alpha}^* = 
    \begin{pmatrix}
    0 & 0.60 & 0 \\
    0 & 0.40  & 0 \\ 0 & 0 & 0.40 
    \end{pmatrix} \quad \textrm{and} \quad 
     \boldsymbol{\alpha}^*=\begin{pmatrix}
    0.20 & 0.10 & 0.25 \\
    0.10 & 0.20  & 0.40 \\ 0.25 & 0.40 & 0.20 
    \end{pmatrix}.
\end{equation*}
This yields the following excitation matrices for models (a) and (b) when $p=12$
\setcounter{MaxMatrixCols}{12}
\small
\begin{equation*}
    \boldsymbol{\alpha} = \begin{pmatrix}
  0 & 0.60 & 0 & 0 & 0 & 0 & 0 & 0 & 0 & 0 & 0 & 0 \\ 
  0 & 0.40 & 0 & 0 & 0 & 0 & 0 & 0 & 0 & 0 & 0 & 0 \\ 
  0 & 0 & 0.40 & 0 & 0 & 0 & 0 & 0 & 0 & 0 & 0 & 0 \\ 
  0 & 0 & 0 & 0 & 0.60 & 0 & 0 & 0 & 0 & 0 & 0 & 0 \\ 
  0 & 0 & 0 & 0 & 0.40 & 0 & 0 & 0 & 0 & 0 & 0 & 0 \\ 
  0 & 0 & 0 & 0 & 0 & 0.40 & 0 & 0 & 0 & 0 & 0 & 0 \\ 
  0 & 0 & 0 & 0 & 0 & 0 & 0 & 0.60 & 0 & 0 & 0 & 0 \\ 
  0 & 0 & 0 & 0 & 0 & 0 & 0 & 0.40 & 0 & 0 & 0 & 0 \\ 
  0 & 0 & 0 & 0 & 0 & 0 & 0 & 0 & 0.40 & 0 & 0 & 0 \\ 
  0 & 0 & 0 & 0 & 0 & 0 & 0 & 0 & 0 & 0 & 0.60 & 0 \\ 
  0 & 0 & 0 & 0 & 0 & 0 & 0 & 0 & 0 & 0 & 0.40 & 0 \\ 
  0 & 0 & 0 & 0 & 0 & 0 & 0 & 0 & 0 & 0 & 0 & 0.40 \\ 
   \end{pmatrix}
\end{equation*}
and 
\setcounter{MaxMatrixCols}{12}
\small
\begin{equation*}
    \boldsymbol{\alpha} = \begin{pmatrix}
  0.20 & 0.10 & 0.25 & 0 & 0 & 0 & 0 & 0 & 0 & 0 & 0 & 0 \\ 
  0.10 & 0.20 & 0.40 & 0 & 0 & 0 & 0 & 0 & 0 & 0 & 0 & 0 \\ 
  0.25 & 0.40 & 0.20 & 0 & 0 & 0 & 0 & 0 & 0 & 0 & 0 & 0 \\ 
  0 & 0 & 0 & 0.20 & 0.10 & 0.25 & 0 & 0 & 0 & 0 & 0 & 0 \\ 
  0 & 0 & 0 & 0.10 & 0.20 & 0.40 & 0 & 0 & 0 & 0 & 0 & 0 \\ 
  0 & 0 & 0 & 0.25 & 0.40 & 0.20 & 0 & 0 & 0 & 0 & 0 & 0 \\ 
  0 & 0 & 0 & 0 & 0 & 0 & 0.20 & 0.10 & 0.25 & 0 & 0 & 0 \\ 
  0 & 0 & 0 & 0 & 0 & 0 & 0.10 & 0.20 & 0.40 & 0 & 0 & 0 \\ 
  0 & 0 & 0 & 0 & 0 & 0 & 0.25 & 0.40 & 0.20 & 0 & 0 & 0 \\ 
  0 & 0 & 0 & 0 & 0 & 0 & 0 & 0 & 0 & 0.20 & 0.10 & 0.25 \\ 
  0 & 0 & 0 & 0 & 0 & 0 & 0 & 0 & 0 & 0.10 & 0.20 & 0.40 \\ 
  0 & 0 & 0 & 0 & 0 & 0 & 0 & 0 & 0 & 0.25 & 0.40 & 0.20 \\ 
   \end{pmatrix}.
\end{equation*}

 For model (c) we consider the following arbitrarily sparse excitation matrix for $p=12$
\setcounter{MaxMatrixCols}{12}
\begin{equation*}
    \boldsymbol{\alpha} =
  \begin{pmatrix}
  0 & 0 & 0.60 & 0 & 0 & 0 & 0 & 0 & 0 & 0 & 0 & 0 \\ 
  0 & 0 & 0 & 0 & 0 & 0 & 0 & 0 & 0 & 0.50 & 0 & 0 \\ 
  0.60 & 0 & 0 & 0.80 & 0 & 0 & 0 & 0 & 0 & 0 & 0 & 0 \\ 
  0 & 0 & 0.80 & 0 & 0 & 0 & 0 & 0 & 0 & 0 & 0 & 0 \\ 
  0 & 0 & 0 & 0 & 0 & 0 & 0 & 0 & 0 & 0 & 0 & 0 \\ 
  0 & 0 & 0 & 0 & 0 & 0 & 0 & 0 & 0 & 0 & 0 & 0 \\ 
  0 & 0 & 0 & 0 & 0 & 0 & 0 & 0 & 0 & 0 & 0 & 0 \\ 
  0 & 0 & 0 & 0 & 0 & 0 & 0 & 0 & 0 & 0 & 0 & 0 \\ 
  0 & 0 & 0 & 0 & 0 & 0 & 0 & 0 & 0 & 0 & 0 & 0 \\ 
  0 & 0.50 & 0 & 0 & 0 & 0 & 0 & 0 & 0 & 0 & 0 & 0 \\ 
  0 & 0 & 0 & 0 & 0 & 0 & 0 & 0 & 0 & 0 & 0 & 0 \\ 
  0 & 0 & 0 & 0 & 0 & 0 & 0 & 0 & 0 & 0 & 0 & 0 \\ 
   \end{pmatrix}.
\end{equation*}
For higher dimensions, each of the above $\boldsymbol{\alpha}$ matrices are used to construct $p$-dimensional processes with block diagonal structure. For example, for dimension $p=48$, the excitation matrix contains $4$ blocks of the above $12\times12$ matrix where there are no interactions outwith these blocks. 

\section{Supplementary Material for Section 5}

\begin{figure}[t]
\centering
\begin{subfigure}[]{0.3\textwidth}
\centering
   \includegraphics[width=0.95\textwidth]{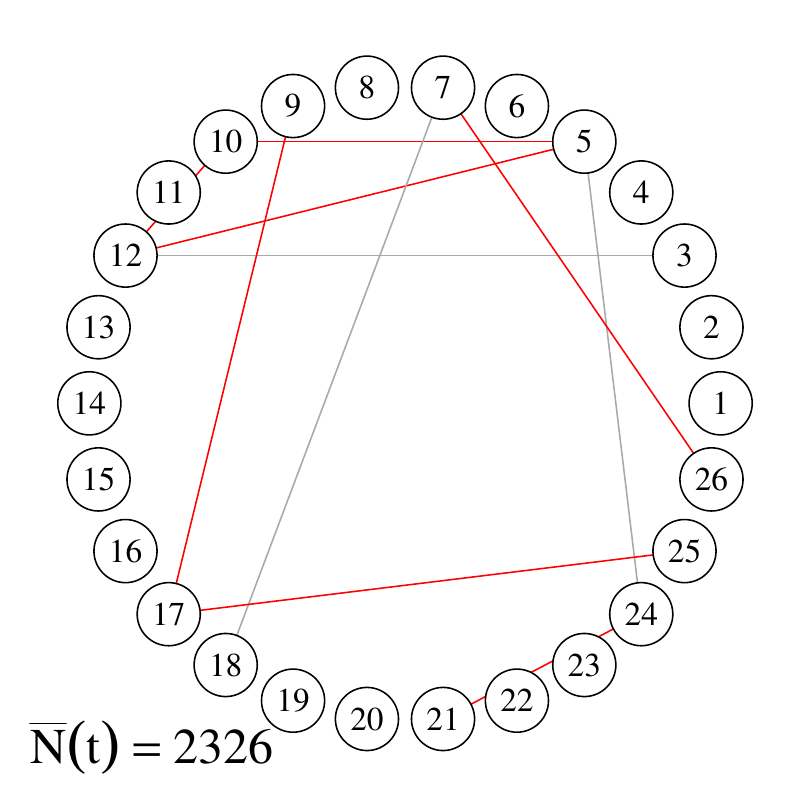}
\end{subfigure}
\begin{subfigure}[]{0.3\textwidth}
   \centering
    \includegraphics[width=0.95\textwidth]{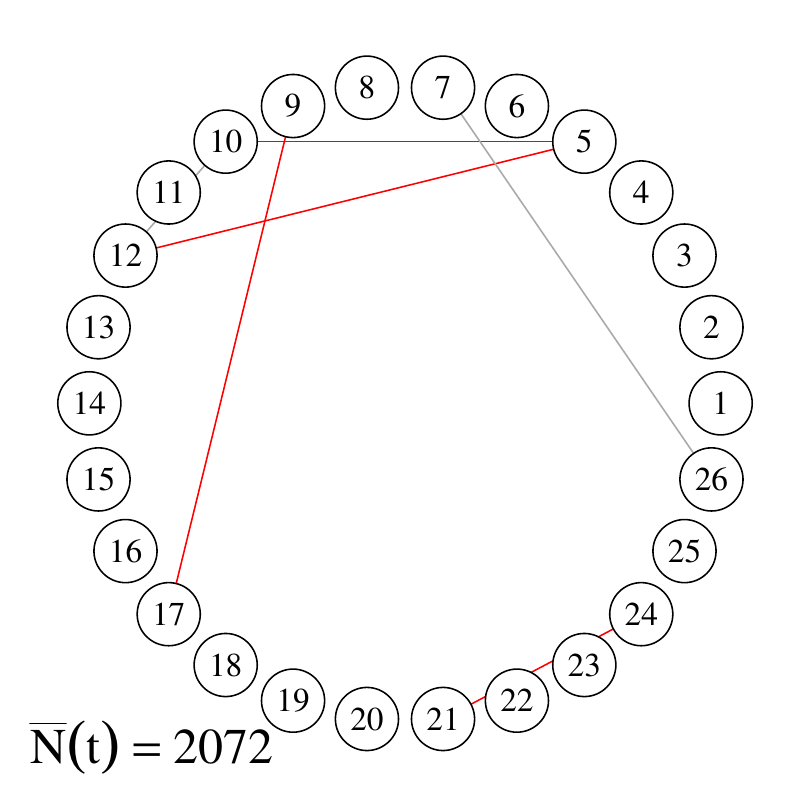}
\end{subfigure}
\begin{subfigure}[]{0.3\textwidth}
\centering
   \includegraphics[width=0.95\textwidth]{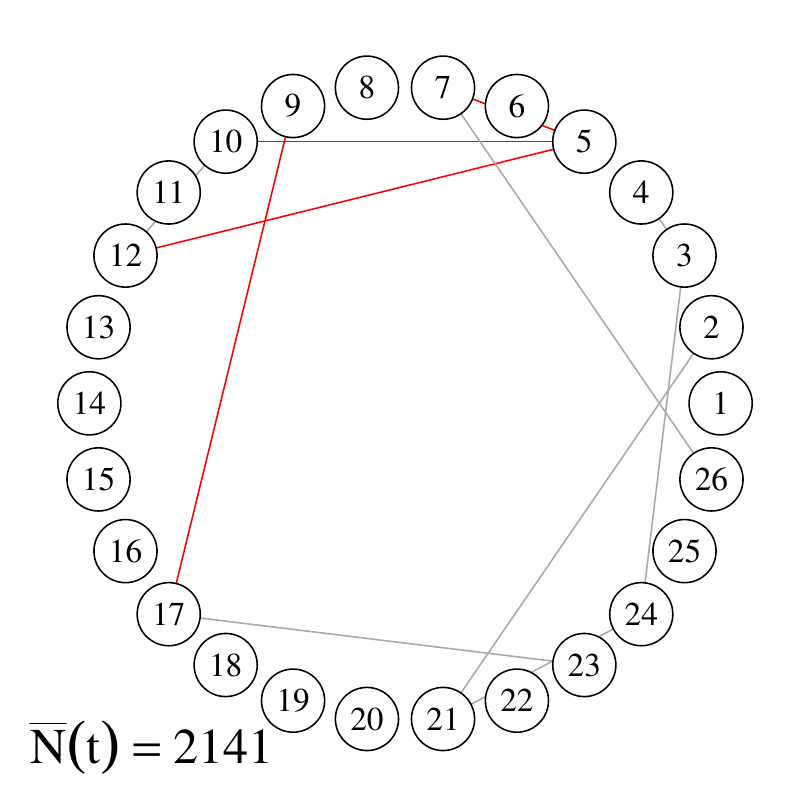}
\end{subfigure}

\begin{subfigure}[]{0.3\textwidth}
\centering
   \includegraphics[width=0.95\textwidth]{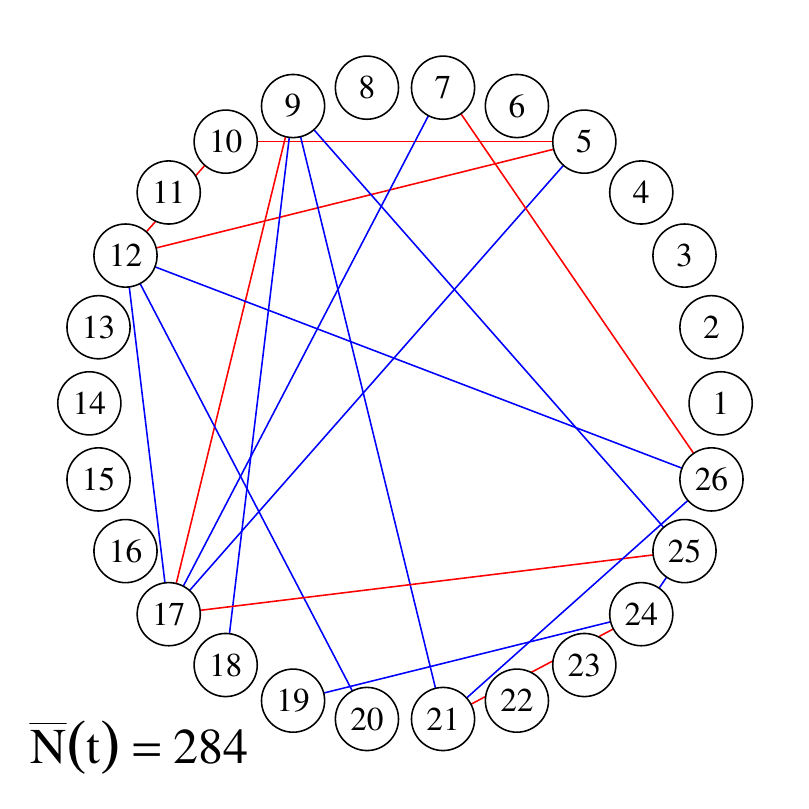}
   \caption{}
\end{subfigure}
\begin{subfigure}[]{0.3\textwidth}
   \centering
    \includegraphics[width=0.95\textwidth]{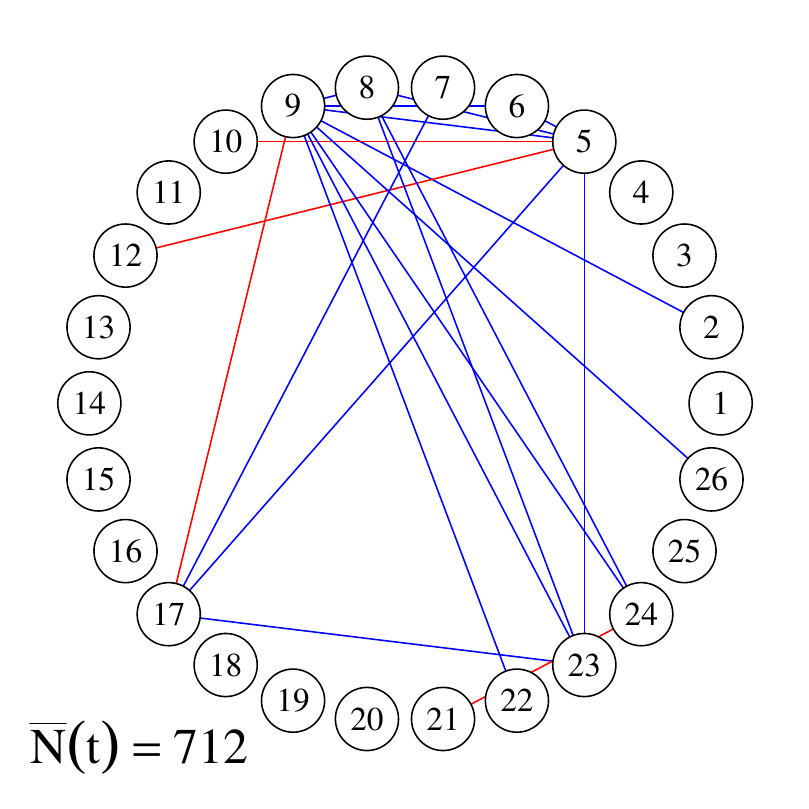}
     \caption{}
\end{subfigure}
\begin{subfigure}[]{0.3\textwidth}
\centering
   \includegraphics[width=0.95\textwidth]{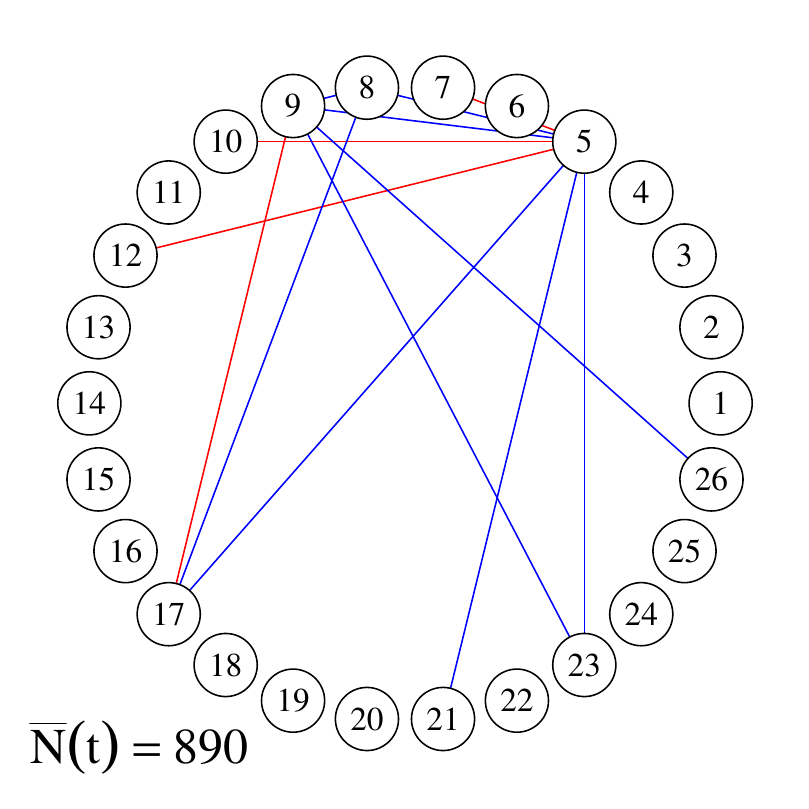}
    \caption{}
\end{subfigure}
\caption{\small{Estimated networks of neural interactions on the delta band using the spike train data from \cite{bolding2018recurrent}. Each column shows the estimated networks for the laser off (top) and laser on (bottom) condition for laser stimuli (a) $0mw/mm^2$, (b) $10mw/mm^2$ and (c) $50mw/mm^2$. Common edges between the on and off conditions for each level of intensity are shown in red, and edges unique to the laser on condition are shown in blue. All other edges are shown in grey.}}
\label{plt:networks2}
\end{figure}

In this section, we present additional figures related to Section 5 of the main paper. Figure \ref{plt:networks2} shows the estimated partial coherence graphs for the theta band obtained under the `laser on' and `laser off' conditions, i.e., at frequencies in the range $(4,8)$Hz. As was the case for the delta band, in general there are more edges detected for the `laser on' condition ($20$ under $10 mW/mm^2$ and 13 under $50 mW/mm^2$ versus $6$ and $11$ under $10 mW/mm^2$ and 13 under $50 mW/mm^2$ for the `laser off' setting). 

\begin{table}[H]
    \centering
    \begin{tabular}{ccccc}
    \hline \hline
       Condition & Frequency Band  & $0 mW/mm^2$ &  $10 mW/mm^2$ &  $50 mW/mm^2$ \\
       \hline \hline 
       
         Laser On & Delta & 0.977 & 1.874 & 1.177 \\
         & Theta & 0.559 & 1.072 & 1.417 \\
        \\
           
         laser Off & Delta & 4.751 & 5.214 & 4.328 \\
         & Theta & 3.944 & 4.329 & 3.594 \\
         \hline \hline
    \end{tabular}
    \caption{ Regularisation parameter values selected via eBIC for each condition at each intensity.}
    \label{tab:reg_delta}
\end{table}

The estimated partial coherence graphs were obtained using the Pooled Glasso estimator. Each Pooled estimator was tuned using the eBIC, whereby we search over a grid of $\lambda$ values and select the value which minimises the eBIC. We record the regularisation parameter selected under each considered scenario in Table \ref{tab:reg_delta} 

Under each scenario, the optimisation problem was solved via ADMM. Figure \ref{plt:checks_delta} shows the stopping criteria for the ADMM algorithm for the considered frequency bands under both the `laser on' and `laser off' conditions. Since the stopping criteria converges to 0 as the number of iterations increases, we can be confident that the algorithm has converged. For completeness we also show the eBIC curves produced as part of the parameter tuning process. 

\begin{figure}[t]
\centering
\begin{subfigure}[]{\textwidth}
\centering
   \includegraphics[width=0.75\textwidth]{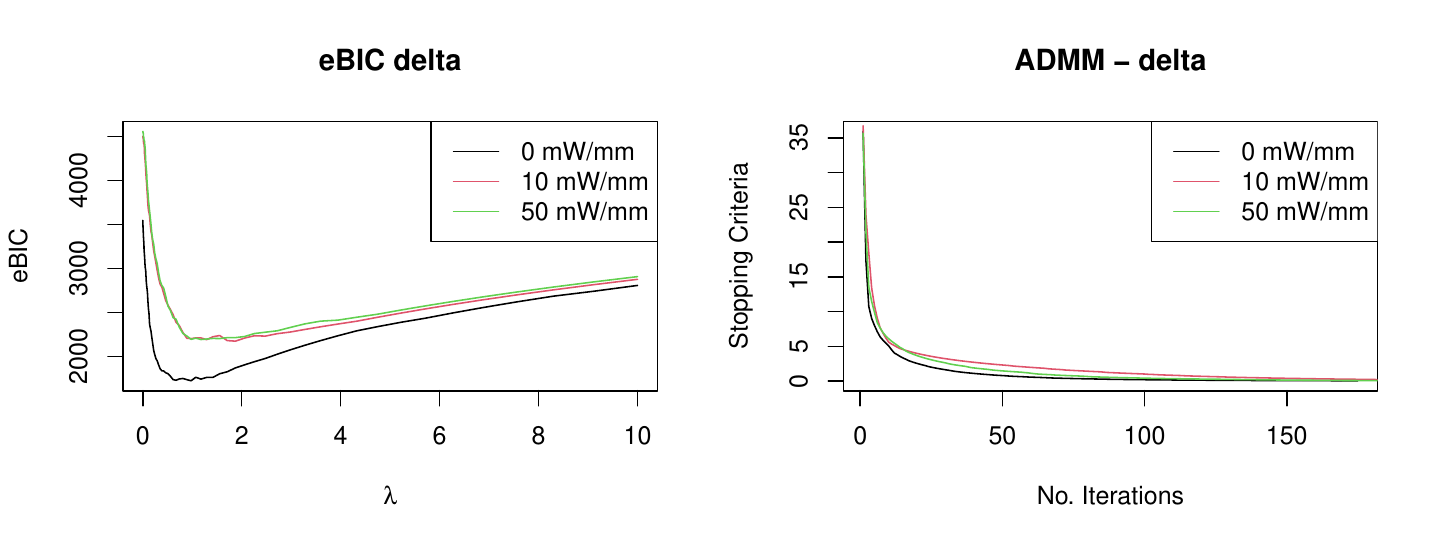}
   \caption{}
\end{subfigure}
\begin{subfigure}[]{\textwidth}
   \centering
   \includegraphics[width=0.75\textwidth]{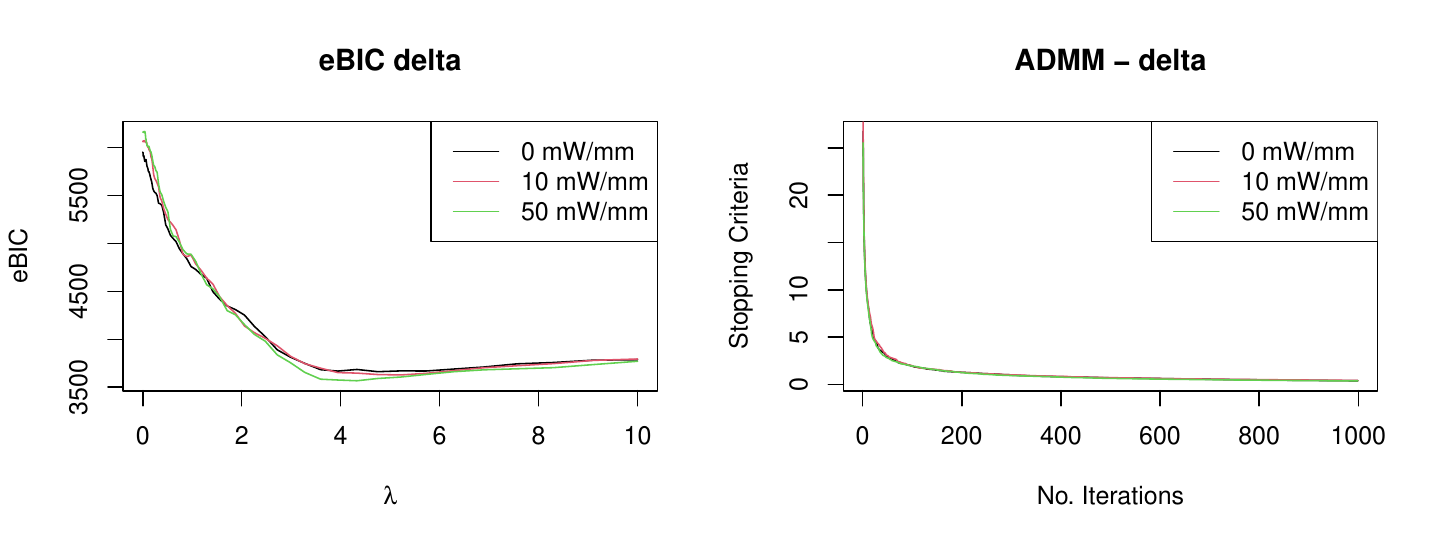}
   \caption{}
\end{subfigure}
\begin{subfigure}[]{\textwidth}
\centering
   \includegraphics[width=0.75\textwidth]{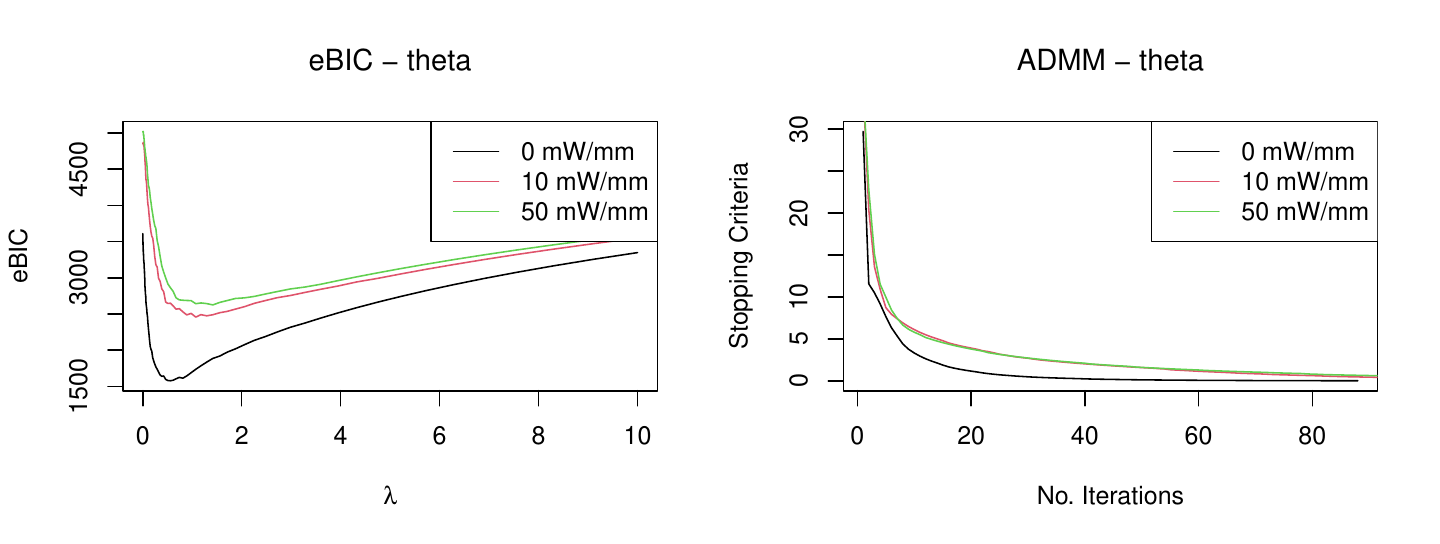}
   \caption{}
\end{subfigure}
\begin{subfigure}[]{\textwidth}
   \centering
   \includegraphics[width=0.75\textwidth]{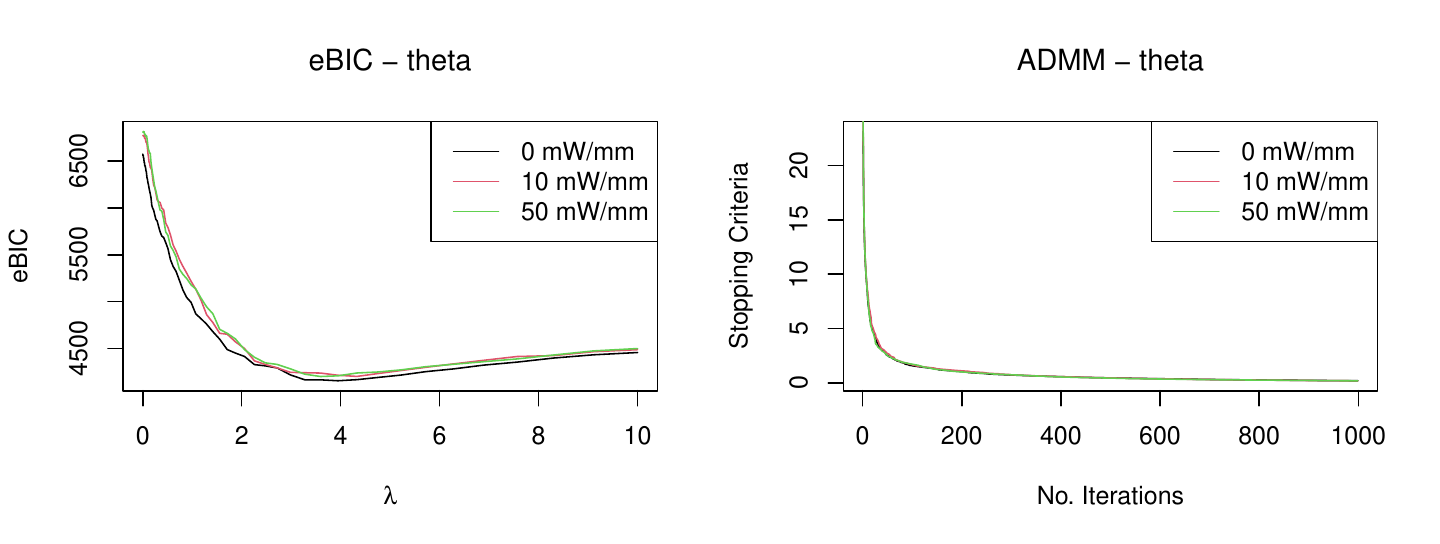}
   \caption{}
\end{subfigure}
\caption{\small{Additional checks for data analysis.}}
\label{plt:checks_delta}
\end{figure}


\end{document}